%% file: main.tex
\documentclass[runningheads]{llncs}
\usepackage[T1]{fontenc}
\usepackage{graphicx}
\usepackage{booktabs}   
\usepackage{subcaption} 
\usepackage{bussproofs}
\usepackage{hyperref} 
\usepackage{xstring}
\usepackage{mathpartir}
\usepackage{bcprules}
\usepackage{float}
\usepackage{amsmath}

\usepackage{amssymb}
\usepackage{stmaryrd}
\usepackage{xcolor}

\def\fst{\pi_1}
\def\snd{\pi_2}
\def\Int{int}
\newcommand{\con}[4]{C_{#1}({#2},({#3},...,{#4}))}
\newcommand{\pair}[2]{(#1,#2)}
\newcommand{\tick}[2]{tick~#1~#2}
\newcommand{\tickl}[2]{lazy\_tick~#1~#2}
\newcommand{\prodT}[2]{#1\times#2}
\newcommand{\app}[2]{#1\ #2}
\newcommand{\lam}[3]{\lambda #1.\,#2#3}
\newcommand{\llam}[3]{\lambda #1:#3.\,#2}
\newcommand{\Lam}[3]{\Lambda #1:#3.\,#2}
\newcommand{\fix}[3]{fix~#1:#2.\,#3}
\newcommand{\Let}[3]{let~#1=#2~in~#3}
\newcommand{\des}[3]{matd(e_0,\overrightarrow{#1(x_0,(x_1,...,x_{#2}).#3)}}
\newcommand{\destructor}[4]{matd(#1,\overrightarrow{#2(x_0,(x_1,...,x_{#3}).#4)}}
\newcommand{\arrow}[6]{
(\ifthenelse{\equal{#1}{}}{#3}{#1:#3}
\rightarrow \ifthenelse{\equal{#6}{}}{#5}{#6:#5})_{[#2\rightarrow #4]}}
\def\ind{{ind}\overrightarrow{(C,(T,m))}}
\newcommand{\poly}[3]{\forall #1:#2.\,#3}
\newcommand{\ectx}[2]{#1~|~#2}

\newcommand{\mctx}[3]{#1,\, #2:#3 }

\newcommand{\fctx}[1]{\bigl[#1\bigr]}
\newcommand{\op}[2]{\mathbf{op_{#1}}(\overrightarrow{#2})}
\newcommand{\mini}[2]{\mathbf{min}_{#1}({#2})}
\newcommand{\maxi}[2]{\mathbf{max}_{#1}({#2})}
\newcommand\doubleplus{+\kern-1.3ex+\kern0.8ex}

\RequirePackage[prologue]{xcolor}
\definecolor[named]{ACMBlue}{cmyk}{1,0.1,0,0.1}
\definecolor[named]{ACMYellow}{cmyk}{0,0.16,1,0}
\definecolor[named]{ACMOrange}{cmyk}{0,0.42,1,0.01}
\definecolor[named]{ACMRed}{cmyk}{0,0.90,0.86,0}
\definecolor[named]{ACMLightBlue}{cmyk}{0.49,0.01,0,0}
\definecolor[named]{ACMGreen}{cmyk}{0.20,0,1,0.19}
\definecolor[named]{ACMPurple}{cmyk}{0.55,1,0,0.15}
\definecolor[named]{ACMDarkBlue}{cmyk}{1,0.58,0,0.21}

\newcommand\ignore[1]{}
\newcommand\tjudge[7]{#1\,|\,#2\,|\,#3\,\vdash\,#4:\,\bigl[#5\bigr]_{#7} #6}
\newcommand\ajudge[7]{#1\,|\,#2\,|\,#3\,\vdash_a\,#4:\,\bigl[#5\bigr]_{#7} #6}

\newcommand\pjudge[5]{#1\,|\,#2\,\vdash\,#3#4#5}

\usepackage{mathtools}
\usepackage{listings}
\usepackage{cleveref}
\crefname{theorem}{Thm.}{Thms.}
\crefname{lemma}{Lem.}{Lemmas}
\crefname{corollary}{Cor.}{Cors.}
\crefname{figure}{Fig.}{Figs.}
\crefname{definition}{Defn.}{Defns.}
\crefname{table}{Tab.}{Tabs.}
\crefformat{section}{\S#2#1#3}
\crefmultiformat{section}{\S#2#1#3}{ and~\S#2#1#3}{, \S#2#1#3}{ and~\S#2#1#3}
\crefname{example}{Ex.}{Exs.}
\crefname{item}{item}{items}
\crefname{footnote}{footnote}{footnotes}
\crefname{observation}{Obs.}{Obs.}
\crefname{remark}{Remark}{Remarks}
\crefname{proposition}{Prop.}{Props.}
\crefname{equation}{Eqn.}{Eqns.}
\crefname{counterexample}{Counterexample}{Counterexamples}
\crefname{property}{Property}{Properties}
\crefname{algorithm}{Algorithm}{Algorithms}

\usepackage{shortcuts}
\usepackage{turnstile}

\newcommand{\mi}[1]{\mathit{#1}}
\newcommand{\ms}[1]{\mathsf{#1}}

\newcommand{\listop}{\mi{List}}
\newcommand{\tlist}[1]{\listop(#1)}
\newcommand{\treeop}{\mi{Tree}}
\newcommand{\ttree}[1]{\treeop(#1)}
\newcommand{\tint}{\mi{int}}
\newcommand{\tunit}{\mi{unit}}
\newcommand{\lang}{\ensuremath{\lambda_\ms{amor}^\ms{na}}}
\newcommand{\vdashpq}[2]{\sststile{#2}{#1}}

\begin{document}
\title{Dependently-Typed AARA: A Non-Affine Approach for Resource Analysis of Higher-Order Programs}
\titlerunning{Dependently-Typed AARA}
%
\author{Han Xu\inst{1}\orcidID{0000-0002-2548-6866} \and
Di Wang\inst{2}\orcidID{0000-0002-2418-7987}\thanks{Corresponding Author} }
%
%
\institute{Princeton University, Princeton NJ 08544, USA \and
Key Lab of HCST (PKU), MOE; SCS, Peking University, China}
\maketitle              
\begin{abstract}

Static resource analysis determines the resource consumption (e.g., time complexity) of a program without executing it.
Among the numerous existing approaches for resource analysis, affine type systems have been one dominant approach.
However, these affine type systems fall short of deriving precise resource behavior of higher-order programs, particularly in cases that involve partial applications.

This article presents \lang{}, a non-affine AARA-style dependent type system for resource reasoning about higher-order functional programs.
The key observation is that the main issue in previous approaches comes from (i) the close coupling of types and resources, and (ii) the conflict between affine and higher-order typing mechanisms.
To derive precise resource behavior of higher-order functions, \lang{} decouples resources from types and follows a non-affine typing mechanism.
The non-affine type system of \lang{} achieves this by using dependent types, which allow expressing type-level potential functions separate from ordinary types.
This article formalizes \lang{}'s syntax and semantics, and proves its soundness, which guarantees the correctness of resource bounds.
Several challenging classic and higher-order examples are presented to demonstrate the expressiveness and compositionality of \lang{}'s reasoning capability.
This article also includes an algorithmic variant of \lang{}'s type system and a discussion of the automation of type checking and inference for \lang{}.
\end{abstract}
\input{introduction}
\input{overview}

\input{details}

\input{evaluation}

\input{discussion}
\input{relatedWork}
\input{conclusion}

\subsubsection*{Data-Availability Statement}
The Appendix is available at \url{https://arxiv.org/abs/2601.12943}.

\subsubsection*{Acknowledgements.} This work was sponsored by National Key R\&D Program of China, Grant No. 2024YFE0204100.

%
%
%
\bibliographystyle{splncs04}
\bibliography{reference.bib,db.bib}
%


\newpage
\appendix
\input{appendix}

\end{document}

%% file: introduction.tex
\section{Introduction}
\label{Se:Introduction}







%
%
%

\paragraph{\textbf{Resource Analysis}}
Resource consumption (e.g., time, space, and complexity) of programs has always been at the heart of computer science, with programming language research being no exception. 
In this article, we focus on the static verification of upper bounds on the resource
consumption of a higher-order functional program.
For example, below shows the type and implementation of a \emph{curried} \texttt{append}
function for integer lists, where we wrap recursive calls in a special construct
\texttt{\color{ACMDarkBlue}tick\;1} to indicate that we specify a resource model that
counts the number of recursive calls.
\begin{lstlisting}
append :: $\tlist{\tint}$ -> $\tlist{\tint}$ -> $\tlist{\tint}$
append = fun x. fun y. case x of nil => y | cons(x0,x1) => cons(x0, ?\color{ACMDarkBlue}tick 1? (append x1 y))
\end{lstlisting}
A resource analysis of the \texttt{append} function should yield that the number
of recursive calls is upper-bounded by the length of the first argument.
There have been tons of techniques for verifying or automatically inferring
such an upper bound; in particular, there have been many \emph{type systems} for
resource analysis
\cite{ma03static,knoth20liquid,PLDI:KWP19,ESOP:HH10,POPL:HDW17,POPL:HAH11,POPL:JHL10,FM:JLH09,POPL:CBG17,POPL:CW00,LICS:LG11,ICFP:DLR15,POPL:HVH20,POPL:KML20,vi21unify,POPL:NSG22,POPL:GNS24,FOSSACS:KH20,ICFP:KH21,LICS:GKH23,POPL:Danielsson08,OOPSLA:WWC17,ICFP:AL17,APLAS:HH10,ESOP:HS15,SP:NDF17}.

Among those type systems, \emph{automatic amortized resource analysis} (AARA) has been
a fundamental approach for deriving symbolic resource bounds, i.e., bounds that are given
by a function of the program's input.
AARA was initially studied by Hofmann and Jost~\cite{ma03static}, and later extended to various resource-analysis situations
with a recent survey provided by Hoffmann and Jost~\cite{JMSCS:HJ22}.
The high-level idea of AARA is to augment the type system with \emph{resource annotations},
which specify \emph{potential functions} carried on data structures.
For example, a value of type $\tint^1$ carries one unit of potential, which can be
used to pay for one unit of resource consumption (i.e., \texttt{\color{ACMDarkBlue}tick\;1}).
A more interesting example is the annotated list type $\tlist{\tint^1}$, the value of
which carries one unit of potential per list element, thus the whole potential equals
exactly the list length.
An AARA type system would derive an \emph{uncurried} type annotation for \texttt{append}:
\begin{lstlisting}
append :: $\tlist{\tint^1}$ $\times$ $\tlist{\tint^0}$ -> $\tlist{\tint^0}$
\end{lstlisting}
The type suggests that \texttt{append} consumes up to one unit of resource per element in the first list.

\paragraph{\textbf{Closure and Resource}}
With the \emph{uncurried} annotated type shown above, can we say that the time complexity (in terms of
the number of recursive calls) of a curried \texttt{append} is $O(n)$, if $n$ is the length
of the first argument?
Unfortunately, the fact that \texttt{append} is \emph{curried} makes things complicated, because
it is obvious that partially applying \texttt{append} to \emph{one} list, e.g., $\texttt{append}\;\ell_1$,
would consume no resource.
Moreover, the partial application would create a closure that captures $\ell_1$.
When applied to another list $\ell_2$, the closure's time complexity is linear---\emph{not} in the length of the closure's argument $\ell_2$---but in the length of the captured $\ell_1$.
In other words, precise resource analysis of higher-order functions
requires that the type system track fine-grained information about closures,
especially those that can capture data structures with potentials.
We call such closures \emph{potential-capturing closures}.

In this article, we focus on supporting AARA-style resource analysis with potential-carrying closures.
To demonstrate the technical challenges, we need to first zoom in on how AARA works.
Most AARA-style type systems rely on \emph{linear} or \emph{affine} typing~\cite{kn:Walker02}, whose (informal)
intension is to disallow multiple usages of the same program variable.
This makes sense because the data structure referred to by a program variable may carry potentials,
while duplicated uses of a variable cause duplications of its potentials.
As a consequence, every function type in an affine type system would be annotated explicitly
with a \emph{multiplicity}, i.e., $A \to^m B$, where $m$ is an upper bound on the times the function
could be applied.
For example, the \emph{curried} \texttt{append} function is actually assigned with the following more verbose type annotation:
\begin{lstlisting}
append :: $\tlist{\tint^1}$ ->$^\infty$ $\tlist{\tint^0}$ ->$^1$ $\tlist{\tint^0}$
\end{lstlisting}
The first arrow is annotated with the multiplicity $\infty$, which indicates
that the $\texttt{append}$ function can be applied an arbitrary number of times.
One can justify this by observing that function \texttt{append} can generate arbitrarily many closures, as long as enough potentials are provided with its first argument $\ell_1$.
The second arrow is annotated with the multiplicity $1$, which indicates that the closure
obtained by $\texttt{append}\;\ell_1$ can be applied at most once if $\ell_1$ is of type
$\tlist{\tint^1}$.
One can justify this by observing that the closure captures $\ell_1$, whose potential equals
the length of $\ell_1$ (denoted by $|\ell_1|$), but applying the closure to $\ell_2$ would consume
$|\ell_1|$ units of resource and use up all the captured potential, thus the closure should not be applied more than once.

However, the verbose type annotation for \texttt{append} does not characterize its precise
resource behavior, e.g., partially applying \texttt{append} to one list $\ell_1$ does not
consume any resources, but the type already requires $|\ell_1|$ units.
In addition, the type system determines the number of uses by its type annotation instead of by its context. Thus, a re-analysis of typing is triggered every time we put \texttt{append} into a different program context, which is rather non-compositional.
As an example, the verbose type annotation shown above is \emph{not} suitable for checking
\begin{lstlisting}
let app_par = append $\ell_1$ in (app_par $\ell_2$, app_par $\ell_3$)
\end{lstlisting}
because the program would use $\texttt{append}\;\ell_1$ twice. We need a different annotation like
\begin{lstlisting}
append :: $\tlist{\tint^2}$ ->$^\infty$ $\tlist{\tint^0}$ ->$^2$ $\tlist{\tint^0}$
\end{lstlisting}
to type-check the program.

This \texttt{append} problem is also pointed out by Scherer and Hoffmann~\cite{LPAR:SH13} but it remains incompletely solved. In their work, they studied type-based analysis of closures and proposed \emph{open closure types},
i.e., function types annotated with type contexts to track data-flow properties
of captured variables.
The \texttt{append} is then possible to track potentials with an annotated type like
\begin{lstlisting}
append :: $[](x:\tlist{\tint^0})$ ->$^\infty$ $[x:\tlist{\tint^1}](y:\tlist{\tint^0})$ ->$^1$ $\tlist{\tint^0}$
\end{lstlisting}
where $[\Gamma](x:A) \to^m B$ uses a resource-annotated type context $\Gamma$ to describe
the behavior of the closure application.
However, even if we adapted open closure types to AARA, the resulting type system
would still be affine and face the issue of having to determine multiplicities upon function definitions.

\paragraph{\textbf{Challenge: Go beyond Affine Typing}}
Our key observation is the following:
\begin{center}
\framebox{
\begin{minipage}{0.75\textwidth}
\emph{Existing AARA-style type systems require affine typing because they tightly couple datatypes with potentials.}
\end{minipage}
}
\end{center}
For example, the type $\tlist{\tint^1}$ shown above is a resource-annotated type that
couples an ordinary list type $\tlist{\tint}$ with a potential function that equals
the length of the list.
More advanced examples include polynomial-potential annotations~\cite{ESOP:HH10}, e.g.,
$\listop^{(q_0,q_1,q_2)}(\tint)$ couples $\tlist{\tint}$ with the potential function
$\lambda\ell.\, q_0 \binom{|\ell|}{0} + q_1 \binom{|\ell|}{1} + q_2 \binom{|\ell|}{2}$.
Resource-annotated types play an important role in the \emph{automation} part of AARA:
a type system can design a specific numerical space of resource potentials (e.g., non-negative integers or rational numbers)
and then reduce the type inference to solving a system of constraints by the annotations (e.g., linear programming).

In this article, we tackle the challenge of going beyond affine typing by
completely \emph{decoupling} potentials from datatypes.
Intuitively, because all datatypes become potential-free, they do not need
to be affine; thus, we fall back to a standard type-system design.
On the other hand, we need a mechanism to \emph{explicitly} represent potential
functions at the type level.
Taking inspiration from open closure types, we introduce a dependent-arrow-like
notation $[f_1]_x T_1 \to [f_2]_y T_2$, where $x$ and $y$ bind the argument
and result respectively, $T_1$ and $T_2$ are ordinary types, and $f_1,f_2$ are
potential functions that can reference free variables in the type context of the function definition.
For example, the \texttt{append} function has the following type in our type system:
\begin{lstlisting}
append :: $\arrow{x}{0}{\tlist{\tint}}{0}{\arrow{y}{\mi{length}(x)}{\tlist{\tint}}{0}{\tlist{\tint}}{}}{}$
\end{lstlisting}
Note that we hide unnecessary binders for brevity.
The type above should be interpreted as follows:
the \texttt{append} function takes an argument $x$ and it does not require any potential
to create a closure which captures $x$.
The obtained closure then takes an argument $y$ and it requires $\mi{length}(x)$ units
of potential to return a list without potential.
This type annotation has two benefits:
(i) it characterizes the precise resource behavior of \texttt{append}, including the behavior
of its partial application, and
(ii) it does not need to determine any multiplicity and it is not an affine type at all.
As a tradeoff, our type system has to support type-level computations such as $\mathit{length}(x)$
that can use program variables.
Thus, our type system is a dependently-typed variant of AARA.

\paragraph{\textbf{Challenge: Support Lightweight Dependent Typing}}
There have been several dependent-type-based approaches for resource analysis or verification.
Wang et al.~\cite{OOPSLA:WWC17} developed a refinement type system for complexity analysis, focusing on constraint-based analysis instead of supporting potential-based analysis like AARA.
Knoth et al.\cite{PLDI:KWP19,knoth20liquid} developed liquid-style type systems that integrate refinement typing with AARA,
while Rajani et al.~\cite{vi21unify} proposed a dependently-typed calculus for amortized resource analysis;
however, those systems still require linear or affine typing for tracking potentials.
Niu et al.~\cite{POPL:NSG22} proposed a cost-aware logical framework that can be instantiated to build
different resource analyses, including amortized analysis, but it emphasizes its logical framework instead of giving a concrete resource-annotated dependent type system.
%

In this article, we focus on designing a lightweight AARA-style dependent type system
\emph{without} affine typing.
Our type system is lightweight in the sense that its dependent part describes \emph{only}
potential functions, e.g., numeric primitive recursions over inductive datatypes.
One major challenge in designing the type system is to ensure \emph{compositional} typing,
especially for dependent function applications.
For example, partially applying \texttt{append} to an expression $e$ of type $\tlist{\tint}$
with $\mathit{length}(e)$ units of potential would yield the following typing:
\begin{lstlisting}
append $e$ :: $\exists x:\tlist{\tint}.\, [\mathit{length}(x)]\arrow{y}{\mi{length}(x)}{\tlist{\tint}}{0}{\tlist{\tint}}{}$ 
\end{lstlisting}
where the existential indicates that $\texttt{append}\;e$ constructs a closure that captures
a list $x$ with $\mathit{length}(x)$ units of potential.
However, managing existentials would quickly complicate the typing, especially when there
are higher-order functions with different levels of existential quantifiers. But fortunately, global existential quantifiers can satisfy our need for potential analysis.

%
We here propose a lightweight solution that keeps a \emph{global} potential context $\Omega$ for those existentials along with
the usual type context $\Gamma$.
In some sense it is similar to the idea of prenex polymorphism, which places all universal quantifiers at
the outermost position.
In our setting, we consider placing all existential quantifiers, which account for potentials captured in
a closure, at the outermost position.
Our type judgements then take the form $\tjudge{\Omega}{\Gamma}{f_1}{e}{f_2}{T}{x}$,
where  $\Omega$ is our global potential context, $\Gamma$ is the usual type context,
$f_1$ and $f_2$ are potential functions, and $x$ is the local binder on $T$ which can be omitted for brevity.
For example, below are possible type judgements for \texttt{append} and $\texttt{append}\;e$:
\begin{align*}
\cdot \mid \cdot \mid 0  \vdash \texttt{append} : \arrow{x}{0}{\tlist{\tint}}{0}{\arrow{y}{\mi{length}(x)}{\tlist{\tint}}{0}{\tlist{\tint}}{}}{} \\
x:\tlist{\tint} \mid \cdot \mid 0  \vdash \texttt{append}\;e : [\mathit{length}(x)]\arrow{y}{\mi{length}(x)}{\tlist{\tint}}{0}{\tlist{\tint}}{}
\end{align*}
In this way, our type system sidesteps the issue of possible ``pollution'' of existentials
while still being effective enough for AARA-style reasoning.

\paragraph{\textbf{Contributions}}
This article makes the following three contributions:
\begin{itemize}
\item We propose a lightweight dependently-typed AARA for reasoning about the resource consumption of higher-order functional programs. Our type system is the first AARA-style type system that does not require linear/affine typing and can describe the precise behavior of curried functions.
\item We formalize our dependent type system and prove its soundness with respect to a cost-aware operational semantics.
\item We demonstrate the effectiveness of our type system on a suite of challenging examples. We also include a discussion towards automating our type system.
\end{itemize}

%% file: overview.tex
\section{Overview}
\label{sec:overview}

We call our new calculus \lang{}, where ``$\ms{amor}$'' and ``$\ms{na}$'' are short
for ``amortized'' and ``non-affine'' respectively.
In this section, we sketch the general idea of how \lang{} works and use a few motivating examples to demonstrate its difference from prior type systems for AARA-style resource analysis and verification.

\subsection{Prior Work: Automatic Amortized Resource Analysis}
\label{Se:BackgroundAARA}

Automatic Amortized Resource Analysis (AARA) is a technique first introduced by Hofmann and Jost~\cite{ma03static} as a type system for deriving linear worst-case bounds on the heap-space consumption of first-order functional programs with eager evaluation strategy.
As surveyed by Hoffmann and Jost~\cite{JMSCS:HJ22}, dozens of works extended AARA to support different resource metrics, evaluation strategies, resource bounds, and language features, making AARA a state-of-the-art methodology for resource analysis and verification.
At the core of AARA is the \emph{potential method} for amortized complexity analysis, which was proposed
by Tarjan~\cite{JADM:Tarjan85} to manually derive an upper bound on the resource consumption of a sequence of
operations.
To apply the potential method to analyze general programs, one needs to specify potential functions
that map program states to non-negative numbers, such that the potential at every possible program
state is sufficient to pay for the cost of the next state transition as well as the potential of the
next state.
AARA introduces type annotations to encode such potential functions;
thus, an AARA-style type system must statically verify the type annotations in a program
to ensure the correct application of the potential method.

Typical AARA typing judgements take the form $\Gamma \vdashpq{p}{q} e:A$ with $p, q \ge 0$,
which reads as:
under type context $\Gamma$ and with $p$ units of potential, the expression $e$ has the type $A$ with a return of $q$ units of potential.
As we discussed in \cref{Se:Introduction}, the type $A$ and those types in $\Gamma$ are \emph{resource-annotated},
e.g., $\tlist{\tint^1}$ which means a list that carries one unit of potential per list element.
It has been shown that the close coupling of potential functions and data structures renders AARA effective and automatable~\cite{ma03static,POPL:JHL10,PLDI:KWP19}.
Because types carry resources, it is natural for AARA type systems to adopt \emph{linear} or \emph{affine} typing.
Below are some canonical typing rules for AARA, where $A \xrightarrow{p/q} B$ denotes a resource-annotated function type with $p$/$q$ as the pre-/post-potentials of the function:
\begin{mathpar}\small
  \Rule{AARA:Var}
  { }
  { x : A \vdashpq{0}{0} x : A }
  ~~
  \Rule{AARA:Nil}
  { p \ge 0 }
  { \cdot \vdashpq{0}{0} \mi{nil} : \tlist{A^p} }
  ~~
  \Rule{AARA:Cons}
  { p \ge 0 }
  { x_h : A, x_t : \tlist{A^p} \vdashpq{p}{0} \mi{cons}(x_h,x_t) : \tlist{A^p} }
  \and
  \Rule{AARA:Let}
  { \Gamma_1 \vdashpq{p}{r} e_1 : A_1 \\
    \Gamma_2, x:A_1 \vdashpq{r}{q} e_2 : A_2
  }
  { \Gamma_1,\Gamma_2 \vdashpq{p}{q} \mi{let}\;x = e_1\;\mi{in}\;e_2 : A_2 }
  ~~
  \Rule{AARA:Abs}
  { \Gamma, x : A \vdashpq{p}{q} e : B ~~ {\color{ACMDarkBlue}\text{$\Gamma$ does not carry any potentials}} }
  { \Gamma \vdashpq{0}{0} \lambda x.\, e : A \xrightarrow{p/q} B  }
\end{mathpar}
Rules \textsc{(AARA:Nil)} and \textsc{(AARA:Cons)} describe how to store potentials in a list;
for example, to build $\mi{cons}(x_h,x_t)$ of type $\tlist{\tint^p}$ carrying $p$ units of potential per list element, one needs $p$ units of potential
(for $x_h$) and a $x_t$ that has type $\tlist{\tint^p}$.
Rule \textsc{(AARA:Let)} indicates that the type system is linear:
one needs two different contexts $\Gamma_1$ and $\Gamma_2$ to check $e_1$ and $e_2$,
thus forbidding $e_1$ and $e_2$ to use the same variable.\footnote{This is not a severe limitation, because AARA can use \emph{sharing} to split a variable into two, with the understanding that the carried potentials are also split.}
Rule \textsc{(AARA:Abs)} is one of the most non-trivial rules: the side condition (marked in \textcolor{ACMDarkBlue}{blue})
is essential because AARA typically treats functions as \emph{copyable} objects, i.e., a function
can be applied an arbitrary number of times.
If the side condition were discarded, the function body $e$ would then be able to consume potentials stored in $\Gamma$
to pay for the costs inside $e$; thus, multiple applications of the function would consume the same piece of potentials multiple times, resulting in unsoundness.
Therefore, the rules above can \emph{not} verify the resource bound of the \emph{curried} \texttt{append} function shown in \cref{Se:Introduction},
because it would require a typing judgement like $x : \tlist{\tint^1} \vdash \lambda y.\, e_{\ms{append}} : \tlist{\tint^0} \xrightarrow{0/0} \tlist{\tint^0}$, where the context \emph{does} carry potentials.

Several efforts have been made to relax the limitation of \textsc{(AARA:Abs)}~\cite{POPL:HDW17,PLDI:KWP19,vi21unify}.
Hoffmann et al.~\cite{POPL:HDW17} adapt a stack-based typing principle, which essentially uncurries function applications if needed.
Below shows two typing rules of Hoffmann et al.~\cite{POPL:HDW17}'s system for function definitions, where a stack $\Sigma$ is used in the typing judgement $\Sigma;\Gamma \vdash  e : T$ to record all the argument types:
\begin{mathpar}\small
  \Rule{AARA:Stack:AbsPop}
  { \Sigma;\Gamma \vdash \lambda x.\, e : T }
  { \cdot;\Gamma \vdash \lambda x.\, e : \Sigma\rightarrow T }
\hva \and
  \Rule{AARA:Stack:AbsPush}
  { \Sigma;\Gamma,x:T_1 \vdash \lambda x.\, e : T }
  { T_1::\Sigma;\Gamma \vdash \lambda x.\, e : \Sigma\rightarrow T }
\end{mathpar}
By the interaction of the \textsc{(AARA:Stack:AbsPop)} and \textsc{(AARA:Stack:AbsPush)}, one becomes able to uncurry a higher-order type like $T_1\to\cdots\to T_n \to T$ into a bracket function type $[T_1,\cdots,T_n] \to T$.
Under this approach, the limitation of \textsc{(AARA:Abs)} is clearly handled by introducing all the $n$ variables at the same time, thus
curried functions like \texttt{append} can have a sound type annotation like $[\tlist{\tint^1},\tlist{\tint^0}]\to\tlist{\tint^0}$, instead of the unsound annotation $\tlist{\tint^1}\to\tlist{\tint^0}\to\tlist{\tint^0}$.
However, this approach sets us aside from one of the most powerful tools---partial application---that higher-order functions can offer.

On the other hand, Knoth et al.~\cite{PLDI:KWP19} and Rajani et al.~\cite{vi21unify} adopt an affine type system with \emph{multiplicities}, which bound the number a function can be applied.
%
%
Take Knoth et al.~\cite{PLDI:KWP19}'s system as an example, function types in their works has the form similar to $m \cdot (A \xrightarrow{p/q} B)$, where $m$ is a non-negative integer for the multiplicity.
Note that we switch to a slightly different notation for multiplicities, compared with the example shown in \cref{Se:Introduction}.
Intuitively, $\infty \cdot (A \xrightarrow{p/q} B)$ has the same meaning of the function type justified by \textsc{(AARA:Abs)},
whereas $1 \cdot (A \xrightarrow{p/q} B)$ is similar to the function type in a usual linear type system, in the sense that such function can only be applied once.
The typing rule for function definitions could become the one shown below, where $m \cdot \Gamma$ constructs
a context with the same bindings as $\Gamma$, but with $m$ times of the potentials in $\Gamma$:
\begin{mathpar}\small
  \Rule{AARA:Abs:Multi}
  { \Gamma, x : A \vdashpq{p}{q} e : B }
  { m \cdot \Gamma \vdashpq{0}{0} \lambda x.\, e : m \cdot (A \xrightarrow{p/q} B) }
\end{mathpar}
With the rule above, it is possible to verify the curried \texttt{append} function has
the type $\infty \cdot (  \tlist{\tint^1} \xrightarrow{0/0} 1 \cdot ( \tlist{\tint^{0}} \xrightarrow{0/0} \tlist{\tint^0} )   )$, which we showed in \cref{Se:Introduction}.

However, as we discussed in \cref{Se:Introduction}, the linear or affine nature of AARA type systems
makes it quite rigid to handle higher-order functions.
(Recall the \texttt{app\_par} example in \cref{Se:Introduction} which partially applies \texttt{append}
to obtain a closure and later uses the closure twice.)
In \cref{sec:overview:thiswork}, we show how our \lang{} tackles the problem.

\subsection{This Work: Non-Affine Dependently-Typed AARA}
\label{sec:overview:thiswork}


The first observation on AARA-style type systems is that the affine typing is necessary for dealing with potential-carrying types.
Most prior systems neglect the desire to separate potentials from types, despite that they actually feature
such a separation in some parts.
For example,
in the typing judgement $\Gamma\vdashpq{p}{q}\,e:T$, you get $p$ and $q$ as input/output \emph{constant} potentials, which stand aside from context $\Gamma$ and type $T$.
However, such a separation is not enough because potentials associated with data structures like $\tlist{\tint^1}$ can \emph{not} be expressed by a constant.
To decouple potentials completely from types, in \lang{}, we extend typing judgements of the form $\Gamma\vdashpq{p}{q}\,e:T$ to the form
\[
\tjudge{\Omega}{\Gamma}{f_1}{e}{f_2}{T}{x} \;,
\]
which replaces constant $p,q$ with potential functions $f_1,f_2$ and reads as: under potential context $\Omega$ (which we will explain later in this section), type context $\Gamma$ with $f_1$ units of input potential, the expression $e$ has the type $T$, whose result is bound by a local binder $x$ and carries an output potential of $f_2$ units.
Thus in our design, $f_1$ and $f_2$ will carry all the potentials used and remaining,
while the potential context $\Omega$ and type context $\Gamma$ carry \emph{no} potentials. 

Then, a typical typing judgement in our system
\[
\cdot \mid y : \tlist{\tint} \mid \mi{length}(y) \vdash \ms{id}\;y : [\mi{length}(x)]_x \tlist{\tint}
\]
may be read as: with a variable $y$ of a list type in the context and $\mi{length}(y)$ units of potential, the expression $\ms{id}\;y$ can return a result of a list type with $\mi{length}(x)$ units of potential, where $x$ binds the result of the expression.
The type-level potential function $\mi{length}(\cdot)$ serves as the resource annotation $\tlist{\tint^1}$ in usual AARA systems.
Since the type-level $\mi{length}(\cdot)$ function can be applied to program variables, our
\lang{} system is naturally \emph{dependently} typed.
%
%
Thus, one would define the type-level $\mi{length}(\cdot)$ function as an ordinary function:
\begin{lstlisting}
$\mi{length}$ = fun x. case x of nil => 0 | cons(x0,x1) => 1 + $\mi{length}$ x1
\end{lstlisting}
It should be noted that type-level functions should guarantee \emph{termination}, in order to render the type system sound.
There have been many techniques in the community of dependent typing, such as measures~\cite{PLDI:PKS16},
well-founded recursion~\cite{POPL:SHK16}, and primitive recursion on inductive datatypes~\cite{knoth20liquid}.
In our design of \lang{}, we do not stick to any specific design of type-level functions; instead,
we focus on the idea that \emph{what capability our type system can offer if we can carry out
arithmetic reasoning on type-level potential functions}.
Thereby in the rest of the section, we assume we can perform arithmetic reasoning in the type system.

%

Because of the dependent nature of \lang{}, arrow types become $\arrow{x}{f_1}{A}{f_2}{B}{y}$ with binders $x$ and $y$, which bind the argument and result and can be referred to in potential functions $f_1$ and $f_2$.
%
%
It is now possible  to describe the curried \texttt{append}'s behavior in \lang{} as follows, where $matd()$ stands for the case analysis for inductive datatypes adopted from Knoth et al.~\cite{knoth20liquid}'s work:
\[
\begin{array}{lrcll}
T_\ms{append} &\defeq& \arrow{x}{0}{\tlist{\tint}}{0}{(\arrow{y}{\mi{length}(x)}{\tlist{\tint}}{0}{\tlist{\tint}}{})}{}\;\\
\ms{append} &\defeq&\fix{\ms{append}}{T_\ms{append}}{\lam{x}{\lam{y}{matd(x,\{nil(x_0).\; y,cons(x_0,x_1).\\
&&\tick{1}{(cons(x_0,(\ms{append}~x_1~y)))}\})}}}
\end{array}
\]
Notably, such a type precisely characterizes the cost behavior of the \emph{curried} \texttt{append}:
applying \texttt{append} with only one argument $x$ would not consume any resources, and
further applying the obtained closure with a second argument $y$ would consume $\mi{length}(x)$ units of resource, where $x$ is the captured first argument.
The key to this precise characterization is that we let the outermost binder in lambda abstraction capture all the potentials in the context:
\begin{mathpar}\small
\inferrule[(TAbs)]
{ \tjudge{\Omega}{\mctx{\Gamma}{x}{T_1}}{f_1}{e}{f_2}{T_2}{y}\\
z\notin \ms{dom}(\Omega)\cup\ms{dom}(\Gamma)
}
{ \tjudge{\Omega}{\Gamma}{0}{\llam{x}{e}{T_1}}{0}{\arrow{x}{f_1}{T_1}{f_2}{T_2}{y}}{z}}
\end{mathpar}
While this rule may look similar to previously shown rules \textsc{(AARA:Abs)} or \textsc{(AARA:Abs:Multi)},
the exception is our \textsc{(TAbs)} does not require potential-free context restrictions or explicit multiplicities.
This is only viable in our \lang{} because all the contexts $\Omega$ and $\Gamma$ are guaranteed to carry $0$ potentials, and all the required potentials can be expressed and captured in the $f_1$ part.
This feature enables us to do \emph{compositional} reasoning about closures and partial applications, but yet a bit more elaboration on lambda applications is needed next.
The second observation on AARA-style type systems is that the \emph{multiplicity} is a remedy for the imprecise characterization of function behaviors.
Take the linear type $\infty \cdot (\tlist{\tint^1}\to 1 \cdot (\tlist{\tint^0}\rightarrow\tlist{\tint^0}))$ for \texttt{append} as an example.
The first argument requires $\mi{length}(\cdot)$ units of potential in advance while the consumption actually happens after the second application.
Thus a restriction of multiplicity $1$ is imposed on the second arrow because the first argument only ask for potentials that is enough for $1$ time of second application.
In practice, one may \emph{not} always be able to decide the number of usages in advance, thus multiplicity results in \emph{non-compositionality}: you have to change the type annotation (especially the multiplicity) by the actual context involved.

In our approach, we remove multiplicities to achieve \emph{compositionality}, but this removal does not come for free, especially in the case of lambda-applications.
Before proceeding to the concrete typing rule, let’s take an example to illustrate the challenge.
Suppose we want to subject an expression $e_1$ of type $\tlist{\tint}$ with $4\times\mi{length}(\ell_1)$ units of potential---where $\ell_1$ binds the evaluation result of $e_1$---to the function \texttt{append}.
We notice that in the first arrow of $T_\ms{append}$, the type requires $0$ units of potential.
Thus, we need to keep the $4 \times \mi{length}(\ell_1)$ units of potential for later use because there is no \emph{multiplicity} for us to restrict the number of usage of the second arrow.
In other words, $\ms{append}\;e_1$ results in a \emph{potential-capturing} closure:
\begin{mathpar}\small
  \inferrule*[right=(TApp?)]
  { \tjudge{\cdot}{\cdot}{0}{\ms{append}}{0}{T_\ms{append}}{} \\
    \tjudge{\cdot}{\cdot}{p_1}{e_1}{4 \times \mi{length}(\ell_1)}{\tlist{\tint}}{\ell_1}
  }
  { \tjudge{\cdot}{\cdot}{p_1}{ \ms{append}\;e_1 }{ 4 \times \mi{length}({\color{ACMRed}\ell_1}) }{ \arrow{y}{\mi{length}({\color{ACMRed}x})}{\tlist{\tint}}{0}{\tlist{\tint}}{} }{} }  
\end{mathpar}
However, the typing judgement shown above is problematic:
the variables $\color{ACMRed}\ell_1$ and $\color{ACMRed}x$ used in the result type are \emph{not} bound.
Intuitively, the application should be safe to proceed.
One possible workaround---adapted by Knoth et al.~\cite{PLDI:KWP19,knoth20liquid}---is to require programs to be in \emph{A-Normal-Form} (ANF); that is, the application arguments must be a variable or a value.
For example, if $e_1$ is a variable $z$ that can be located in the context, we can substitute both $\ell_1$ and $x$ with $z$ to make the type valid.
Another approach is to introduce \emph{existential} types~\cite{PLDI:PKS16,POPL:BVJ24};
for example, even if $e_1$ is not a variable or a value, one can assign
$\ms{append}\;e_1$ with the type $\exists z : \tlist{\tint}.\, \bigl[4 \times \mi{length}(z)\bigr ] \arrow{y}{\mi{length}(z)}{\tlist{\tint}}{0}{\tlist{\tint}}{} $,
which uses an existential binder $z$ for the evaluation result of $e_1$.

In this work, we aim to target a more flexible language, thus we do not want to restrict
our language to allow only ANF programs.
On the other hand, introducing existential binders would complicate the type system, which in itself
is worth a separate study (as shown by Borkowski et al.~\cite{POPL:BVJ24}).
%
%
%
As mentioned in \cref{Se:Introduction}, in our \lang{} system, we propose a lightweight approach:
we introduce a \emph{global potential context} $\Omega$ to keep track of those existentially-bound potential-carrying variables.
Below shows a derivation for the partial application $\ms{append}\;e_1$ in \lang{}:
\begin{mathpar}\small
  \inferrule*[right=(TApp)]
  { \tjudge{\cdot}{\cdot}{0}{\ms{append}}{0}{T_\ms{append}}{} \\
    \tjudge{\cdot}{\cdot}{p_1}{e_1}{4 \times \mi{length}(\ell_1)}{\tlist{\tint}}{\ell_1}
  }
  {{{\color{ACMDarkBlue}x:\tlist{\tint}}}~|~{\cdot}~|~{p_1}\vdash{ \ms{append}\;e_1 }:\\\qquad[{ 4 \times \mi{length}({\color{ACMDarkBlue}x}) }]{ \arrow{y}{\mi{length}({\color{ACMDarkBlue}x})}{\tlist{\tint}}{0}{\tlist{\tint}}{} }{} }  
\end{mathpar}
As you can see, a substitution of $\ell_1$ with $x$ is triggered and the global potential context is extended with a binding of $x$.
%
%
The typing judgement reads as follows: after the application of \texttt{append} on $e_1$ carrying $4 \times \mi{length}(\ell_1)$ units of potential with $\ell_1$ bound to $e_1$'s result, $\ell_1$ is then instantiated to the variable $x$ while $x$ is introduced as an existential variable to the global potential context.
%
%
We can further apply the result of $\ms{append}\;e_1$ to another expression $e_2$ of type $\tlist{\tint}$ with no potentials; the typing derivation is shown below:
\begin{mathpar}\small
\mprset{flushleft}
  \inferrule*[right=(TErase)]
  {\inferrule*[right=(TApp)]
  {  {x:\tlist{\tint}}~|~{\cdot}~|~{p_1}\vdash{ \ms{append}\;e_1 }:\\\qquad[{ 4 \times \mi{length}({x}) }]{ \arrow{y}{\mi{length}({x})}{\tlist{\tint}}{0}{\tlist{\tint}}{} }{} \\
    \tjudge{\cdot}{\cdot}{p_2}{e_2}{0}{\tlist{\tint}}{}
  }
  {{x:\tlist{\tint},\;{\color{ACMDarkBlue}y:\tlist{\tint}}}~|~{\cdot}~|~{p_1 + p_2}\vdash{\ms{append}\;e_1\;e_2}:\\\qquad[{3 \times \mi{length}(x)}]{\tlist{\tint}}{} } 
  }
  { \tjudge{x:\tlist{\tint}}{\cdot}{p_1 + p_2}{\ms{append}\;e_1\;e_2}{3 \times \mi{length}(x)}{\tlist{\tint}}{} }
\end{mathpar}
Similar to the previous use of the \textsc{(TApp)} rule, the application of $\ms{append}\;e_1$ to
$e_2$ introduces another existential binding $y:\tlist{\tint}$.
The difference is that now the application needs $\mi{length}(x)$ units of potential
to proceed.
Because the closure returned by $\ms{append}\;e_1$ carries $4 \times \mi{length}(x)$ units of potential, it is sound to make the application and leave $3 \times \mi{length}(x)$ units of potential to the next.
In the derivation, we also demonstrate the use of the \textsc{(TErase)} rule, which removes
unnecessary existentially-bound variables from the global potential context.
In this example, $y$ is no longer needed so we can safely remove it.


We can now show \lang{}'s typing rule for lambda applications:
\begin{mathpar}\small
\inferrule[(TApp)]
{\tjudge{\Omega}{\Gamma}{f_1}{e_1}{f_2}{(\arrow{x}{f_3}{T_1}{f_4}{T_2}{y})}{z}\\
\tjudge{\Omega}{\Gamma}{f_5}{e_2}{f_6}{T_1}{w}\\
\pjudge{\mctx{\Omega}{x}{T_1}}{\Gamma}{f_3[w\mapsto x]}{\leq}{(f_2 + f_6)[w\mapsto x]}\\
}
{\tjudge{\mctx{\Omega}{x}{T_1}}{\Gamma}{f_1+f_5}{\app{e_1}{e_2}}{(f_2 + f_6-f_3)[w\mapsto x]+f_4}{T_2}{y}}
\end{mathpar}
The first two premises include the typing judgements for $e_1$ and $e_2$, respectively.
Then for the application, we need to instantiate all the potential associated with local binder $w$ for $e_2$, to the local binder $x$ for the argument type in $e_1$.
The third premise checks whether the potentials carried by $e_1$ and $e_2$ are sufficient for the application or not.
Finally, we store the unused potentials in $(f_2 + f_6-f_3)[w\mapsto x]+f_4$, and add an existential binder $x$ to the global potential context.
In this way, we can derive $\tjudge{x:\tlist{\tint},\,y:\tlist{\tint}}{\cdot}{p_1+p_2}{\ms{append}\;e_1\;e_2}{3 \times \mi{length}(x)}{\tlist{\tint}}{}$ as shown earlier in this section.

With the lambda-abstraction rule and application rule provided, we can finally derive a typing judgement for the curried recursive function \texttt{append}.
Let
\begin{align*}
\Gamma_0 & \defeq   \ms{append}:T_\ms{append},\,x:\tlist{\tint},\,y:\tlist{\tint} \,, \\
\Gamma_1 & \defeq \Gamma_0,\,x_0:\tint,\,x_1:\tlist{\tint}\, .
\end{align*}
We have the following derivation for the sub-expression that corresponds to the $cons$ case. The notable thing here is how the potential is carried along with the sequential applications of $x_1$ and $y$:
\begin{mathpar}\small
\mprset{flushleft}
  \inferrule*[Right=(TTickp)]
  {\inferrule*[Right=(TCons)]
  {\inferrule*[Right=(TErase)]
  {\inferrule*[Right=(TApp)]
  {\inferrule*[Right=(TApp)]
  {\tjudge{\cdot}{\Gamma_1}{0}{\ms{append}}{0}{T_\ms{append}}{}\\\tjudge{\cdot}{\Gamma_1}{\mi{length}(x_1)}{x_1}{\mi{length}(x_1)}{\tlist{\tint}}{x_1}
  }
  {\tjudge{x:\tlist{\tint}}{\Gamma_1}{\mi{length}(x_1)}{\ms{append}~x_1\\\qquad}{\mi{length}(x)}{\arrow{y}{\mi{length}(x)}{\tlist{\tint}}{0}{\tlist{\tint}}{}}{}}
  }
  {{x:\tlist{\tint},y:\tlist{\tint}}|{\Gamma_1}|{\mi{length}(x_1)}\vdash{\ms{append}~x_1~y}:[{0}]{\tlist{\tint}}}
  }
  {\tjudge{\cdot}{\Gamma_1}{\mi{length}(x_1)}{\ms{append}~x_1~y}{0}{\tlist{\tint}}{}}
  }
  {\tjudge{\cdot}{\Gamma_1}{\mi{length}(x_1)}{cons(x_0,(\ms{append}~x_1~y))}{0}{\tlist{\tint}}{}
  }
  }
  { \tjudge{\cdot}{\Gamma_1}{\mi{length}(x_1)+1}{\tick{1}{cons(x_0,(\ms{append}~x_1~y))}}{0}{\tlist{\tint}}{} }
\end{mathpar}
We then proceed to the case analysis as well as the fixed-point definition:
\begin{mathpar}\small
\mprset{flushleft}
  \inferrule*[Right=(TFix)]
  {\inferrule*[Right=(TAbs)]
  {\inferrule*[Right=(TAbs)]
  {
  {\inferrule*[Right=(TDes)]
  {\cdots \\\tjudge{\cdot}{\Gamma_0,\,x_0:{\tunit}}{0}{y}{0}{\tlist{\tint}}{}\\\tjudge{\cdot}{\Gamma_0}{\mi{length}(x)}{x}{\mi{length}(x)}{\tlist{\tint}}{x}
  }
  {\tjudge{\cdot}{\Gamma_0}{\mi{length}(x)}{matd(\cdots)}{0}{\tlist{\tint}}{}
  }}}{\cdot~|~{\ms{append}:T_\ms{append},\,x:{\tlist{\tint}}}~|~{0}\vdash{\lam{y}{\cdots}{}}:\\\qquad[{0}]{(\arrow{y}{length(x)}{\tlist{\tint}}{0}{\tlist{\tint}}{})}}}
  {\tjudge{\cdot}{\ms{append}:T_\ms{append}}{0}{\lam{x}{\lam{y}{}}\cdots}{0}{T_\ms{append}}{}}
  }
  { \tjudge{\cdot}{\cdot}{0}{\fix{\ms{append}}{T_\ms{append}}{\lam{x}{\lam{y}{}}\cdots}}{0}{T_\ms{append}}{} }
\end{mathpar}
Let us conclude using the higher-order example \texttt{app\_par} introduced in \cref{Se:Introduction}.
%
%
We can reimplement \texttt{app\_par} as follows using lambda abstraction and application:
\[
\app{ (\lam{z}{\pair{\app{z}{\ell_2}}{\app{z}{\ell_3}}}) }{(\app{\ms{append}}{\ell_1})}
\]
To show the capability of \lang{} handling higher-order functions, we let $\ell_1$ carrying $2\times \mi{length}(\ell_1)$ units of potential, while $\ell_2$ and $\ell_3$ carry no potentials.
Let $T_{z}$ be $\arrow{y}{\mi{length}(x)}{\tlist{\tint}}{0}{\tlist{\tint}}{}$, i.e., the type that the argument $z$ is supposed to have.
The typing derivation for the $\lam{z}{\cdots}{}$ part in \lang{} is presented below:
\begin{mathpar}\small

 \inferrule*[Right=(TAbs)]
{
\inferrule*[Right=(TPair)]
{
 \inferrule*[Right=(TApp)]
{\tjudge{x:\tlist{\tint}}{z:T_{z}}{0}{z}{0}{\arrow{y}{\mi{length}(x)}{\tlist{\tint}}{0}{\tlist{\tint}}{}}{}}
{\tjudge{x:\tlist{\tint}}{z:T_{z}}{\mi{length}(x)}{\app{z}{\ell_2}}{0}{\tlist{\tint}}{}}
}
{
\tjudge{x:\tlist{\tint}}{z:T_{z}}{2 \times \mi{length}(x)}{\pair{\app{z}{\ell_2}}{\app{z}{\ell_3}}}{0}{\prodT{\tlist{\tint}}{\tlist{\tint}}}{}
}
}
{{x:\tlist{\tint}}|{\cdot}|{0}\vdash{\lam{z}{\pair{\app{z}{\ell_2}}{\app{z}{\ell_3}}}{}}:[{0}]{\arrow{z}{2\times \mi{length}(x)}{T_z}{0}{\prodT{\tlist{\tint}}{\tlist{\tint}}}{}}{}
}
\end{mathpar}
Next, we derive a typing judgement for the partial application $\ms{append}\;\ell_1$:
\begin{mathpar}\small
  \inferrule*[Right=(TApp)]
{\tjudge{\cdot}{\cdot}{p}{\ell_1}{2\times \mi{length}(l)}{\tlist{\tint}}{l}}
{{x:\tlist{\tint}}~|~{\cdot}~|~{p}\vdash{\app{\ms{append}}{\ell_1}}:\\\qquad[{2\times \mi{length}(x)}]{\arrow{y}{\mi{length}(x)}{\tlist{\tint}}{0}{\tlist{\tint}}{}}{}}
\end{mathpar}
Finally, we can compositionally reason about the resource consumption of the \texttt{app\_par} example:
\begin{mathpar}\small
  \inferrule*[Right=(TApp)]
{\cdots\\\cdots}
{
  \inferrule*[Right=(TErase)]
{{x:\tlist{\tint}}|{\cdot}|{p}\vdash{\app{ (\lam{z}{\pair{\app{z}{\ell_2} }{\app{z}{\ell_3}}}{} ) }{(\app{\ms{append}}{\ell_1})}}:[{0}]{\prodT{\tlist{\tint}}{\tlist{\tint}}}{}}
{\tjudge{}{\cdot}{p}{\app{ (\lam{z}{\pair{\app{z}{\ell_2}}{\app{z}{\ell_3}}} )}{(\app{\ms{append}}{\ell_1})}}{0}{\prodT{\tlist{\tint}}{\tlist{\tint}}}{}}
}
\end{mathpar}
That is, $p$ is an upper bound on the resource consumption if $p$ units of potential are sufficient to build
the list $\ell_1$ with $2 \times \mi{length}(\ell_1)$ units of potential.



%% file: details.tex
\section{Technical Details}
\label{sec:details}
In this section, we will give an overall description of the syntax and semantics of \lang{}, along with sketches of the type soundness proof under cost semantics. While examples and automation of \lang{} will be later introduced at \cref{sec:eval} and \cref{Se:Discussion}.

\subsection{Syntax}
 {\small
    \begin{align*}
    &\textsc{Types} &T\Coloneqq&~  \Int ~|~ \prodT{T_1}{T_2} ~|~ \arrow{x}{f_1}{T_1}{f_2}{T_2}{y} ~\\&&&|~ \poly{x}{T_1}{T_2} ~|~ \ind   \\
    &\textsc{Expression} &e\Coloneqq&~   x ~|~ i ~|~ \op{i}{e} ~|~ \Lam{x}{e}{T} ~|~ \llam{x}{e}{T} ~ |~\app{e_1}{e_2} ~|~ \pair{e_1}{e_2} ~\\&&&|~ \fst~e~|~\snd~e~|~\fix{x}{T}{e}~|~\tick{i}{e}~|~~\Let{x}{e_1}{e_2}  \\
   &&&~|~\con{i}{e_0}{e_1}{e_{m_i}}~|~\des{C}{m}{e}\\
    &\textsc{Pre-Values} &pv\Coloneqq&~  x ~|~ i ~|~ \op{i}{pv} ~|~ \Lam{x}{e}{T} ~|~ \llam{x}{e}{T}~|~\pair{pv_1}{pv_2}~\\\
&&&|~ \con{i}{pv_0}{pv_1}{pv_{m_i}} \\
    &\textsc{Values} &v\Coloneqq&~   i ~|~ \Lam{x}{e}{T} ~|~ \llam{x}{e}{T}~|~\pair{v_1}{v_2}~|~ \con{i}{v_0}{v_1}{v_{m_i}} \\
    &\textsc{Type Context} &\Gamma\Coloneqq&~   \cdot ~|~  \mctx{\Gamma}{x}{T} \\
    &\textsc{Potential Context} &\Omega\Coloneqq&~   \cdot ~|~  \mctx{\Omega}{x}{T}
  \end{align*}
  }%
\paragraph{\textbf{Types}} The types and terms in our calculus are a combination of AARA~\cite{ma03static} and Liquid Resource Type~\cite{knoth20liquid}. We choose $\Int$ to serve as basic types to do arithmetic reasoning, while other resource types such as $\mathbf{R}$ work just as well in our system. $\prodT{T_1}{T_2}$ and $\poly{x}{T_1}{T_2}$ are standard product types and dependent types. Our arrow type $\arrow{x}{f_1}{T_1}{f_2}{T_2}{y}$ is a generalization from AARA arrow type $A\stackrel{p/q}{\longrightarrow}B$. $f_1$ and $f_2$ are the generalized potential resource input/output from $p$ and $q$, while $x$ and $y$ specify potential variables bound on the input and output potential function $f_1$ and $f_2$. Here $f_1$, $f_2$, $T_1$ and $T_2$ can depend on $x$ while only $f_2$ and $T_2$ can depend on $y$. $\ind$ is a simplified inductive data type adapted from prior work~\cite{knoth20liquid,POPL:HDW17}. $C$ stands for the name of the constructor, $T$ is the type of content, and $m$ is the number of copies of the inductive type itself. Such as $\tlist{\tint}$ is represented as $\tlist{\tint}=ind(\{nil(Unit,0),cons(\Int,1)\})$, and $\ttree{\tint}$ is represented as $\ttree{\tint}=ind(\{leaf(\Int,0),node(Unit,2)\})$.
\begin{figure}[thb]
\begin{center}
\begin{mathpar}\small
\inferrule[(PConst)]
{\Omega\cap\Gamma=\emptyset}
{
\pjudge{\Omega}{\Gamma}{c}{}{}
}
\hva \and
\inferrule[(PInt)]
{\Omega\cap\Gamma=\emptyset\\
x:\Int\in\Omega\cup\Gamma}
{
\pjudge{\Omega}{\Gamma}{x}{}{}
}
\hva \and
\inferrule[(POp-Linear)]
{\mathbf{op~is~linear}\\\\
\forall i,~\pjudge{\Omega}{\Gamma}{f_i}{}{}\\
}
{\pjudge{\Omega}{\Gamma}{\op{}{f_i}}{}{}}
\hva \and
\inferrule[(POp-Non-Linear)]
{\mathbf{op~is~non~linear}\\\\
\forall i,~\pjudge{\Omega_i}{\Gamma_i}{f_i}{}{}\\\\
\forall i,j,~i\neq j\rightarrow \ms{dom}(\Omega_i)\cap\ms{dom}(\Gamma_j)=\ms{dom}(\Omega_i)\\\\\cap\ms{dom}(\Omega_j)=\ms{dom}(\Gamma_i)\cap\ms{dom}(\Gamma_j)=\emptyset\\
}
{\pjudge{\bigcup\limits_i\Omega_i}{\bigcup\limits_i\Gamma_i}{\op{}{f_i}}{}{}}
~~~~
\inferrule[(PPair)]
{\forall i\in [1,2],x_i\notin \ms{dom}(\Omega)\cup\ms{dom}(\Gamma)\\\\
\\ x:\prodT{T_1}{T_2}\in\Omega\cup\Gamma,\qquad\forall i\in [1,2],\\\\
\pjudge{\mctx{\Omega\backslash x}{x_i}{T_i}}{\Gamma\backslash x}{f_i[x\mapsto (x_1,x_2)]}{}{}\\
}
{
\pjudge{\Omega}{\Gamma}{\sum\limits_i f_i}{}{}
}
\hva \and
\inferrule[(PCons)]
{\forall i,\forall j\in [1,m_i],x_{ij}\notin \ms{dom}(\Omega)\cup\ms{dom}(\Gamma)\\ x:\ind\in\Omega\cup\Gamma\\
\forall i,\forall j\in [1,m_i],\pjudge{\mctx{\mctx{\Omega\backslash x}{x_{ij}}{T_{ij}}}{f}{\Int}}{\Gamma\backslash x}{f_{ij}[x\mapsto \con{i}{x_0}{x_1}{x_{m_i}}]}{}{}\\
}
{
\pjudge{\Omega}{\Gamma}{fix~f.\;\destructor{x}{C}{m_i}{\sum\limits_{j\in m_i} f_{ij}}}{}{}
}
\end{mathpar}
\end{center}
\caption{Potential Functions}
\label{fig:potential}
\end{figure}

\paragraph{\textbf{Expressions}}
Our expressions are based on the standard lambda calculus (including type abstraction, fixpoints, and product types), extended with cost semantics and recursive data structures. We annotate lambda abstractions $\llam{x}{e}{T}$ with the argument type $T$ to enable an algorithmic typing version, as discussed in \Cref{Se:Discussion}. However, these annotations can be omitted in practice, as all typing, evaluation, and soundness results hold regardless of their presence. Indeed, we omit them in the presentation of examples for readability.
The cost construct $\tick{i}{e}$ expresses the consumption or release of $i$ units of resource before evaluating the expression $e$—consuming resources when $i$ is positive and releasing them when $i$ is negative. A typical example of resource release is memory deallocation: once a program releases ownership of a block of memory, that portion of the resource becomes available for reuse, which corresponds to $\tick{-1}$.
For recursive datatypes, the term $\con{i}{e_0}{e_1}{e_{m_i}}$ denotes the $i$-th constructor, and $\des{C}{m}{e}$ represents a case analysis over the  constructors of datatype $C$ applied to expression $e$. We also introduce the arithmetic operation $\op{i}{e}$ to support the definition and computation of potential functions.

\paragraph{\textbf{Values and Pre-Values}}
Values are the normal forms that cannot be further reduced. Our values are standard, but we also include pre-values in our calculus. Pre-values are values with variables and arithmetic operations. We need these pre-values because we want to subject some potential functions using the potential abstraction $\Lam{X}{e}{T}$. Dependent type systems with fix-points will not be decidable, so here we subject the restricted forms as pre-values into potential polymorphism $\Lam{X}{e}{T}$, as shown in the typing rule \textsc{(TPapp)} rule in \Cref{fig:typing1}.

\paragraph{\textbf{Type Context and Potential Context}}
As shown in the \cref{Se:Introduction} and \cref{sec:overview}, there are two kinds of contexts, type context $\Gamma$ and potential context $\Omega$. Type context $\Gamma$ captures bound variables in the $\lambda$ abstraction and fix-points, while potential context $\Omega$ captures existential global variables to bind potential functions. All the variables used in expressions or potentials function should appear in either type context or potential context, and common variables in type context and potential context should share the same type.

\paragraph{\textbf{Potential functions}}
Potential functions $f$ are introduced to decouple potentials from types, especially from recursive data types. In our context, potential functions are just the primitive recursions that map pre-values to $\Int$. There is a set of rules associated with potential functions. $\pjudge{\Omega}{\Gamma}{f}{}{}$ is the well-formedness judgements, requiring all the variables in $f$ appears in the context $\Omega$ or $\Gamma$. Besides, we have the inequality judgement $\pjudge{\Omega}{\Gamma}{f_1}{\le}{f_2}$ is the inequality judging, requires $f_1\le f_2$ under all instantiation of free variables by values $v$, to be solved by Linear Arithmetic Solver (with only linear operations), or other SMT solvers (with non-linear operations). The primitive recursion under our context only proceeds through product types $\prodT{T_1}{T_2}$ and inductive types $\ind$, but have constant values on other types, similar to the \emph{measures} used in some liquid type systems~\cite{knoth20liquid,PLDI:PKS16}.

\subsection{Typing Rules}
Next we can introduce the type system of \lang{}. The typing judgement $\tjudge{\Omega}{\Gamma}{f_1}{e}{f_2}{T}{x}$ reads as follows, under the potential context $\Omega$, typing context $\Gamma$ and with an input potential function $f_1$, we can have expression $e$ as type $T$, with an output potential function $f_2$. Besides, the output potential function $f_2$ has a local binder $x$ on the type $T$, suggesting $f_2$ and $T$ can depend on $x$. We sometimes omit the local binder $x$ if $x$ does not appear in $f_2$ and $T$, but here all the local binders are explicitly provided in the typing rules. Local binders can be variables that already appears in the context, as rule \textsc{(TVar)} shows, but for the other cases, we use a fresh variable to bind
the result. The most important thing to notice in \lang{}'s typing is the interaction between input/output potential functions, type structures and local binders.

\begin{figure}[thb]

\begin{center}
\begin{mathpar}\small
\inferrule* [lab=(TTickp),narrower=1]
{\tjudge{\Omega}{\Gamma}{f_1}{e}{f_2}{T}{x}\\
p\ge 0
}
{\tjudge{\Omega}{\Gamma}{f_1+p}{\tick{p}{e}}{f_2}{T}{x}}
\hva \and
\inferrule* [lab=(TTickn)]
{\tjudge{\Omega}{\Gamma}{f_1-p}{e}{f_2}{T}{x}\\
p < 0
}
{\tjudge{\Omega}{\Gamma}{f_1}{\tick{p}{e}}{f_2}{T}{x}}
\hva \and
\inferrule[(TInt)]
{ x\notin \ms{dom}(\Omega)\cup\ms{dom}(\Gamma)}
{\tjudge{\Omega}{\Gamma}{\fctx{0}}{i}{0}{\Int}{x}}
~~~~
\inferrule* [lab=(TDrop)]
{\tjudge{\Omega}{\Gamma}{f_1}{e}{f_2}{T}{x}\\\\
\pjudge{\mctx{\Omega}{x}{T}}{\Gamma}{f_3}{\leq}{f_2}\\
}
{
\tjudge{\Omega}{\Gamma}{f_1}{e}{f_3}{T}{x}
}
~~~~
\inferrule* [lab=(TFix)]
{z\notin \ms{dom}(\Omega)\cup\ms{dom}(\Gamma)\\\\
x\notin \mathbf{typefv}(e)\cup\mathbf{fv}(T)\\\\
\tjudge{\Omega}{\mctx{\Gamma}{x}{T}}{0}{e}{0}{T}{y}\\
}
{\tjudge{\Omega}{\Gamma}{0}{\fix{x}{T}{e}}{0}{T}{z}}
\hva \and
\inferrule[(TRename)]
{x,y\notin \ms{dom}(\Omega)\cup\ms{dom}(\Gamma)\\\\
\tjudge{\Omega}{\Gamma}{f_1}{e}{f_2}{T}{x}\\
}
{\tjudge{\Omega}{\Gamma}{f_1}{e}{f_2[x\mapsto y]}{T[x\mapsto y]}{y}\\}
\hva \and
\inferrule[(TProj1)]
{y\notin \ms{dom}(\Omega)\cup\ms{dom}(\Gamma)\\\\
\tjudge{\Omega}{\Gamma}{f_1}{e}{f_2}{\prodT{T_1}{T_2}}{x}\\\\
\pjudge{\mctx{\Omega}{x}{\prodT{T_1}{T_2}}}{\Gamma}{f_3[y\mapsto \fst x]}{\leq}{f_2}\\
}
{\tjudge{\Omega}{\Gamma}{f_1}{\fst{e}}{f_3}{T_1}{y}}
\hva \and
\inferrule[(TRelax)]
{\tjudge{\Omega}{\Gamma}{f_1}{e}{f_2}{T}{x}\\\\
\pjudge{\Omega}{\Gamma}{f_3}{\ge}{0}\\
}
{
\tjudge{\Omega}{\Gamma}{f_1+f_3}{e}{f_2+f_3}{T}{x}
}
~~~~
\inferrule* [lab=(TVar)]
{x:T \in \Gamma
}
{\tjudge{\Omega}{\Gamma}{0}{x}{0}{T}{x}}
\end{mathpar}
\end{center}
\caption{Typing Rules (Selected, full in Appendix)}
\label{fig:typing1}
\vspace{-20pt}
\end{figure}

Most of rules like rule \textsc{(TInt)}, rule \textsc{(TPair)}, or  rule \textsc{(TPabs)} are rather common since they do not involve potential changes. In \Cref{fig:typing1},
Rule \textsc{(TTickp)} and rule \textsc{(TTickn)} demonstrate the different behavior of consuming and releasing $p$ units of resource into the context. Rule \textsc{(TVar)} allows you to use variables non-affinely but set a local binder that has the same name of the variable while potentials can be assigned later with the rule \textsc{(TRelax)} or rule \textsc{(TDrop)}. Rule \textsc{(TRelax)} shows an increase on both the input and output resource by an equal amount is possible. Rule \textsc{(TDrop)} suggests the output resource can always be dropped without any side effects.
Rule \textsc{(TFix)} shows a standard rule for fixpoints, except we prohibit appearance of $x$ instead the types and potentials of $e$ to avoid self pointing. Finally rule \textsc{(TErase)} helps erase unnecessary variables in potential contexts, while \textsc{(TRename)} helps rename local binder that does not appear in context.

\begin{figure}[thb]

\begin{center}
\begin{mathpar}\small
\inferrule[(TAbs)]
{ \tjudge{\Omega}{\mctx{\Gamma}{x}{T_1}}{f_1}{e}{f_2}{T_2}{y}\\
z\notin \ms{dom}(\Omega)\cup\ms{dom}(\Gamma)
}
{ \tjudge{\Omega}{\Gamma}{0}{\llam{x}{e}{T_1}}{0}{(\arrow{x}{f_1}{T_1}{f_2}{T_2}{y}}{z})}
\hva \and
\inferrule[(TApp)]
{\tjudge{\Omega}{\Gamma}{f_1}{e_1}{f_2}{(\arrow{x}{f_3}{T_1}{f_4}{T_2}{y})}{z}\\
\tjudge{\Omega}{\Gamma}{f_5}{e_2}{f_6}{T_1}{w}\\
\pjudge{\mctx{\Omega}{x}{T_1}}{\Gamma}{f_3[w\mapsto x]}{\leq}{(f_2 + f_6)[w\mapsto x]}\\
}
{\tjudge{\mctx{\Omega}{x}{T_1}}{\Gamma}{f_1+f_5}{\app{e_1}{e_2}}{(f_2 + f_6-f_3)[w\mapsto x]+f_4}{T_2}{y}}
\hva \and
\inferrule[(TLet)]
{\tjudge{\Omega}{\Gamma}{f_1}{e_1}{f_2}{T_1}{z}\\
\tjudge{\Omega}{\mctx{\Gamma}{x}{T_1}}{f_3}{e_2}{f_4}{T_2}{y}\\
\pjudge{\Omega}{\mctx{\Gamma}{x}{T_1}}{f_3}{\leq}{f_2[z\mapsto x]}\\
}
{\tjudge{\mctx{\Omega}{x}{T_1}}{\Gamma}{f_1+f_5}{\Let{x}{e_1}{e_2}}{f_2[z\mapsto x] -f_3+f_4}{T_2}{y}}
\hva \and
\inferrule[(TCons)]
{y\notin \ms{dom}(\Omega)\cup\ms{dom}(\Gamma)\\
\tjudge{\Omega}{\Gamma}{f_1}{e_0}{f_2}{T_i}{x_0}\\
\forall j\in[1,m_i],\tjudge{\Omega}{\Gamma}{f_{3_j}}{e_j}{f_{4_j}}{\ind}{x_j}\\
\pjudge{\Omega,\,x_0:T_i,...,\,x_{m_i}:\ind}{\Gamma}{f_5[y\mapsto \con{i}{x_0}{x_1}{x_{m_i}}]}{\leq}{f_2+\sum\limits_{j=1}^{m_i} f_{4_j}}{}\\
}
{\tjudge{\Omega}{\Gamma}{f_1 + \sum\limits_{j=1}^{m_i} f_{3_j}}{\con{i}{e_0}{e_1}{e_{m_i}}}{f_5}{\ind}{y}}
\hva \and
\inferrule[(TDes)]
{\tjudge{\Omega}{\Gamma}{f_1}{e_0}{f_2}{\ind}{x}\\
\forall i,\forall j\in[1,m_i],x_j\notin \mathbf{fv}(f_3)\cup\mathbf{fv}(T_1)\\
\forall i,\tjudge{\Omega}{\Gamma,\,x_0:T_i,...,\,x_{m_i}:\ind}{f_2[x\mapsto \con{i}{x_0}{x_1}{x_{m_i}}]}{e_i}{f_3}{T_1}{y}
}
{\tjudge{\Omega}{\Gamma}{f_1}{\des{C}{m}{e}}{f_3}{T_1}{y}}
\end{mathpar}

\end{center}
\caption{Typing Rules (Selected, full in Appendix)}
\label{fig:typing2}
\end{figure}

Followed by the second set of typing rules in \Cref{fig:typing2}, the most notable one is rule \textsc{(TAbs)} and rule \textsc{(TApp)}. For rule \textsc{(TAbs)}, we pack all the potential function $f_1$ at this layer into the 
abstraction $\arrow{x}{f_1}{T_1}{f_2}{T_2}{y}$. This design is only possible under the separation of types and potentials, and helps us gain \emph{compositionality} in reasoning about higher-order functions. For Rule \textsc{(TApp)}, there is an instantiation from $w$ to $x$ because $e_2$ will be the actual $x$ in the context. When doing an application, we check whether the potential $f_2$ carried by the function itself and the potential $f_6$ carried by the parameter suffices or not, and then put the residual potential into the result. The rule \textsc{(TLet)} is analogous to a combination of \textsc{(TAbs)} and rule \textsc{(TApp)}. The rule \textsc{(TCons)} and rule \textsc{(TDes)} work just as rules \textsc{(TPair)}, \textsc{(TProj1)} and \textsc{(TProj2)}. The immediate consequence of such typing rules is that the AARA type system is embedable in our system, where we delayed the proof in Appendix.
\begin{theorem}[AARA embedding]
For $\Gamma \vdashpq{p}{q} e:A$, with $Pure(\cdot)$ maps a AARA type context
to a potential free \lang{} type context, $\Phi(\cdot)$ maps a AARA type  context to
the potential it carries, $h(\cdot)$ a non-trivial mapping from a AARA term to a \lang{} term, 
we have that there exists a potential function
$\Phi_x'(A)$ such that
$\pjudge{Pure(\Gamma)}{\cdot}{q+\Phi_x(A)}{\le}{\Phi_x'(A)}$ and
\[
\tjudge{Pure(\Gamma)}{\cdot}{\Phi(\Gamma)+p}{h(e)}{\Phi_x'(A)}{Type(A)}{x}.
\]
\end{theorem}

\subsection{Operational Semantics}
The evaluation rule $\ectx{e_1}{p}\hookrightarrow \ectx{e_2}{q}$ follows the line of resource-aware small-step semantics~\cite{PLDI:KWP19,knoth20liquid,LICS:GKH23,OOPSLA:WWC17} without changing anything in substance, with full version in Appendix. The evaluation contexts contain constant potentials $p$ and $q$, instead of using the potential function $f_1$ and $f_2$. The only restriction for potential $p$ and $q$ in the context is they have to be non-negative. Among all the terms, the only operation consumes potential here is the $\tick{p}{e}$, while other operations just pass the potential to the next. With the potential $p$ and $q$ in the evaluation context, and input potential function $f_1$ in the typing judgement, now we can set up our soundness lemma.

\subsection{Soundness}

Next, we are going to prove the type soundness of the calculus. The proof for progress is standard since the only evaluation rule that consume potential is the $\tick{i}{e}$, while other cases either do a substantial reduction like rule \textsc{(EProj1)} or rule \textsc{(EProj2)}, or do a substructural reduction like rule \textsc{(EPair1)} and rule \textsc{(EPair2)}. For substantial reductions other than $\tick{i}{e}$, they do not consume potential thus can always proceed. For substructural reductions, everything needed is a weakening lemma. Here we leave the full proof in Appendix.
\begin{lemma}[Progress Weakening]\label{lemma:progress-weakening-1}
If $\ectx{e}{p}\hookrightarrow\ectx{e'}{q}$, and $p\le p'$, then there exists q' such that $\ectx{e}{p'}\hookrightarrow\ectx{e'}{q'}$
\end{lemma}

\begin{theorem}[Progress]\label{lemma:progress}
If $\tjudge{\Omega}{\cdot}{f}{e}{g}{T}{x}$ and $\pjudge{\Omega}{\cdot}{f}{\le}{p}$, then e is a value or exists e' q,  $\ectx{e}{p}\hookrightarrow\ectx{e'}{q}$.
\end{theorem}

The preservation proof is technically more involved than the relatively straightforward progress theorem. The difficulty arises from the fact that variables in the potential context $\Omega$ are \emph{erased} as evaluation proceeds. As illustrated in the example in \Cref{sec:overview}, suppose $v_1$ and $v_2$ are both values, and the function $\ms{append}$ captures a variable $x$. Then we may derive a typing judgment of the form:
\begin{mathpar}\small
\tjudge{x:\tlist{\tint}}{\cdot}{p_1 + p_2}{\ms{append}\;v_1\;v_2}{3 \times \mi{length}(x)}{\tlist{\tint}}{}
\end{mathpar}

However, after evaluation, the result $v_1 \doubleplus v_2$ no longer refers to $x$. We would instead expect the result typing to be:
\begin{mathpar}\small
\tjudge{\cdot}{\cdot}{p_1 + p_2}{v_1 \doubleplus v_2}{3 \times \mi{length}(v_1)}{\tlist{\tint}}{}
\end{mathpar}

To establish preservation between such pre- and post-evaluation typings, we must relate the two potential contexts. The insight is that evaluation continues to substitute program variables with values in function applications. In this example, the variable $x$ captured by $\ms{append}$ is eventually instantiated with the value $v_1$. In our system, such substitution applies not only to term variables but also to existential variables inside potential functions. Preservation is established via continuous instantiation of potential variables. Concretely, substituting $x$ with $v_1$ in the potential expression ensures a valid post-evaluation typing.

To formalize this idea, we introduce the substitution notation $\Omega[x \mapsto v]$, which denotes replacing all occurrences of $x$ with $v$ in $\Omega$ and removing the binding $x : T$. This substitution mechanism allows us to formally state and prove the preservation theorem, which can be found in Appendix. Please note that this preservation lemma also applies to the cases where no substitution happens, where we can use a fresh variable $y$ so that $\Omega[y\mapsto v]=\Omega$, $g[y\mapsto v]=g$, $T[y\mapsto v]=T$ and $(q-f)[y\mapsto v]=q-f$.

\begin{lemma}[Value Substitution]\label{lemma:substitution-1}
If $\tjudge{\Omega}{\mctx{\Gamma_1}{x}{T_0},\,\Gamma_2}{f_1}{e}{f_2}{T_1}{y}$, and $\tjudge{\Omega}{\Gamma_1}{0}{v}{0}{T_0}{}$, then $\tjudge{\Omega[x\mapsto v]}{\Gamma_1,\Gamma_2[x\mapsto v]}{f_1[x\mapsto v]}{e[x\mapsto v]}{(f_2[x\mapsto v])}{T_1[x\mapsto v]}{y}$.
\end{lemma}

\begin{theorem}[Preservation]\label{lemma:preservation}
If $\tjudge{\Omega}{\cdot}{f}{e}{g}{T}{x}$ and $\ectx{e}{q}\hookrightarrow\ectx{e'}{q'}$, then there exists $f'$ $y$ $v$, such that $\tjudge{\Omega[y\mapsto v]}{\cdot}{f'}{e}{g[y\mapsto v]}{T[y\mapsto v]}{x}$, and we have that $\pjudge{\Omega[y\mapsto v]}{\cdot}{(q-f)[y\mapsto v]}{\le}{q'-f'}$.
\end{theorem}

%% file: evaluation.tex
\section{Case Studies}
\label{sec:eval}

In this section, we make remarks on the generality of our approach, and show the inference process for several useful examples. We have shown the complete inference process for \texttt{append} and \texttt{app\_par} in \cref{sec:overview}, and next are \texttt{traverse}, \texttt{curry}, \texttt{sort} and \texttt{map\_append}.
Due to the limit of space, we put \texttt{curry}, \texttt{sort} in the Appendix and show the inference process of \texttt{traverse} and \texttt{map\_append} here.


\subsection{List Traverse} We can first take a simple example: the traversal of $\tlist{\tint}$. The memory consumption of \texttt{traverse} accumulates as it proceeds through the list, and it will release the memory once the traversal is finished.
\begin{align*}\small
&\ms{traverse}&::&\qquad\arrow{x}{\mi{length}(x)}{\tlist{\tint}}{\mi{length}(y)}{\tlist{\tint}}{y}\\
&\ms{traverse}& \defeq &\qquad\fix{\ms{traverse}}{x}{matd(x,\{nil(x_0).nil(x_0),\\
&&&\qquad cons(x_0,x_1).\tick{1}{(\tickl{(-1)}{cons(x_0,\app{\ms{traverse}}{x_1})})}\})}
\end{align*}
We apply an eager semantics for our $\tick{p}{e}$, that is: consuming $p$ units of resource before evaluating the expression $e$. However, for list traversal where we want to release the memory resource after the evaluation, we encode the semantics using the following semantics:
\begin{align*}\small
&\tickl{p}{e}=(\lam{x}{(\tick{p}{x})})~e
\end{align*}
In such a case, the inference process for $\tickl{-1}{e}$ with the judgement
$\tjudge{\Omega}{\Gamma}{f_1}{e}{f_2}{\tlist{\tint}}{y}$ where $y$ is a local variable will be as follows:
\begin{mathpar}\small
\mprset{flushleft}
  \inferrule*[Right=(TTickn)]
{
\tjudge{\Omega}{\Gamma,\,x:\tlist{\tint}}{1}{x}{1}{\tlist{\tint}}{x}
}
{
  \inferrule*[Right=(TAbs)]
{
\tjudge{\Omega}{\Gamma,\,x:\tlist{\tint}}{0}{\tick{-1}{x}}{1}{\tlist{\tint}}{x}
}
{
  \inferrule*[Right=(TApp)]
{
\tjudge{\Omega}{\Gamma}{0}{\lam{x}{\cdots}}{0}{\arrow{x}{0}{\tlist{\tint}}{1}{\tlist{\tint}}{x}}{}
}
{
  \inferrule*[Right=(TErase)]
{\tjudge{\Omega,\,x:\tlist{\tint}}{\Gamma}{f_1}{\lam{x}{\cdots}~e}{f_2[y\mapsto x]+1}{\tlist{\tint}}{x}}
{
  \inferrule*[Right=(TRename)]
{\tjudge{\Omega}{\Gamma}{f_1}{\lam{x}{\cdots}~e}{f_2[y\mapsto x]+1}{\tlist{\tint}}{x}}
{
\tjudge{\Omega}{\Gamma}{f_1}{\lam{x}{\cdots}~e}{f_2+1}{\tlist{\tint}}{y}
}
}
}
}
}
\end{mathpar}
Next we want to type the \texttt{traverse}, let
\begin{align*}\small
T_{tvs}& \defeq{\arrow{x}{length(x)}{\tlist{\tint}}{length(y)}{\tlist{\tint}}{y}}\\
\Gamma_0 & \defeq   \ms{traverse}:T_{tvs},\,x:\tlist{\tint} \\
\Gamma_1 & \defeq \Gamma_0,\,x_0:\tint,\,x_1:\tlist{\tint}
\end{align*}
Then the inference process shows as follows:
\begin{mathpar}\small
\mprset{flushleft}
  \inferrule*[Right=(TApp)]
{\tjudge{\cdot}{\Gamma_1}{0}{\ms{traverse}}{0}{T_{tvs}}{}~\tjudge{\cdot}{\Gamma_1}{length(x_1)}{x_1}{length(x_1)}{\tlist{\tint}}{x_1}}
{
  \inferrule*[Right=(TCons)]
{\tjudge{x:\tlist{\tint}}{\Gamma}{length(x_1)}{\app{\ms{traverse}}{x_1}}{length(y)}{\tlist{\tint}}{y}
}
{
  \inferrule*[Right=(LazyTick)]
{\tjudge{\cdot}{\Gamma_0}{length(x_1)}{cons(\cdots)}{length(y)-1}{\tlist{\tint}}{y}
}
{
  \inferrule*[Right=(TTick)]
{
\tjudge{\cdot}{\Gamma_0}{length(x_1)}{\tickl{(-1)}{\cdots}}{length(y)}{\tlist{\tint}}{y}
}
{
  \inferrule*[Right=(TDes)]
{
\tjudge{\cdot}{\Gamma_0}{length(x_1)+1}{\tick{1}{\cdots}}{length(y)}{\tlist{\tint}}{y}
}
{
  \inferrule*[Right=(TFix)]
{
\tjudge{\cdot}{\Gamma_0}{length(x)}{matd(\cdots)}{length(y)}{\tlist{\tint}}{y}
}
{
\tjudge{\cdot}{\cdot}{0}{\fix{\ms{traverse}}{T_{tvs}}{\cdots}}{0}{T_{tvs}}{}
}
}
}
}
}
}
\end{mathpar}
One may notice a transition of output potential from $length(y)$ to $length(y)-1$ at rule \textsc{(TCons)}, and a change of input potential from $length(x_1)$ to $length(x)$ at rule \textsc{(TDes)}. These transition are natural according to the increase or decrease of $length()$ by the constructors and destructors.

\subsection{Map Append}
Our last example is a higher-order usage of \texttt{append}. We want to use \texttt{map} to apply the partial application of \texttt{append} to another list of list $\ell_2$
\begin{lstlisting}
map_append = map (append $\ell_1$) $\ell_2$
\end{lstlisting}
Where \texttt{map} is implemented as follows:
\begin{align*}\small
&&\fix{\ms{map}}{T_\ms{map}}{\lam{z}{\lam{w}{matd(w,\{nil(w_0).\; nil(w_0),\\&&cons(w_0,w_1).(cons(z~w_0,\ms{map}~z~w_1))\})}}}
\end{align*}
 Let the following be:
\begin{align*}\small
T_{map}&\defeq~{\arrow{z}{0}{(\arrow{y}{\mi{length}(x)}{\tlist{\tint}}{0}{\tlist{\tint}}{})}{0}{}{}}\\
&~~(\arrow{w}{length(x)\times length(w)}{\tlist{\tlist{\tint}}}{0}{\tlist{\tlist{\tint}}}{})\\
\Gamma_0 & \defeq   \ms{map}:T_{map},\,z:\arrow{y}{\mi{length}(x)}{\tlist{\tint}}{0}{\tlist{\tint}}{},\\&\qquad\,w:\tlist{\tlist{\tint}} \\
\Gamma_1 & \defeq \Gamma_0,w_0:\tlist{\tint},\,w_1:\tlist{\tlist{\tint}} \\
\Gamma_2 & \defeq l_1:\tlist{\tint},\,l_2:\tlist{\tlist{\tint}} \\
 f(x,y) & \defeq length(x)\times length(y)
\end{align*}
Then the inference process for \texttt{map} shows as follows:
\begin{mathpar}\small
\mprset{flushleft}
 \inferrule*[Right=(TApp)]
{
\tjudge{x:\tlist{\tint}}{\Gamma_1}{f(x,w_1)}{w_1}{f(x,w_1)}{\tlist{\tlist{\tint}}}{w_1}
}
{
  \inferrule*[Right=(TCons)]
{\tjudge{x:\tlist{\tint}}{\Gamma_1}{f(x,w_1)}{\ms{map}~z~w_1}{0}{\tlist{\tlist{\tint}}}{}}
{
  \inferrule*[Right=(TDes)]
{\tjudge{x:\tlist{\tint}}{\Gamma_1}{length(x)\\\times (length(w_1)+1)}{(cons(z~w_0,map~z~w_1))}{0}{\tlist{\tlist{\tint}}}{}}
{
  \inferrule*[Right=(TAbs)]
{\tjudge{x:\tlist{\tint}}{\Gamma_0}{f(x,w)}{matd(w,\cdots)}{0}{\tlist{\tlist{\tint}}}{}}
{
  \inferrule*[Right=(TFix)]
{\tjudge{x:\tlist{\tint}}{\,\ms{map}:T_{map}}{0}{\lam{x}{\lam{y}{\cdots}{}}{}}{0}{T_{map}}{}}
{\tjudge{x:\tlist{\tint}}{\cdot}{0}{\fix{\ms{map}}{T_{map}{\cdots}}{}}{0}{T_{map}}{}}
}
}
}
}
\end{mathpar}
Then we can type the \texttt{map\_append}:
\begin{mathpar}\small
\mprset{flushleft}
 \inferrule*[Right=(TApp)]
{
\tjudge{\cdot}{\Gamma_2}{f(l_1,l_2)}{l_1}{f(l_1,l_2)}{\tlist{\tlist{\tint}}}{l_1}
}
{
  \inferrule*[Right=(TApp)]
{{x:\tlist{\tint}}~|~{\Gamma_2}~|~{f(l_1,l_2)}\vdash{append~l_1}:\\\qquad[{f(x,l_2)}]{\arrow{y}{\mi{length}(x)}{\tlist{\tint}}{0}{\tlist{\tint}}{}}{}
}
{
  \inferrule*[Right=(TApp)]
{\begin{array}{ll}\tjudge{x:\tlist{\tint}}{\Gamma_2}{f(l_1,l_2)}{\ms{map}~(append~l_1)\\}{f(x,l_2)}{    \arrow{w}{f(x,w)}{\tlist{\tlist{\tint}}}{0}{\tlist{\tlist{\tint}}}{}}{}\end{array}}
{
  \inferrule*[Right=(TErase)]
{\begin{array}{ll}\tjudge{x:\tlist{\tint}}{\Gamma_2}{f(l_1,l_2)}{\ms{map}~(append~l_1)~l_2\\}{0}{\tlist{\tlist{\tint}}}{}\end{array}}
{\tjudge{\cdot}{l_1:\tlist{\tint},\,l_2:\tlist{\tlist{\tint}}}{length(l_1)\\\times length(l_2)}{\ms{map}~(append~l_1)~l_2}{0}{\tlist{\tlist{\tint}}}{}}
}
}
}
\end{mathpar}
Moveover, this program is not typable in AARA, with a proof delayed in Appendix.
\begin{theorem}
The example Map Append is not typable in AARA type system.
\end{theorem}

%% file: discussion.tex
\section{Discussion: Automated Type Checking and Inference}
\label{Se:Discussion}

Our work focuses on formalizing a dependently-typed calculus with non-affine AARA
for resource analysis; thus, developing an algorithm for automatic resource-bound
inference is out of the scope of this article.
Nevertheless, in this section, we discuss possible pathways towards automation.

\begin{figure}[thb]

\begin{center}
\begin{mathpar}
\inferrule[(AProj1)]
{y,z\notin \ms{dom}(\Omega)\cup\ms{dom}(\Gamma)\\\\
\ajudge{\Omega}{\Gamma}{f_1}{e}{f_2}{\prodT{T_1}{T_2}}{x}\\
}
{\ajudge{\Omega}{\Gamma}{f_1}{\fst{e}}{\mini{z:T_2}{f_2[x\mapsto(y,z)]}}{T_1}{y}}
\hva \and
\inferrule[(AFix)]
{z\notin \ms{dom}(\Omega)\cup\ms{dom}(\Gamma)\\\\
\ajudge{\Omega}{\mctx{\Gamma}{x}{T}}{0}{e}{f}{T}{y}\\
}
{\ajudge{\Omega}{\Gamma}{0}{\fix{x}{T}{e}}{0}{T}{z}}
\hva \and
\inferrule[(AInt)]
{ }
{\ajudge{\Omega}{\Gamma}{\fctx{0}}{i}{0}{\Int}{x}}
~~~~~~
\inferrule[(ATickn)]
{\ajudge{\Omega}{\Gamma}{f_1}{e}{f_2}{T}{x}\\\\
p < 0\\
}
{\begin{array}{ll}\ajudge{\Omega}{\Gamma}{\maxi{}{f_1+p,0}}{\\\tick{p}{e}}{f_2-\mini{}{f_1+p,0}}{T}{x}\end{array}}
\hva \and
\inferrule[(AAppp)]
{\ajudge{\Omega}{\Gamma}{f_1}{e_1}{f_2}{(\arrow{x}{f_3}{T_1}{f_4}{T_2}{y})}{z}\\
\ajudge{\Omega}{\Gamma}{f_5}{e_2}{f_6}{T_1}{w}\\
\pjudge{\mctx{\Omega}{x}{T_1}}{\Gamma}{f_3[w\mapsto x]}{\leq}{(f_2 + f_6)[w\mapsto x]}\\
}
{\ajudge{\mctx{\Omega}{x}{T_1}}{\Gamma}{f_1+f_5}{\app{e_1}{e_2}}{(f_2 + f_6-f_3)[w\mapsto x]+f_4}{T_2}{y}}
\hva \and
\inferrule[(AAppn)]
{\ajudge{\Omega}{\Gamma}{f_1}{e_1}{f_2}{(\arrow{x}{f_3}{T_1}{f_4}{T_2}{y})}{z}\\
\ajudge{\Omega}{\Gamma}{f_5}{e_2}{f_6}{T_1}{w}\\
\pjudge{\mctx{\Omega}{x}{T_1}}{\Gamma}{f_3[w\mapsto x]}{>}{(f_2 + f_6)[w\mapsto x]}\\
\mathbf{if~}w\notin \ms{dom}(\Omega)\cup\ms{dom}(\Gamma),\\
\mathbf{~then~}f=\mini{x:T_1}{(f_3-f_2-f_6)[w\mapsto x]}\\
\mathbf{else~}f=f_3-f_2-f_6
}
{\ajudge{\mctx{\Omega}{x}{T_1}}{\Gamma}{f+f_1+f_5}{\app{e_1}{e_2}}{(f+f_2 + f_6-f_3)[w\mapsto x]+f_4}{T_2}{y}}
\end{mathpar}
\end{center}
\caption{Selected Algorithmic Typing Rules}
\label{fig:select-algo-typing}
\end{figure}

\subsection{An Algorithmic System for Type Checking}
\label{Se:TypeChecking}

As stated in \cref{sec:overview} and \cref{sec:details}, we assume we can perform arithmetic reasoning in the system, but still the type system is declarative and thus non-deterministic.
That is to say, more than one rule can be applied even if we fix the context $\Omega$, $\Gamma$ and expression $e$. The non-determinism of typing rules is caused by the rule \textsc{(TRelax)}, rule \textsc{(TDrop)}, rule \textsc{(TRename)} and rule \textsc{(TErase)}. While it is hard to set up a complete algorithmic version of the type system due to the \emph{let} construct and underlying decision theory of potential functions, we can still set up a sound and determinstic algorithm using the parametrized $\mini{}{}$ and $\maxi{}{}$ functions, which take an expression together with its parameter variables as input and return the minimum or maximum integer over all possible assignments to those variables (e.g., $\mini{x:\mathsf{Int}}{x^2} = 0$).

We here set up an algorithmic version $\ajudge{\Omega}{\Gamma}{f_1}{e}{f_2}{T}{x}$, based on the observation that usages of rule \textsc{(TRelax)} and rule \textsc{(TDrop)} can be combined with usages of all the other rules.
Thus, we let $f_1$ be the minimum input potential function and $f_2$ be the maximum output potential function under the minimum input in the algorithmic typing rules.
If a rule like \textsc{(AAppn)} finds the potential provided by premise under the minimum input is not enough, it then use the rule \textsc{(TRelax)} to borrow the potentials, with the help of $\mini{}{}$ and $\maxi{}{}$ operators. Determinism is then achieved by eliminating the rule \textsc{(TRelax)} and rule \textsc{(TDrop)}, while arithmetic solver needs to deal with constraints with $\mini{}{\cdot}$ and $\maxi{}{\cdot}$.
As for the other two structural rules,
the rule \textsc{(TErase)} can always be applied at the end of the inference while rule \textsc{(TRename)} is only needed for matching the local binder in the rule \textsc{(TAbs)}. Thus by applying the strategy of not using \textsc{(TErase)} and \textsc{(TRename)} unless necessary, we get a deterministic result $\ajudge{\Omega}{\Gamma}{f_1}{e}{f_2}{T}{x}$. Then compare the needed $\tjudge{\Omega}{\Gamma}{f_3}{e}{f_4}{T}{x}$ and check $\pjudge{\Omega}{\Gamma}{f_3}{\ge}{f_1}$ and $\pjudge{\mctx{\Omega}{x}{T}}{\Gamma}{f_2+f_3}{\ge}{f_1+f_4}$ we can get the result whether a typing is feasible or not. We have set up the following theorems for the soundness of algorithm, see full proof in Appendix.

\begin{theorem}[Soundness]\label{lemma:soundness}
If $\ajudge{\Omega}{\Gamma}{f_1}{e}{f_2}{T}{x}$, then $\tjudge{\Omega}{\Gamma}{f_1}{e}{f_2}{T}{x}$.
\end{theorem}

\begin{theorem}[Determinisism]\label{lemma:Determinisism}
If $\ajudge{\Omega}{\Gamma}{f_1}{e}{f_2}{T}{x}$, $\ajudge{\Omega}{\Gamma}{f_3}{e}{f_4}{T}{x}$, then
$\pjudge{\Omega}{\Gamma}{f_1}{=}{f_3}$, and $\pjudge{\mctx{\Omega}{x}{T}}{\Gamma}{f_2}{=}{f_4}$.
\end{theorem}

\subsection{Towards Resource-Bound Synthesis for Type Inference}
\label{Se:TypeInference}

Although the type system presented in \cref{Se:TypeChecking} is algorithmic,
it still requires all functions, including both lambda abstractions and fixpoints, be
annotated with their types.
Type inference for dependently-typed systems is usually undecidable;
in our system \lang{}, even modulo constraint solving for arithmetics,
we conjecture that type inference is undecidable, because it
essentially needs to synthesize possibly-recursive potential functions
that satisfy (first-order) arithmetic constraints.
Thus, we discuss possible incomplete approaches for type inference.

\paragraph{Scenario I: all usable potential functions are provided.}
We first consider a simplified setting, where the user already provides
essential potential functions defined over inductive datatypes, such
as the $\mi{length}(\cdot)$ function for lists or the $\mi{depth}(\cdot)$ function
for trees.
The type-inference problem can then be reduced to infer arithmetic
expressions over those potential functions as potential annotations used in the type system.
For example, suppose we want to infer a type for \textsf{append}:
\[
T_\ms{append}  \defeq [f_1]_x \tlist{\tint} \to [f_2](~ [f_3]_y \tlist{\tint} \to [f_4]_z \tlist{\tint} ~)
\]
Using the algorithmic type system, we may collect some constraints:
\begin{gather*}
  \Forall{x}  f_{1,x} \ge f_{2,x}, \quad \Forall{x,y}  (x = \mi{nil}) \implies f_{3,x,y} \ge f_{4,x,y,y}, \\
  \Forall{x,y,z,h,t}  (x = \mi{cons}(h,t)) \implies f_{3,x,y} \ge 1 + f_{1,t} \wedge  f_{3,x,y} - 1 - f_{1,t} + f_{2,t}  \\
  \hspace{3em} \ge f_{3,t,y} \wedge {} f_{3,x,y} - 1 - f_{1,t} + f_{2,t} - f_{3,t,y} + f_{4,t,y,z} \ge f_{4,x,y,\mi{cons}(h,z))}  .
\end{gather*}
These constraints also indicate what set of variables $f_1,f_2,f_3,f_4$ can use, i.e., $f_1,f_2$ can use $x$, $f_3$ can use $x,y$, and $f_4$ can use $x,y,z$.
It is straightforward to verify that $f_1 = f_2 = f_4 = 0$ and $f_3 = \mi{length}(x)$ form a valid solution, which yields the type of \texttt{append} discussed in \cref{sec:overview}.
You can also verify that $f_1=f_2=0$, $f_3 = 2 \times \mi{length}(x) + \mi{length}(y)$, and $f_4 = \mi{length}(z)$ give another valid solution.

There have been template-based approaches for synthesizing unknown expressions like $f_1,f_2,f_3,f_4$.
For example, Avanzini et al.~\cite{OOPSLA:AMS20} studied modular cost analysis for imperative probabilistic programs,
but their implementation features a constraint solver named \emph{GUBS} for GUBS Upper Bound Solver,
which tries to synthesize unknown arithmetic expressions (more precisely, polynomials with $\mathbf{max}$
operators) subject to a set of (in)equalities over arithmetic expressions.
Such an approach is \emph{template}-based in the sense that the solver pre-selects possible forms for
the unknown expressions (polynomials in the GUBS case).
%
%
%

\paragraph{Scenario II: no potential functions are provided.}
We now consider the harder case: the type inference must also synthesize all the needed potential functions
defined over inductive datatypes.
Because our system \lang{} permits primitive recursion in the definition of potential functions,
we can reduce the type-inference problem to synthesis of \emph{recursive} functions, subject to
(first-order) arithmetic constraints.
Synthesis of recursive functional programs has been a popular research topic;
for example, there are techniques that synthesize recursive programs from input-output examples~\cite{PLDI:HA24,PLDI:OZ15,PLDI:YRS23,POPL:LC23}, from logical specifications~\cite{PLDI:PKS16,POPL:MNB22}, or from reference implementations~\cite{PLDI:JZP24,TOPLAS:JZX24}.
None of those approaches would be applicable
to our setting:
(i) the constraints for bound synthesis are logical, and (ii) the constraints do \emph{not} directly
specify recursive functions, but unknown arithmetic expressions over unknown recursive functions.
However, existing techniques that synthesize recursive programs from logical specifications require
that the constraints be declared directly on the unknown recursive function.

Nevertheless, recent work by Hong and Aiken~\cite{PLDI:HA24} gives us some inspiration towards developing a type-inference
algorithm for \lang{}.
Hong and Aiken~\cite{PLDI:HA24} propose to use \emph{paramorphisms}---which generalize the usual \textbf{\texttt{fold}} combinator---to reduce synthesis of recursive programs to synthesis of non-recursive programs with primitive recursion achieved by paramorphisms.
For example, the $\mi{length}(\cdot)$ function on lists can be implemented with \textbf{\texttt{fold}} without recursion:
\begin{lstlisting}
$\mi{length}$ = fun x. fold x 0 (+1)
\end{lstlisting}
Therefore, the type $T_\ms{append}$ shown above for the curried \texttt{append} function can be defined as follows:
\[
T_\ms{append} \defeq [0]_x \tlist{\tint} \to [0] (~ [\textbf{\texttt{fold}}\;x\;0\;(+1)]_y \tlist{\tint} \to [0]_z \tlist{\tint}  ~)
\]
%

%% file: relatedWork.tex
\section{Related Work}
\label{sec:relatedWork}


\paragraph{\textbf{Automatic Amortized Resource Analysis}}
%
%
AARA has been extended to support various language features,
such as higher-order functions~\cite{POPL:JHL10}, inductive datatypes~\cite{POPL:HDW17,FM:JLH09},
regular recursive types~\cite{LICS:GKH23}, mutable references~\cite{FSCD:LH17,OOPSLA:LW25}, lazy evaluation~\cite{ICFP:SVF12,ESOP:VJF15},
parallel computation~\cite{ESOP:HS15}, imperative programming~\cite{PLDI:CHS15,CAV:CHR17,ESOP:Atkey10} and
object-oriented programming~\cite{ESOP:HJ06}.
Researchers have proposed automated resource-bound inference algorithms for
lower bounds~\cite{SP:NDF17}, linear bounds~\cite{ma03static},
univariate/multivariate polynomial bounds~\cite{TLCA:HM15,POPL:HAH11,ESOP:HH10,APLAS:HH10},
logarithmic bounds~\cite{MSCS:HLO22,CAV:LMZ21}, and exponential bounds~\cite{FOSSACS:KH20}.
There is a line of studies on integrating amortized resource analysis
into theorem provers~\cite{ITP:CP15,ITP:Nipkow15,POPL:PGJ24,ESOP:MJP19}.
These techniques rely on the principle of affine/linear typing
and resource-annotated types to enable automatic bound inference, whereas
our work builds a non-affine variant of AARA and focuses on verification
of expressive resource bounds.

%
%
%
Some researchers proposed \emph{cost-aware logical frameworks}~\cite{POPL:NSG22,POPL:GNS24,LICS:NH23},
which
%
can be instantiated to carry out amortized analysis, \emph{without} requiring affine/linear typing.
Those studies focus on devising a general framework for various methodologies of resource analysis,
but it is unclear how they can make better use of existing AARA-style type systems.
As an analogy, they propose logical frameworks like LF~\cite{JACM:HHP93}, and both AARA and our work are concrete logics, such as first-order or higher-order logic.
In this sense, our work is orthogonal and focuses primarily on AARA itself, building a dependently typed AARA.
%

\paragraph{\textbf{Dependent Type Systems for Resource Analysis}}
%
%
Knoth et al.~\cite{PLDI:KWP19,knoth20liquid} combined liquid types~\cite{PLDI:RKJ08} with AARA and proposed liquid resource types (LRT), which support user-defined non-linear potential functions on inductive datatypes, while still keeping the capability of automated type checking.
LRT relies on the principle of affine typing for the resource-analysis
part; in particular, LRT requires tightly-coupled resource-annotated inductive types.
There have been dependent type systems utilizing linear typing to perform resource-bound verification, though not directly connected to AARA.
Lago and Gaboardi~\cite{LICS:LG11} presented d$\ell$PCF, which uses linear dependent types---where the dependency is achieved by an index language layer---to reason about the resource consumption of PCF expressions.
Orchard et al.~\cite{ICFP:OLE19} extended d$\ell$PCF with graded modal typing and proposed Granule, which has
many applications (e.g., reasoning about privacy), despite ones of amortized resource analysis.
Rajani et al.~\cite{vi21unify} proposed $\lambda$-amor as a unifying type theory for amortized resource analysis:
their system has a similar design to d$\ell$PCF using dependent affine typing, but introduces an indexed cost monad to enable amortization.
%
%
Compared to our work, aforementioned approaches still require linear or affine typing.

There have also been some non-affine dependent type systems for resource analysis, though not directly supporting amortized analysis.
Wang et al.~\cite{OOPSLA:WWC17} developed TiML, which extends Dependent ML~\cite{POPL:XP99} with type-level \emph{indices} that describe data-structure sizes and time complexities, thus enabling complexity analysis.
\emph{Sized types}~\cite{POPL:HPS96,phd:Vasconcelos08,ICFP:AL17} form another line of work that also uses dependent types to keep track of sizes of data structures.
Danielsson~\cite{POPL:Danielsson08} described a lightweight approach to implement a dependent cost monad in the Agda theorem prover and reduce resource-bound verification to reasoning about the monad in Agda. 
Handley et al.~\cite{POPL:HVH20} devised a different way from Re\textsuperscript{2}~\cite{PLDI:KWP19} and LRT~\cite{knoth20liquid} to extend liquid types to perform resource analysis.
%
%
These techniques diverge from AARA and it is unclear how they would support amortized analysis naturally.
\paragraph{\textbf{Closure Types and Contextual Types}}
As we discussed in \cref{Se:Introduction}, our work takes motivation from
Scherer and Hoffmann~\cite{LPAR:SH13}'s work, which mentioned the desire to precisely characterize 
the resource behavior of curried higher-order functions.
Scherer and Hoffmann proposed open closure types that can carry a set of tagged
captured variables to track data-flow.
Such an idea is not new, e.g., Leroy~\cite{TR:Leroy92} used closure types to track
type variables that cannot be further generalized during type checking.
In those type systems, arrow types are annotated with contexts;
such a notion resembles the contextual arrow types from \emph{contextual type theory}~\cite{TOCL:NPP08,PPDP:PD08,POPL:SS12}.
The difference is that, the context attached to a closure type essentially keeps track
of the bindings in the typing context of the closure, whereas a
contextual arrow type---e.g., of the form $[\Psi](A \to B)$---uses the context $\Psi$
to record meta-variables that can be used to build inhabitants.
Our potential-carrying arrow types $[f_1]_x T_1 \to [f_2]_y T_2$ are more similar 
to the closure types, in the sense that the potential functions $f_1,f_2$ should also
be well-typed under the typing context of the closure.

Recall that our typing rule for lambda applications essentially introduces existential binders.
This looks similar to a standard approach to understand the capturing behavior of closures, e.g., Minamide et al.~\cite{POPL:MMH96} presented \emph{typed closure conversion} that uses existential types to represent the captured data.
However, our setting is different because our binders bind \emph{values}, whereas typed closure conversion's existential binders bind \emph{types};
for example, Minamide et al.'s approach would assign $\ms{append}\;\ell_1$ with a type
$\exists T_{env}.\, T_{env} \times (T_{env} \to \tlist{\tint} \to \tlist{\tint}) $,
i.e., the type does \emph{not} reflect that the closure captures a list.
In contrast, our mechanism of introducing existential binders is similar to existing techniques
in refinement type systems~\cite{PLDI:PKS16,POPL:BVJ24}, which propose complex mechanisms to handle general existential types
resulted from function applications.
In our system, we devise a lightweight approach using a global potential context to track those existential binders, striking a balance between expressibility and convenience.

There is another work worth mentioning by Kahn and Hoffmann~\cite{ICFP:KH21}, who augmented AARA-style affine typing with
\emph{remainder contexts} and \emph{resource tunneling}, in order to express that a function
may return potentials that depend on the input and/or even the typing context.
For example, the resource-annotated type $\tlist{\tint^2} \to \tint \circ \tlist{\tint^1}$ describes a function that takes a list $\ell$ with $2 \times |\ell|$ units of potential as input,
returns an integer as its result, and there still remains $|\ell|$ units of potential associated with the input list $\ell$.
In our system, we can simply express such behavior by $[2 \times \mi{length}(x)]_x \tlist{\tint} \to [\mi{length}(x)] \tint$ without affine typing.

%% file: conclusion.tex
\section{Conclusion}
\label{sec:conclusion}

In this article, we introduced \lang{}, a non-affine AARA-style dependent type system for resource reasoning about higher-order programs.
We formalized our \lang{} in syntax and semantics, and provided a soundness proof based on type-level potential functions and a cost-aware semantics.
Besides the theoretical aspects, we demonstrated the expressiveness and compositionality of \lang{}'s reasoning process through several challenging examples that involve higher-order functions.
Additionally, we developed an algorithmic variant of \lang{}'s type system, which transforms typing derivation into constraint solving.
Future work includes extending \lang{} with refinement types~\cite{PLDI:FP91} or liquid types~\cite{PLDI:RKJ08}, as well as building up program-synthesis techniques for automatic resource-bound inference.

%% file: appendix.tex
\section{Appendix}
\label{sec:Appendix}

\subsection{Definitions}

\paragraph{Free Variables}
All the free variables in the type and potential functions are bound by the type context and potential context. And the free variable $\mathbf{fv}(f)$ and $\mathbf{fv}(T)$ are defined as followed:

\begin{itemize}
\item $\mathbf{fv}(f)$ is all the variables in the $f$. \newline

\item $\mathbf{fv}(\Int) = \emptyset$.\newline

\item $\mathbf{fv}(\prodT{T_1}{T_2}) = \mathbf{fv}(T_1)\cup\mathbf{fv}(T_2)$.\newline

\item $\mathbf{fv}(\arrow{x}{f_1}{T_1}{f_2}{T_2}{y}) = (\mathbf{fv}(f_1)\cup\mathbf{fv}(T_1)\cup\mathbf{fv}(f_2)\backslash y\cup\mathbf{fv}(T_2))\backslash x$.\newline

\item $\mathbf{fv}(\ind) = \bigcup\mathbf{fv}(T_i)$.\newline
\end{itemize}

\paragraph{Substitution}
Next we are going to define the substitution. The substitution happens when a application $e_1~e_2$ reduces, the variable in the $\lambda$, fixpoint or $\Lambda$ then disappear, so we need to erase them in the Type System correspondingly. The substitution actually substitution the variables in potential context into an existing value. If a variable $x$ is substituted by a value $v$, then the substitution $T[x\mapsto v]$ defined inductively as followed:
\begin{itemize}
\item $\Int[x\mapsto v] = \Int$.\newline

\item $(\prodT{T_1}{T_2})[x\mapsto v] = \prodT{T_1[x\mapsto v]}{T_2[x\mapsto v]}$.\newline

\item if $z\neq x$, then $\arrow{y}{f_1}{T_1}{f_2}{T_2}{z}[x\mapsto v] = \prod\limits_{y:[(f_1[x\mapsto v])]T_1[x\mapsto v]}\bigl[(f_2[x\mapsto v])\bigr]_zT_2[x\mapsto v]$
Otherwise: $\arrow{y}{f_1}{T_1}{f_2}{T_2}{z}[x\mapsto v] = \prod\limits_{y:[(f_1[x\mapsto v])]T_1[x\mapsto v]}\bigl[f_2\bigr]_zT_2[x\mapsto v]$ (For substitution we use, it is guaranteed that $y\neq x$)\newline

\item $\poly{y}{T_1}{T_2}[x\mapsto v] = \poly{y}{T_1[x\mapsto v]}{T_2[x\mapsto v]}$ (For substitution we use, it is guaranteed that $y\neq x$)\newline

\item $\ind[x\mapsto v] = {ind}\overrightarrow{(C,(T[x\mapsto v],m))}$.\newline
\end{itemize}

Based on the substitution $T[x\mapsto v]$, we can define the substitution on type contexts $\Gamma[x\mapsto v]$ and potential context
\begin{itemize}
\item For $y\neq x$, we have $(\mctx{\Gamma}{y}{T})[x\mapsto v]=\mctx{\Gamma[x\mapsto v]}{y}{T[x\mapsto v]}$.\newline

\item $(\mctx{\Gamma}{x}{T})[x\mapsto v]=\Gamma[x\mapsto v]$ .

\item For $y\neq x$, we have $(\mctx{\Omega}{y}{T})[x\mapsto v]=\mctx{\Omega[x\mapsto v]}{y}{T[x\mapsto v]}$ and.\newline

\item $(\mctx{\Omega}{x}{T})[x\mapsto v]=\Omega[x\mapsto v]$.\newline

\item For both type context and potential context, we have $\cdot[x\mapsto v] = \cdot$.\newline

\end{itemize}

As for substitution on terms, we need to be careful about the variables in type annotations, so the substitution would be
\begin{itemize}
\item $x[x\mapsto v] = v$

\item $y[x\mapsto v] = v$ for $y \neq x$

\item $i[x\mapsto v] =i$

\item $\llam{y}{e}{T}[x\mapsto v] = \llam{y}{e[x\mapsto v]}{T[x\mapsto v]}$

\item $\app{e_1}{e_2}[x\mapsto v] = \app{e_1[x\mapsto v]}{e_2[x\mapsto v]}$

\item $\pair{e_1}{e_2}[x\mapsto v] = \pair{e_1[x\mapsto v]}{e_2[x\mapsto v]}$

\item $(\fst e)[x\mapsto v] = \fst e[x\mapsto v]$

\item $(\snd e)[x\mapsto v] = \snd e[x\mapsto v]$

\item $\fix{y}{T}{e}[x\mapsto v] = \fix{y}{T[x\mapsto v]}{e[x\mapsto v]}$ (For the substitution we use, it is guaranteed that $y\neq x$).

\item $\tick{i}{e}[x\mapsto v] = \tick{i}{e[x\mapsto v]}$

\item $\op{i}{e}[x\mapsto v] = \op{i}{e[x\mapsto v]}$

\item $\con{i}{e_0}{e_1}{e_{m_i}}[x\mapsto v] = \con{i}{e_0[x\mapsto v]}{e_1[x\mapsto v]}{e_{m_i}[x\mapsto v]}$

\item $\des{C}{m}{e}[x\mapsto v] = \des{C}{m}{e[x\mapsto v]}$
\end{itemize}

\subsection{Curry and Uncurry Functions}

An interesting higher-order example involves typing the \emph{curry} and \emph{uncurry} functions. These functions convert between curried and uncurried forms of higher-order functions. They can be defined as follows:
\begin{align*}\small
&\ms{curry} &:: &\qquad\arrow{}{0}
    {\left(\arrow{z}{f_1}{T_1 \times T_2}{f_2}{T_3}{w}\right)}
    {0}
    {\left(\arrow{x}{0}{T_1}{0}
        {\left(\arrow{y}{f_1[z \mapsto (x,y)]}{T_2}{f_2}{T_3}{w}\right)}{}\right)}{}\\
&\ms{curry} &\defeq&\qquad\lam{f}{}{\lam{x}{}{\lam{y}{}{f~(x,y)}}} \\[1ex]
&\ms{uncurry}& :: &\qquad\arrow{}{0}
    {\left(\arrow{x}{f_1}{T_1}{0}
        {\left(\arrow{y}{f_2}{T_2}{f_3}{T_3}{w}\right)}{}\right)}
    {0}
    {\\&&&\qquad\left(\arrow{z}{f_1[x \mapsto \fst z] + f_2[x \mapsto \snd z]}{T_1 \times T_2}{f_3}{T_3}{w}\right)}{}\\
&\ms{uncurry}& \defeq&\qquad \lam{f}{}{\lam{z}{}{f~(\fst z)~(\snd z)}}
\end{align*}

The type of $\ms{curry}$ reads as follows:  
If the input function consumes $f_1$ units of resource and produces $f_2$ units of potential, then the curried version initially consumes 0 resources when given the first argument $x$, but will require $f_1$ units of potential upon receiving the second argument $y$. The return value carries $f_2$ units of potential.

Let:
\begin{align*}\small
T_{\ms{fun0}} &\defeq \arrow{z}{f_1}{T_1 \times T_2}{f_2}{T_3}{w} \\
\Gamma_0 &\defeq f : T_{\ms{fun0}},\, x : T_1,\, y : T_2
\end{align*}

The type inference proceeds as follows:

\begin{mathpar}\small
\mprset{flushleft}
\inferrule*[Right=(TApp)]
{
  \tjudge{\cdot}{\Gamma_0}{0}{f}{0}{T_{\ms{fun0}}}{} \\
  \tjudge{\cdot}{\Gamma_0}{f_1[z \mapsto (x, y)]}{(x, y)}{f_1}{T_3}{z}
}
{
  \inferrule*[Right=(TAbs)]
  {
    \tjudge{z : T_1 \times T_2}{\Gamma_0}{f_1[z \mapsto (x, y)]}{f~(x, y)}{f_2}{T_3}{w}
  }
  {
    \inferrule*[Right=(TAbs)]
    {
      \tjudge{z : T_1 \times T_2}{f : T_{\ms{fun0}},\, x : T_1}{0}
      {\lam{y}{}{f~(x, y)}}
      {0}
      {\arrow{y}{f_1[z \mapsto (x, y)]}{T_2}{f_2}{T_3}{w}}{}
    }
    {
      \inferrule*[Right=(TAbs)]
      {
        \tjudge{z : T_1 \times T_2}{f : T_{\ms{fun0}}}{0}
        {\lam{x}{}{\lam{y}{}{f~(x, y)}}}
        {0}
        {\arrow{x}{0}{T_1}{0}{\arrow{y}{f_1[z \mapsto (x, y)]}{T_2}{f_2}{T_3}{w}}{}}{}
      }
      {
        \tjudge{z : T_1 \times T_2}{\cdot}{0}
        {\ms{curry}}
        {0}
        {\arrow{}{0}{T_{\ms{fun0}}}{0}{\arrow{x}{0}{T_1}{0}{\arrow{y}{f_1[z \mapsto (x, y)]}{T_2}{f_2}{T_3}{w}}{}}{}}{}
      }
    }
  }
}
\end{mathpar}

Note that the final typing judgment still contains an existential variable $z$ in the potential function. This is intentional: since the potential function $f_2$ may refer to the function argument $z$, we preserve this binding. However, if $f_2$ does not actually depend on $z$, we may erase it using rule \textsc{(TErase)}.

We now perform a similar type inference for the \texttt{uncurry} function. Let:
\begin{align*}\small
T_{\ms{fun1}} &\defeq \arrow{x}{f_1}{T_1}{0}{\left(\arrow{y}{f_2}{T_2}{f_3}{T_3}{w}\right)}{} \\
\Gamma_1 &\defeq f : T_{\ms{fun1}},\, z : T_1 \times T_2
\end{align*}
We derive its type as follows:

\begin{mathpar}\small \mprset{flushleft} \inferrule*[Right=(TApp)] {\tjudge{\cdot}{\Gamma_0}{0}{f}{0}{T_{\ms{fun1}}}{}\\\tjudge{\cdot}{\Gamma_0}{f_1[x\mapsto \fst z]}{\fst z}{f_1}{T_1}{x}} { \inferrule*[Right=(TApp)] { \tjudge{x:T_1}{\Gamma_0}{f_1[x\mapsto \fst z]}{f~(\fst z)}{f_1}{\arrow{y}{f_2}{T_2}{f_3}{T_3}{w}}{}\\\cdots } { \inferrule*[Right=(TAbs)] {\tjudge{x:T_1,\,y:T_2}{\Gamma_0}{f_1[x\mapsto \fst z]+f_2[x\mapsto \snd z]}{\lam{z}{}{f~(\fst z)~(\snd z)}}{0}{T_3}{w} } { \inferrule*[Right=(TAbs)] {\tjudge{x:T_1,\,y:T_2}{f:T_{\ms{fun1}}}{0}{f~(\fst z)~(\snd z)}{0}{\arrow{z}{f_1[x\mapsto \fst z]+f_2[x\mapsto \snd z]}{T_1\times T_2}{f_3}{T_3}{w}}{} } { \tjudge{x:T_1,\,y:T_2}{\cdot}{0}{\ms{uncurry}}{0}{\arrow{}{0}{ T_{\ms{fun1}}}{0}{(\arrow{z}{f_1[x\mapsto \fst z]+f_2[x\mapsto \snd z]}{T_1\times T_2}{f_3}{T_3}{w})}{}}{} } } } } \end{mathpar}

Just like in the $\ms{curry}$ case, the final type for $\ms{uncurry}$ contains dependencies on potential variables $x$, $y$, and $z$. However, if the final potential function $f_3$ does not actually depend on any of these variables, we may safely erase them using the rule \textsc{(TErase)}. This results in a cleaner and more general type.

\subsection{Insertion Sort} Next, we can take a more complicated example for \lang{}. In this example, we will type the classic insertion sort. The result type in our implementation is:
\[\small
T_{srt}\defeq\poly{i}{\Int}{\arrow{z}{\frac{1}{2}(length(z)^2+(3+2i)\times length(z))}{\tlist{\tint}}{i\times length(l)}{\tlist{\tint}}{l}}
\]
One can see the power as well as the limitations of \lang{}. The separation of type and potentials enables us to precisely describe the cost behavior of \texttt{sort}. However, we lose track of the information of $length(z)=length(l)$, which suggests the sorting does not change the length of a list.
Such invariant tracking is only possible with approaches like refinement types, which will be addressed in our future work. In the current case, a polymorphism $i$ is added as a remedy to exchange potentials of $length(z)$ to $length(l)$, so that one can later instantiate $i$ with $0$ to get the classic sorting bound as:
\[
T_{srt0}\defeq\arrow{z}{\frac{1}{2}(length(z)^2+3\times length(z))}{\tlist{\tint}}{0}{\tlist{\tint}}{l}
\]
Back to our example, we take the sorting on $\tlist{\tint}$ as the case. We want to first show that how to type the \texttt{insert}, and then go to \texttt{sort}. We implement \texttt{insert} and \texttt{sort} as follows:
\begin{align*}\small
&&\fix{\ms{insert}}{T_{ins}}{\Lam{i}{\lam{x}{\lam{y}{\tick{1}{matd(y,\{nil(x_0).cons(x,x_0),cons(x_0,x_1).(matd(x<x_0),\\
&&\{false(x_0').cons(x_0,\ms{insert}~i~x~x_1),true(x_0').cons(x,(cons(x_0,x_1)))\})\})}}}}{\Int}}\\
&&\fix{\ms{sort}}{T_{srt}}{\Lam{i}{\lam{z}{matd(z,\{nil(z_0).y,cons(z_0,z_1).\tick{1}{\ms{insert}~i~z_0~(\ms{sort}~(i+1)~z_1)}\})}}{\Int}}\end{align*}
Where $<$ is function of the type $\arrow{}{0}{\Int}{0}{(\arrow{}{0}{\Int}{0}{\mathsf{Bool}}{})}{}$.
Here we encode type $Bool$  as an inductive type containing only two constructors $false$ and $true$.
then we let the following be:
\begin{align*}\small
T_{ins}\defeq~&\poly{i}{\Int}{\arrow{x}{0}{\Int}{0}{(\arrow{y}{(i+1)\times (length(y)+1)}{\tlist{\tint}}{i\times length(l)}{\tlist{\tint}}{l})}{}}\\
\Gamma_0\defeq~&\ms{insert}:T_{srt},\,i:\Int,\,x:\Int,\,y:{\tlist{\tint}}\\
\Gamma_1\defeq~&\Gamma_0,\,x_0:{\Int},\,x_1:{\tlist{\tint}},\,x_0':Bool\\
\Gamma_1'\defeq~&\Gamma_0,\,x_0:{\Int},\,x_1:{\tlist{\tint}}
\end{align*}
Thus we get the inductive branch of $false(x_0').cons(x_0,\ms{insert}~i~x~x_1)$ types as follows:
\begin{mathpar}\small
\mprset{flushleft}
  \inferrule*[Right=(TApp)]
{
\tjudge{\cdot}{\Gamma_1}{0}{\ms{insert}}{0}{T_{ins}}{}
\\
\tjudge{\cdot}{\Gamma_1}{0}{i}{0}{\Int}{}
\\
\tjudge{\cdot}{\Gamma_1}{0}{x}{0}{\Int}{}
}
{
  \inferrule*[Right=(TApp)]
{\tjudge{\cdot}{\Gamma_1}{0}{\app{\app{\ms{insert}}{i}}{x}}{0}{\arrow{y}{(i+1)\times (length(y)+1)}{\tlist{\tint}}{i\times length(y)}{\tlist{\tint}}{l}}{}}
{
  \inferrule*[Right=(TCons)]
{
\tjudge{\cdot}{\Gamma_1}{(i+1)\times (length(x_1)+1)}{\app{\app{\app{\ms{insert}}{i}}{x}}{x_1}}{i\times length(l)}{\tlist{\tint}}{l}
}
{
\tjudge{\cdot}{\Gamma_1}{(i+1)\times (length(x_1)+1)}{cons(x_0,\app{\app{\app{\ms{insert}}{i}}{x}}{x_1}) }{i\times (length(l)-1)}{\tlist{\tint}}{l}
}
}
}
\end{mathpar}
For the other branch $true(x_0').cons(x,(cons(x_0,x_1)))$ we have:
\begin{mathpar}\small
\mprset{flushleft}
  \inferrule*[Right=(TCons)]
{
\tjudge{\cdot}{\Gamma_1}{0}{x_0}{0}{\Int}{}
\\
\tjudge{\cdot}{\Gamma_1}{(i+1)\times (length(x_1)+1)}{x_1}{(i+1)\times (length(x_1)+1)}{\Int}{x_1}
}
{
  \inferrule*[Right=(TCons)]
{\tjudge{\cdot}{\Gamma_1}{(i+1)\times (length(x_1)+1)}{cons(x_0,x_1)}{(i+1)\times length(l)}{\Int}{l}}
{
  \inferrule*[Right=(TDrop)]
{
\tjudge{\cdot}{\Gamma_1}{(i+1)\times (length(x_1)+1)}{cons(x,(cons(x_0,x_1)))}{(i+1)\times (length(l)-1)}{\Int}{l}
}
{
\tjudge{\cdot}{\Gamma_1}{(i+1)\times (length(x_1)+1)}{cons(x,(cons(x_0,x_1)))}{i\times (length(l)-1)}{\tlist{\tint}}{l}
}
}
}
\end{mathpar}
Combined together we can have:
\begin{mathpar}\small
\mprset{flushleft}
 \inferrule*[Right=(TDes)]
{\cdots\\\cdots}
{\tjudge{\cdot}{\Gamma_1'}{(i+1)\times (length(x_1)+1)}{matd(x<x_0,\cdots)}{i\times (length(l)-1)}{\tlist{\tint}}{l}}
\end{mathpar}
Then on the upper level branch $nil(x_0).cons(x,x_0)$ we have:
\begin{mathpar}\small
\mprset{flushleft}
  \inferrule*[Right=(TCons)]
{\tjudge{\cdot}{\Gamma_0,\, x_0:nil}{(i+1)\times length(x_0)}{x_0}{(i+1)\times length(x_0)}{\tlist{\tint}}{x_0}
}
{
  \inferrule*[Right=(TDrop)]
{
\tjudge{\cdot}{\Gamma_0,\, x_0:nil}{(i+1)\times length(x_0)}{cons(x,x_0)}{(i+1)\times (length(l)-1)}{\tlist{\tint}}{l}}
{
\tjudge{\cdot}{\Gamma_0,\, x_0:nil}{(i+1)\times length(x_0)}{cons(x,x_0)}{i\times (length(l)-1)}{\tlist{\tint}}{l}}
}
\end{mathpar}
Thus we have\newline
\begin{mathpar}\small
\mprset{flushleft}
 \inferrule*[Right=(TDes)]
{
\cdots\\\cdots
}
{
  \inferrule*[Right=(TTickp)]
{\tjudge{\cdot}{\Gamma_0}{(i+1)\times length(y)}{{matd(y,\cdots)}}{i\times (length(l)-1)}{\tlist{\tint}}{l}}
{
  \inferrule*[Right=(TRelax)]
{\tjudge{\cdot}{\Gamma_0}{(i+1)\times length(y)+1}{\tick{1}{matd(y,\cdots)}}{i\times (length(l)-1)}{\tlist{\tint}}{l}}
{
  \inferrule*[Right=(TAbs+TPabs)]
{\tjudge{\cdot}{\Gamma_0}{(i+1)\times (length(y)+1)}{\tick{1}{matd(y,\cdots)}}{i\times length(l)}{\tlist{\tint}}{l}}
{
  \inferrule*[Right=(TFix)]
{\tjudge{\cdot}{\,\ms{insert}:T_{ins}}{0}{\Lam{x}{\lam{x}{\lam{y}{\cdots}{}}{}}{\Int}}{0}{T_{ins}}{}}
{\tjudge{\cdot}{\cdot}{0}{\fix{\ms{insert}}{T_{ins}}{\cdots}}{0}{T_{ins}}{}}
}
}
}
}
\end{mathpar}
Then we can go the \texttt{sort}. Let the following be:

\begin{align*}\small
T_{srt}\defeq~&\poly{i}{\Int}{\arrow{z}{\frac{1}{2}(length(z)^2+(3+2i)\times length(z))}{\tlist{\tint}}{i\times length(l)}{\tlist{\tint}}{l}}\\
\Gamma_2\defeq~&sort:T_{srt},i:\Int,\,\,z:\tlist{\tint},z_0,\,\Int,\,z_1:\tlist{\tint}\\
f(z)\defeq~&\frac{1}{2}(length(z)^2+(5+2i)\times length(z)+2i+2)\\
g(z)\defeq~&\frac{1}{2}(length(z)^2+(3+2i)\times length(z))
\end{align*}
Then the inference process will be
\begin{mathpar}\small
\mprset{flushleft}
  \inferrule*[Right=(TApp)]
{
\tjudge{\cdot}{\Gamma_2}{0}{\app{sort}{(i+1)}}{0}{\arrow{z}{f(z)-i-1}{\tlist{\tint}}{(i+1)\times length(l)}{\tlist{\tint}}{l}}{}
\\
\tjudge{\cdot}{\Gamma_2}{f(z_1)}{z_1}{f(z_1)}{\tlist{\tint}}{z_1}}
{
  \inferrule*[Right=(TApp)]
{\tjudge{\cdot}{\Gamma_2}{f(z_1)}{\app{\app{sort}{(i+1)}}{z_1}}{(i+1)\times (length(l)+1)}{\tlist{\tint}}{l}}
{
  \inferrule*[Right=(TTickp)]
{\tjudge{\cdot}{\Gamma_2}{f(z_1)}{\app{\app{insert~i}{z_0}}{\app{\app{sort}{(i+1)}}{z_1}}}{i\times length(l)}{\tlist{\tint}}{l}}
{
  \inferrule*[Right=(TDes)]
{\tjudge{\cdot}{\Gamma_2}{f(z_1)+1}{\tick{1}{\cdots}}{i\times(length(l)}{\tlist{\tint}}{l}}
{
  \inferrule*[Right=(TAbs)]
{\tjudge{\cdot}{sort:T_{srt},i:\Int,\,\,z:\tlist{\tint}}{g(z)}{matd(z,\cdots)}{i\times (length(l))}{\tlist{\tint}}{l}}
{
  \inferrule*[Right=(TPabs)]
{\tjudge{\cdot}{sort:T_{srt},i:\Int}{0}{\lam{z}{\cdots}{}}{0}{\arrow{z}{g(z)}{\tlist{\tint}}{i\times length(l)}{\tlist{\tint}}{l}}{}}
{
  \inferrule*[Right=(TFix)]
{\tjudge{\cdot}{sort:T_{srt},}{0}{\Lam{i}{\Int}{\cdots}{}}{0}{T_{srt}}{}}
{\tjudge{\cdot}{\cdot}{0}{\fix{sort}{T_{srt}}{\cdots}}{0}{T_{srt}}{}}
}
}
}
}
}
}
\end{mathpar}
With an instantiation of \texttt{sort} by $0$, we can get our \texttt{sort} as $\arrow{z}{\frac{1}{2}(length(z)^2+3\times length(z))}{\tlist{\tint}}{0}{\tlist{\tint}}{l}$.

\subsection{Progress}

\begin{lemma}[Progress Weakening]\label{lemma:progress-weakening}
If $\ectx{e}{p}\hookrightarrow\ectx{e'}{q}$, and $p\le p'$, then there exists q' such that $\ectx{e}{p'}\hookrightarrow\ectx{e'}{q'}$
\end{lemma}
\begin{proof}
By direct induction on $\ectx{e}{p}\hookrightarrow\ectx{e'}{q}$, the only case we need to care about is the tick. But it is again easy to prove by the linearity of addition.
\end{proof}

Then we can prove the progress lemma.
\begin{theorem}[Progress]\label{lemma:progress}
If $\tjudge{\Omega}{\cdot}{f}{e}{g}{T}{x}$ and $\pjudge{\Omega}{\cdot}{f}{\le}{p}$, then e is a value or exists e' q,  $\ectx{e}{p}\hookrightarrow\ectx{e'}{q}$.
\end{theorem}
\begin{proof}

By induction on the inductive hypothesis $\tjudge{\Omega}{\cdot}{f}{e}{g}{T}{x}$.

\begin{itemize}
\item Case \textsc{TAbs}: $e$ is already a value.\newline

\item Case \textsc{TPAbs}: $e$ is already a value.\newline

\item Case \textsc{TApp}:  By induction on $\tjudge{\Omega}{\cdot}{f}{\app{e_1}{e_2}}{g}{T}{y}$, suppose the induction hypothesis breaks into the following:\newline

\begin{enumerate}
\item $H_1$: $\tjudge{\Omega'}{\Gamma}{f_1}{e_1}{f_2}{\arrow{x}{f_3}{T_1}{f_4}{T_2}{y}}{z}$
\item $H_2$: $\tjudge{\Omega'}{\Gamma}{f_5}{e_2}{f_6}{T_1}{w}$
\item $H_3$: $\pjudge{\mctx{\Omega'}{x}{T_1}}{\Gamma}{f_3[w\mapsto x]}{\leq}{(f_2 + f_6)[w\mapsto x]}$
\item $Constraint_1$: $\Omega =\mctx{\Omega'}{x}{T_1}$
\item $Constraint_2$: $\Gamma = \cdot$
\item $Constraint_3$: $f_1+f_5 = f$
\item $Constraint_4$: $(f_2 + f_6-f_3)[w\mapsto x]+f_4 = g$
\item $Constraint_5$: $T_2 = T$\newline
\end{enumerate}

First by applying the inductive hypothesis on $H_1$, we have either $e_1$ is a value, or for all $p_1$ such that $\pjudge{\Omega}{\cdot}{f_1}{\le}{p_1}$, there exists $e_1'$ and $q_1$ such that $\ectx{e_1}{p_1}\hookrightarrow\ectx{e_1'}{q_1}$.\newline

If $e_1$ is nnot a value, then since we have $\pjudge{\Omega}{\cdot}{f}{\le}{p}$, thus we have $\pjudge{\Omega}{\cdot}{f_1}{\le}{p}$. By the reduction rule \textsc{EAppl}, we prove the case.\newline

Using the same technique we can also prove the case for $e_1$ is a value and $e_2$ is not. Finally is  the case $e_1$ and $e_2$ are all values, because $e_1$ is of type
$\arrow{x}{f_3}{T_1}{f_4}{T_2}{y}$, so it has to be in the form of $\lam{x}{e}{:\arrow{x}{f_3}{T_1}{f_4}{T_2}{y}}$. Then by the reduction rule \textsc{EApp}, we prove the case.\newline

\item Case \textsc{TPapp}: Same as \textsc{TApp}.\newline

\item Case \textsc{TLet}: Same as \textsc{TApp}.\newline

\item Case \textsc{TVar}: type context is empty, so it is not well-typed.\newline

\item Case \textsc{TPair}: Same as case \textsc{TApp1}, by induction on $\tjudge{\Omega}{\cdot}{f}{\pair{e_1}{e_2}}{g}{T}{x}$, and discuss the three cases where $e_1$ is not a value, $e_1$ is a value while $e_2$ is not, and both $e_1$ and $e_2$ are both values.\newline

\item Case \textsc{TInt}: $e$ is already a value.\newline

\item Case \textsc{TRelax}: By induction on $\tjudge{\Omega}{\cdot}{f}{e}{g}{T}{x}$, suppose the induction hypothesis breaks into the following:\newline

\begin{enumerate}
\item $H_1$: $\tjudge{\Omega}{\Gamma}{f_1}{e}{f_2}{T}{x}$
\item $H_2$: $\pjudge{\Omega}{\Gamma}{f_3}{\ge}{0}$
\item $Constraint_1$: $\Omega = \Omega$
\item $Constraint_2$: $\Gamma = \cdot$
\item $Constraint_3$: $f_1+f_3 = f$
\item $Constraint_4$: $f_2+f_3 = g$\newline
\end{enumerate}

By the inductive hypothesis on $H_1$, we have either $e$ is a value, or for all $p_1$ such that $\pjudge{\Omega}{\cdot}{f_1}{\le}{p_1}$, there exists $e'$ and $q_1$ such that $\ectx{e}{p_1}\hookrightarrow\ectx{e'}{q_1}$.\newline

If $e$ is a value, then we prove the case. If $e$ is not a value, then since we have $\pjudge{\Omega}{\cdot}{f}{\le}{p}$, thus we have $\pjudge{\Omega}{\cdot}{f_1}{\le}{p}$, therefore $\ectx{e}{p_1}\hookrightarrow\ectx{e'}{q_1}$ thus we prove the case.\newline

\item Case $\sf{TProj}_1$ and $\sf{TProj}_2$: By induction on $\tjudge{\Omega}{\cdot}{f}{\fst e}{g}{T}{x}$ or $\tjudge{\Omega}{\cdot}{f}{\snd e}{g}{T}{x}$, we have the induction hypothesis that either $e$ is a value or not. If $e$ is a not value, then apply the same technique in Case TApp we get the proof. If $e$ is a value, then $e$ has to be in the form of a pair $\pair{v_1}{v_2}$, then by the reduction rule \textsc{EProj1} and \textsc{EProj2} we obtain the proof.\newline

\item Case \textsc{TTickn} and \textsc{TTickp}: By induction on $\tjudge{\Omega}{\cdot}{f}{\tick{q}{e}}{g}{T}{y}$, we get for \textsc{TTickn}, it always get stepped because $q<0$. For \textsc{TTickp}, it is guaranteed that $\pjudge{\Omega}{\cdot}{f}{\ge}{q}$. Thus it can be safely stepped.\newline

\item Case \textsc{TFix}: Fix-points can always be reduced without consuming potentials.\newline

\item Case \textsc{TCons}: $\tjudge{\Omega}{\cdot}{f}{\con{i}{e_0}{e_1}{e_{m_i}}}{g}{T}{x}$  by induction we obtain the cases where for $j\in [0,m_i]$, and for all $j'\in[0,j-1]$, $e_{j'}$ is a value while $e_j$ is not, or the case that forall $j''\in[0,m_i]$, $e_{j''}$ is a value. For the former cases, we can apply the same technique in Case \textsc{TApp} and then we get the proof. For the later case, it is already a value.\newline

\item Case \textsc{TDes}: By direct induction on the judgement, we have the induction hypothesis that either $e_0$ is a value or not. If $e_0$ is a not value, then apply the same technique in case \textsc{TApp} we get the proof. If $e_0$ is a value, then it has to be in the form of $\con{i}{v_0}{v_1}{v_{m_i}}$, then by the reduction rule \textsc{ECase} we obtain the proof.\newline

\item Case \textsc{TOp}: Same as \textsc{TCons}.\newline

\item Case \textsc{TDrop}, \textsc{TErase} and \textsc{TRename}: We directly prove the case.
\end{itemize}

Thus we prove the progress.
\end{proof}

\subsection{Preservation}

Next we can prove the preservation lemma, we first need to prove some auxiliary lemmas.\newline
\begin{lemma}[Value Strengthening]\label{lemma:value-strengthening}
If $\tjudge{\Omega}{\Gamma}{f_1}{v}{f_2}{T}{x}$, then $\tjudge{\Omega}{\Gamma}{0}{v}{0}{T}{x}$.
\end{lemma}
\begin{proof}
By direct induction on $\tjudge{\Omega}{\Gamma}{f_1}{v}{f_2}{T}{x}$ we obtain the proof.
\end{proof}

\begin{lemma}[Potentials of Values]\label{lemma:value-potentials}
If $\tjudge{\Omega}{\Gamma}{f_1}{v}{f_2}{T}{x}$, then $\pjudge{\Omega}{\Gamma}{f_2[x\mapsto v]}{\le}{f_1}$.
\end{lemma}
\begin{proof}
By direct induction on the inductive hypothesis $\tjudge{\Omega}{\Gamma}{f_1}{v}{f_2}{T}{x}$, we prove the Rule \textsc{(TAbs)}, \textsc{(TPAbs)}, and  \textsc{(TInt)} cases because here $f_1=f_2=0$. For Rule \textsc{(TErase)}, rule \textsc{(TRelax)}, rule \textsc{(TRename)} and  \textsc{(TDrop)}, it is trivial. For Rule \textsc{(TPair)} and  \textsc{(TCons)}, since by induction we have for all the subcases that $\pjudge{\Omega}{\Gamma}{f_2'[x'\mapsto v']}{\le}{f_1'}$. Finally add up these inequations we prove the case.
\end{proof}

\begin{lemma}[Value Relaxing]\label{lemma:value-relax}
If $\pjudge{\Omega}{\Gamma}{f}{}{}$ and $x\notin \Omega\cup\Gamma$
, then we have $\tjudge{\Omega}{\Gamma}{f[x\mapsto v]}{v}{f}{T}{x}$.
\end{lemma}
\begin{proof}
By direct induction on $v$, we have the following cases:
\begin{itemize}
\item Case $v$ is the value don't have structures to decompose, where $v$ is $i$, $\lam{x}{e}{:\arrow{x}{f_1}{T_1}{f_2}{T_2}{y}}$,or  $\Lam{X}{T}{e}$.\newline

We notice that for every value here, we have $\tjudge{\Omega}{\Gamma}{0}{v}{0}{T}{x}$, then by \textsc{TRelax} we have $\tjudge{\Omega}{\Gamma}{f[x\mapsto v]}{v}{f[x\mapsto v]}{T}{x}$. Since potential function can not cotain structures for such $T$ type, thus $\pjudge{\mctx{\Omega}{x}{T}}{\Gamma}{f[x\mapsto v]}{=}{f}$ because the actual value of $x$ doesn't matter here. Then by the rule \textsc{TDrop} we prove the case.\newline

\item Case $\pair{v_1}{v_2}$ or $\con{i}{v_0}{v_1}{v_{m_i}}$. They are technically the same, we go with the proof for $\pair{v_1}{v_2}$.
Notice that for the final function $f_2[x_1\mapsto \fst{x}]+f_4[x_2\mapsto \snd{x}]$ in the type judgement $\tjudge{\Omega}{\Gamma}{f_1+f_3}{\pair{e_1}{e_2}}{f_2[x_1\mapsto \fst{x}]+f_4[x_2\mapsto \snd{x}]}{\prodT{T_1}{T_2}}{x}$, we can always separate the variable $\fst{X}$ from $\snd{X}$, thus by induction we prove the case.
\end{itemize}

\end{proof}

\begin{lemma}[Value Substitution on one variable]\label{lemma:substitution}
If $\tjudge{\Omega}{\mctx{\Gamma_1}{x}{T_0},\,\Gamma_2}{f_1}{e}{f_2}{T_1}{y}$, and $\tjudge{\Omega}{\Gamma_1}{0}{v}{0}{T_0}{}$, then $\tjudge{\Omega[x\mapsto v]}{\Gamma_1,\Gamma_2[x\mapsto v]}{f_1[x\mapsto v]}{e[x\mapsto v]}{(f_2[x\mapsto v])}{T_1[x\mapsto v]}{y}$.
\end{lemma}
\begin{proof}

By direct induction on $\tjudge{\Omega_1}{\mctx{\Gamma_1}{x}{T_0},\,\Gamma_2}{f_1}{e}{f_2}{T_1}{y}$, we have the following Cases:
\begin{itemize}
\item Case \textsc{TPAbs}, \textsc{TAbs}, and \textsc{TDes}: Since all newly introduced variables are appended to the tail of $\Gamma_2$, Thus we can apply the induction hypothesis and obtain the proof (Also note that the type annotation changed accordingly).\newline

\item Case \textsc{TApp}, \textsc{TPair}, \textsc{TRename}, \textsc{TErase}, \textsc{TProj1}, \textsc{TProj2}, \textsc{TTickn}, \textsc{TTickp}, \textsc{TOp},  \textsc{TLet}, \textsc{TPApp}, and \textsc{TCons}. Since the substitution keeps the equality and linear order by definition, thus we prove the case.\newline

\item Case \textsc{TRelax}, the only minor case we need to care is that the local binder $x$ is exactly the $x$ we are substituting. But by lemma \ref{lemma:value-relax} we can prove the case.

\item Case \textsc{TInt}: We directly obtain the proof.\newline

\item Case \textsc{TVar}: by lemma \ref{lemma:value-strengthening} we prove the case.\newline
\end{itemize}
\end{proof}

\begin{lemma}[Potential instantiation on one variable]\label{lemma:instantiation}
If $\tjudge{\Omega}{\cdot}{f_1}{e}{f_2}{T}{x}$, $y:T'\in\Omega$ and $\tjudge{\Omega}{\cdot}{0}{v}{0}{T}{x}$, then $\tjudge{\Omega[y\mapsto v]}{\cdot}{f_1[y\mapsto v]}{e[y\mapsto v]}{f_2[y\mapsto v]}{T[y\mapsto v]}{x}$.
\end{lemma}
\begin{proof}
First we add the $y:T$ to the context we get $\tjudge{\Omega}{y:T}{f_1}{e}{f_2}{T}{x}$. This is safe for that $y$ shares the same type in both potential context and type context. Then by lemma \ref{lemma:substitution} we prove the case.
\end{proof}

Then we can go to the preservation.
\begin{theorem}[Preservation]\label{lemma:preservation}
If $\tjudge{\Omega}{\cdot}{f}{e}{g}{T}{x}$ and $\ectx{e}{q}\hookrightarrow\ectx{e'}{q'}$, then there exists $f'$ $y$ $v$, such that $\tjudge{\Omega[y\mapsto v]}{\cdot}{f'}{e}{g[y\mapsto v]}{T[y\mapsto v]}{x}$, and we have that $\pjudge{\Omega[y\mapsto v]}{\cdot}{(q-f)[y\mapsto v]}{\le}{q'-f'}$.
\end{theorem}
\begin{proof}
By direct induction on $\tjudge{\Omega}{\cdot}{f}{e}{g}{T}{x}$, then we have
\begin{itemize}
\item Case \textsc{TAbs}, \textsc{TInt}, \textsc{TPAbs}, \textsc{TFix}: $e$ is already a value, thus it can't be further reduced, so $\ectx{e}{q}\hookrightarrow\ectx{e'}{q'}$ causes a contradiction.\newline

\item Case \textsc{TApp}: By induction on $\tjudge{\Omega}{\cdot}{f}{\app{e_1}{e_2}}{g}{T}{y}$, we suppose the induction hypothesis breaks into the following:
\begin{enumerate}\small
\item $H_1$: $\tjudge{\Omega'}{\cdot}{f_1}{e_1}{f_2}{\arrow{x}{f_3}{T_1}{f_4}{T}{y}}{z}$
\item $H_2$: $\tjudge{\Omega'}{\cdot}{f_5}{e_2}{f_6}{T_1}{w}$
\item $H_3$: $\pjudge{\mctx{\Omega'}{x}{T_1}}{\Gamma}{f_3[w\mapsto x]}{\leq}{(f_2 + f_6)[w\mapsto x]}$
\item $Constraint_1$: $\Omega =\mctx{\Omega'}{x}{T_1}$
\item $Constraint_2$: $f_1+f_5 = f$
\item $Constraint_3$: $(f_2 + f_6-f_3)[w\mapsto x]+f_4 = g$
\end{enumerate}
We first get either $e_1$ is a value, or $e_1$ is not a value.  If $e_1$ is not a value, then the reduction of $\app{e_1}{e_2}$ must take form of $\ectx{\app{e_1}{e_2}}{q}\hookrightarrow\ectx{\app{e_1'}{e_2}}{q'}$, which means $\ectx{e_1}{q}\hookrightarrow\ectx{e_1'}{q'}$. Then by the induction hypothesis on $H_1$, we have that there exists $f_1'$, $y'$ and $v$ such that: 
\begin{enumerate}\small
\item $\tjudge{\Omega'[y'\mapsto v]}{\cdot}{f_1'}{e_1'}{f_2[y'\mapsto v]}{(\arrow{x}{f_3}{T_1}{f_4}{T}{y})[y'\mapsto v]}{z}$
\item $\pjudge{\Omega'[y'\mapsto v]}{\cdot}{(q-f_1)[y'\mapsto v]}{\le}{q'-f_1'}$
\end{enumerate}
Then applying \ref{lemma:instantiation} on $H_2$ we get
\begin{enumerate}\small
\item $\tjudge{\Omega'[y'\mapsto v]}{\cdot}{f_5[y'\mapsto v]}{e_2[y'\mapsto v]}{f_6[y'\mapsto v]}{T_1[y'\mapsto v]}{w}$
\end{enumerate}
Thus we take  $f_1'+f_5[y'\mapsto v]$ as the existential $f'$, then we have
\begin{enumerate}\small
\item $\tjudge{\Omega[y'\mapsto v]}{\cdot}{f_1'+f_5[y'\mapsto v]}{\app{e_1'}{e_2}}{f[y'\mapsto v]}{T[y'\mapsto v]}{y}$
\item $\pjudge{\Omega'[y'\mapsto v]}{\cdot}{(q-f_1)[y'\mapsto v]-f_5[y'\mapsto v]}{\le}{q'-f_1'-f_5[y'\mapsto v]}$
\end{enumerate}
Notice that $\Omega =\mctx{\Omega'}{x}{T_1}$, thus the inequality (2) still holds after we add $x$ to the $\Omega'$. Thus we prove the case where $e_1$ is not a value. Similarly, we can prove the case where $e_1$ is a value and $e_2$ is not a value by the same proof strategy.\newline

If $e_1$ and $e_2$ are both values, then the reduction will be $\ectx{\app{e_1}{e_2}}{q}\hookrightarrow\ectx{e'}{q}$ where $e'$ is the application result of $e_1$ on $e_2$ by rule \textsc{EApp}. We first apply the lemma \ref{lemma:value-strengthening} to $H_1$ and $H_2$ and get:
\begin{enumerate}\small
\item $\tjudge{\Omega'}{\cdot}{0}{e_1}{0}{\arrow{x}{f_3}{T_1}{f_4}{T}{y}}{z}$
\item $\tjudge{\Omega'}{\cdot}{0}{e_2}{0}{T_1}{w}$
\end{enumerate}
Then apply the substitution lemma \ref{lemma:substitution} to (1) and (2), we get the application result will have:
\begin{enumerate}\small
\item
 $\tjudge{\Omega'[x\mapsto e_2]}{\cdot}{f_3[x\mapsto e_2]}{e'}{(f_4[x\mapsto e_2])}{T[x\mapsto e_2]}{y}$
\end{enumerate}
Next we find $e_2$ is a value, thus the local binder $w$ here will not appear in type context. Therefore, $f_2$ and $f_3$ will contain no free occurrence of $w$, otherwise it is not well-typed. Then we can rewrite the $H_3$ to:
\begin{enumerate}\small
\item
$\pjudge{\Omega'[x\mapsto e_2]}{\cdot}{(f_2 + (f_6[w\mapsto x])-f_3)[x\mapsto e_2]}{\ge}{0}$
\end{enumerate}
We can now use \textsc{TRelax} to add both sides of typing on $e'$ by $(f_2 + (f_6[w\mapsto x])-f_3)[x\mapsto e_2]$.
\begin{enumerate}\small
\item $\tjudge{\Omega'[x\mapsto e_2]}{\cdot}{(f_2 + (f_6[w\mapsto x]))[x\mapsto e_2]}{e'}{g[x\mapsto e_2]}{T[x\mapsto e_2]}{y}$
\end{enumerate}
We then apply the lemma \ref{lemma:value-potentials} to $H_1$ and $H_2$ to get:
\begin{enumerate}\small
\item $\pjudge{\Omega'}{\cdot}{f_2[z\mapsto e_1]}{\le}{f_1}$
\item $\pjudge{\Omega'}{\cdot}{f_6[w\mapsto e_2]}{\le}{f_5}$
\end{enumerate}
Since $z$ binds on an arrow type, while potential functions do not proceed over arrow type, thus we can rewrite the constraint as:
\begin{enumerate}\small
\item $\pjudge{\Omega'[x\mapsto e_2]}{\cdot}{f_2[x\mapsto e_2]}{\le}{f_1[x\mapsto e_2]}$
\item $\pjudge{\Omega'[x\mapsto e_2]}{\cdot}{f_6[w\mapsto x][x\mapsto e_2]}{\le}{f_5[x\mapsto e_2]}$\end{enumerate}
Then by $Constraint_2$ we have:
\begin{enumerate}\small
\item $\pjudge{\Omega'[x\mapsto e_2]}{\cdot}{(f_2 + (f_6[w\mapsto x]))[x\mapsto e_2]}{\le}{f[x\mapsto e_2]}$
\end{enumerate}
Therefore we can use a \textsc{TRelax} rule to add both sides of $e'$ typing by $(f-((f_2 + (f_6[w\mapsto x])))[x\mapsto e_2])$, and then
 followed by a \textsc{TDrop} rule that drops the right side by $(f-((f_2 + (f_6[w\mapsto x])))[x\mapsto e_2])$ then we got:
\begin{enumerate}\small
\item $\tjudge{\Omega'[x\mapsto e_2]}{\cdot}{f[x\mapsto e_2]}{e'}{g[x\mapsto e_2]}{T[x\mapsto e_2]}{y}$
\end{enumerate}
Notice that $\Omega[x\mapsto e_2]=(\mctx{\Omega'}{x}{T_1})[x\mapsto e_2]=\Omega'[x\mapsto e_2]$,
 thus we prove the case since $\pjudge{\Omega[x\mapsto e_2]}{\cdot}{(q-f)[x\mapsto e_2]}{=}{q-f[x\mapsto e_2]}$ as the evaluation process is  $\ectx{\app{e_1}{e_2}}{q}\hookrightarrow\ectx{e'}{q}$.\newline

\item Case \textsc{TPApp}: Same as TApp but simpler, we can directly get the proof by applying the theorem \ref{lemma:substitution}.\newline

\item Case \textsc{TLet}: Same as TApp.\newline

\item Case \textsc{TVar}: it is not well-typed since type context is empty.\newline

\item Case \textsc{TPair}: By induction we can handle the case where $e_1$ is not a value or $e_1$ is a value while $e_2$ is not just as the case \textsc{TApp}. As for the case where $e_1$ and $e_2$ are both values, then it is a value it slfe.\newline

\item Case \textsc{TFix}: By induction it is easy to show if $\tjudge{\Omega}{\mctx{\Gamma}{x}{T}}{f_1}{e[x\mapsto \fix{x}{T}{e}]}{f_2}{T}{y}$ and $x\notin \mathbf{typefv}(e)\cup\mathbf{fv}(T)$, then
$\tjudge{\mctx{\Omega}{x}{T}}{\Gamma}{f_1}{e[x\mapsto \fix{x}{T}{e}]}{f_2}{T}{y}$. By premise we have $\tjudge{\Omega}{\mctx{\Gamma}{x}{T}}{0}{e}{0}{T}{y}$. Thus $\tjudge{\mctx{\Omega}{x}{T}}{\Gamma}{0}{e[x\mapsto \fix{x}{T}{e}]}{0}{T}{y}$, then by rule \textsc{TErase} we prove the case.\newline

\item Case \textsc{TRelax}: By induction on $\tjudge{\Omega}{\cdot}{f}{e}{g}{T}{x}$, suppose the induction hypothesis breaks into the following:\newline

\begin{enumerate}
\item $H_1$: $\tjudge{\Omega}{\Gamma}{f_1}{e}{f_2}{T}{x}$
\item $H_2$: $\pjudge{\Omega}{\Gamma}{f_3}{}{}$
\item $Constraint_1$: $\Omega = \Omega$
\item $Constraint_2$: $\Gamma = \cdot$
\item $Constraint_3$: $f_1+f_3 = f$
\item $Constraint_4$: $f_2+f_3 = g$\newline
\end{enumerate}

If $e$ is a value, then it can not further step. If $e$ is not a value, then by the inductive hypothesis on $H_1$, for $\ectx{e}{q}\hookrightarrow\ectx{e'}{q'}$, we have that $\tjudge{\Omega[y\mapsto v]}{\cdot}{f_1'}{e'}{f_2[y\mapsto v]}{T}{x}$ and $\pjudge{\Omega[y\mapsto v]}{\cdot}{q-f_1[y\mapsto v]}{\le}{q'-f_1'}$. Thus we have $\tjudge{\Omega[y\mapsto v]}{\cdot}{f_1'+f_3[y\mapsto v]}{e'}{f_2[y\mapsto v]+f_3[y\mapsto v]}{T[y\mapsto v]}{x}$ by \textsc{TRelax} and lemma \ref{lemma:instantiation}, and that $\pjudge{\Omega[y\mapsto v]}{\cdot}{q-f_1[y\mapsto v]-f_3[y\mapsto v]}{\le}{q'-f_1'-f_3[y\mapsto v]}$. Thus we prove the case.
\newline

\item Case \textsc{TProj1} and \textsc{TProj2}: By induction on $\tjudge{\Omega}{\cdot}{f}{\fst e}{g}{T}{y}$ or $\tjudge{\Omega}{\cdot}{f}{\snd e}{g}{T}{y}$, we have the induction hypothesis that either $e$ is a value or not. If $e$ is a not value, then apply the same technique in case \textsc{TApp} we get the proof. If $e$ is a value, then $e$ has to be in the form of a pair $\pair{v_1}{v_2}$, for the case $\sf{TProj}_1$. We may get the rule $\tjudge{\Omega}{\cdot}{f_1+f_3}{\pair{e_1}{e_2}}{f_2[x_1\mapsto \fst{y}]+f_4[x_2\mapsto \snd{y}]}{\prodT{T_1}{T_2}}{y}$ from \textsc{TRelax}, \textsc{TDrop}, \textsc{TErase}, \textsc{TRename}, \textsc{TPair}.\newline

Rule \textsc{TDrop}, \textsc{TErase}, \textsc{TRename} has no effect on our proof since they don't change the input potenital $f$. For rule \textsc{TRelax}, we can always combine it with the previous inference rule. If the rule before \textsc{TRelax} is \textsc{TRelax}, then we can combine this two into one \textsc{TRelax}. If the rule before \textsc{TRelax} is \textsc{TPair}, then we can lift the relaxation to its first projection. Thus we only need to consider the case for \textsc{TPair}.\newline

Suppose it is typed as followed with $\pjudge{\mctx{\Omega}{y}{\prodT{T_1}{T_2}}}{\cdot}{f_5[x_1\mapsto \fst{y}]}{\leq}{f_2[x_1\mapsto \fst{y}]+f_4[x_2\mapsto \snd{y}]}$:
\begin{mathpar}\small
\mprset{flushleft}
  \inferrule*[Right=(TPair)]
{\tjudge{\Omega}{\cdot}{f_1}{v_1}{f_2}{T_1}{x_1}\\
\tjudge{\Omega}{\cdot}{f_3}{v_2}{f_4}{T_2}{x_2}
}
{
  \inferrule*[Right=(TProj1)]
{\tjudge{\Omega}{\cdot}{f_1+f_3}{\pair{e_1}{e_2}}{f_2[x_1\mapsto \fst{y}]+f_4[x_2\mapsto \snd{y}]}{\prodT{T_1}{T_2}}{y}}
{\tjudge{\Omega}{\cdot}{f_1+f_3}{\fst{\pair{v_1}{v_2}}}{f_5}{T_1}{}}
}
\end{mathpar}

Notice that $\pjudge{\mctx{\Omega}{x_1}{T_1}}{\cdot}{(f_5[x_1\mapsto \fst y])[y\mapsto \pair{x_1}{v_2}]}{\leq}{(f_2[x_1\mapsto \fst{y}]+f_4[x_2\mapsto \snd{y}])[y\mapsto \pair{x_1}{v_2}]}$, thus $\pjudge{\mctx{\Omega}{x_1}{T_1}}{\cdot}{f_5}{\leq}{f_2+f_4[x_2\mapsto v_2]}$.\newline

By lemma \ref{lemma:value-potentials} we have $\pjudge{\Omega}{\cdot}{f_4[x_2\mapsto v_2]}{\leq}{f_3}$. Thus by the \textsc{TRelax} we have $\tjudge{\Omega}{\Gamma}{f_1+f_3}{v_1}{f_2+f_4[x_2\mapsto v_2]}{T_1}{x_1}$. Then by \textsc{TDrop} we obtain the proof for the case $\sf{TProj}_1$. Similar we can prove the case  $\sf{TProj}_2$.\newline

\item Case \textsc{TTickn} and \textsc{TTickp}: By induction on $\tjudge{\Omega}{\cdot}{f}{\tick{q}{e}}{g}{T}{x}$ we can easily get that $\tjudge{\Omega}{\cdot}{f'}{\tick{q}{e}}{g}{T}{x}$ and $\pjudge{\Omega}{\cdot}{q-f}{=}{q'-f'}$. Take $y$ as a fresh variable we prove the case.\newline

\item Case \textsc{TCons}: $\tjudge{\Omega}{\cdot}{f}{\con{i}{e_0}{e_1}{e_{m_i}}}{g}{\ind}{y}$ by induction we obtain the cases where for $j\in [0,m_i]$, and for all $j'\in[0,j-1]$, $e_{j'}$ is a value while $e_j$ is not, or the case that forall $j''\in[0,m_i]$, $e_{j''}$ is a value. For the former cases, we can apply the same technique in case \textsc{TApp} and then we get the proof. For the later case, it cannot be further reduced.\newline

\item Case \textsc{TOp}: Same as \textsc{TCons}.\newline

\item Case \textsc{TDes}: By direct induction on the judgement, we have the induction hypothesis that either $e_0$ is a value or not. If $e_0$ is a not value, then apply the same technique in Case TApp we get the proof. If $e_0$ is a value, then it has to be in the form of $\con{i}{v_0}{v_1}{v_{m_i}}$, then by apply the substitution lemma \ref{lemma:substitution} in the premise for multiple times get the proof.\newline

\item Case \textsc{TRename}, \textsc{TDrop}, \textsc{TErase}: We directly prove the case.
\end{itemize}
\end{proof}

\begin{theorem}[Soundness]\label{lemma:soundness}
If $\ajudge{\Omega}{\Gamma}{f_1}{e}{f_2}{T}{x}$, then $\tjudge{\Omega}{\Gamma}{f_1}{e}{f_2}{T}{x}$.
\end{theorem}
\begin{proof}
By direct induction of $\ajudge{\Omega}{\Gamma}{f_1}{e}{f_2}{T}{x}$, The only difference is in \textsc{Aproj1}, \textsc{Aproj2}, \textsc{ATickn}, \textsc{AAppn}, \textsc{ACons}, \textsc{ADes}. All of them can be achieved by adding the \textsc{TRelax} and \textsc{TDrop}, thus we prove the case.
\end{proof}

\begin{theorem}[Determinisism]\label{lemma:Determinisism}
If $\ajudge{\Omega}{\Gamma}{f_1}{e}{f_2}{T}{x}$, $\ajudge{\Omega}{\Gamma}{f_3}{e}{f_4}{T}{x}$, then
$\pjudge{\Omega}{\Gamma}{f_1}{=}{f_3}$, and $\pjudge{\mctx{\Omega}{x}{T}}{\Gamma}{f_2}{=}{f_4}$.
\end{theorem}
\begin{proof}
It is easy to prove for all $\Omega,\Gamma,f_1,e,f_2,T,x$, if $\ajudge{\Omega}{\Gamma}{f_1}{e}{f_2}{T}{x}$ and $y\notin \Omega\cup\Gamma$, then $\ajudge{\mctx{\Omega}{y}{T_1}}{\Gamma}{f_1}{e}{f_2}{T}{x}$.
Notice that for fixed $\Omega,\Gamma,e,x$, we can only get it from either \textsc{TErase}, \textsc{TRename} or another concrete rule. It is easy to show that \textsc{TErase} commutes with all other rule, and \textsc{TRename} has no effect on all of the subsequent typing judgement $\Omega,\Gamma,f_1$, while $f_2$ is equivalent up to variable renaming. Since $\Omega,\Gamma,f_1,e,T,x$ are all fixed, thus $f_3=f_4$.
\end{proof}

\subsection{Full Algorithm}
The full rules for algorithm shows here:
\begin{figure}[thb]

\begin{center}
\begin{mathpar}

\inferrule[(AProj1)]
{y,z\notin \ms{dom}(\Omega)\cup\ms{dom}(\Gamma)\\\\
\ajudge{\Omega}{\Gamma}{f_1}{e}{f_2}{\prodT{T_1}{T_2}}{x}\\
}
{\ajudge{\Omega}{\Gamma}{f_1}{\fst{e}}{\mini{z:T_2}{f_2[x\mapsto(y,z)]}}{T_1}{y}}
\hva \and
\inferrule[(AVar)]
{x:T \in \Gamma
}
{\ajudge{\Omega}{\Gamma}{0}{x}{0}{T}{x}}
\hva \and
\inferrule[(AProj2)]
{y,z\notin \ms{dom}(\Omega)\cup\ms{dom}(\Gamma)\\\\
\ajudge{\Omega}{\Gamma}{f_1}{e}{f_2}{\prodT{T_1}{T_2}}{x}\\
}
{\tjudge{\Omega}{\Gamma}{f_1}{\snd{e}}{\mini{y:T_1}{f_2[x\mapsto(y,z)]}}{T_2}{z}}
\hva \and
\inferrule[(AInt)]
{ }
{\ajudge{\Omega}{\Gamma}{\fctx{0}}{i}{0}{\Int}{x}}
\hva \and
\inferrule[(AErase)]
{\ajudge{\Omega}{\Gamma}{f_1}{e}{f_2}{T}{y}\\
x:T'\in\Omega\\\\
x\notin \mathbf{fv}(\Omega)\cup\mathbf{fv}(\Gamma)\cup\mathbf{fv}(f_1)\cup\mathbf{fv}(f_2)\cup\mathbf{fv}(T)\\
}
{\ajudge{\Omega\backslash x}{\Gamma}{f_1}{e}{f_2}{T}{y}}
\hva \and
\inferrule[(ATickp)]
{\ajudge{\Omega}{\Gamma}{f_1}{e}{f_2}{T}{x}\\
p\ge 0
}
{\ajudge{\Omega}{\Gamma}{f_1+p}{\tick{p}{e}}{f_2}{T}{x}}
\hva \and
\inferrule[(ATickn)]
{\ajudge{\Omega}{\Gamma}{f_1}{e}{f_2}{T}{x}\\\\
p < 0\\
}
{\ajudge{\Omega}{\Gamma}{\maxi{}{f_1+p,0}}{\tick{p}{e}}{f_2-\mini{}{f_1+p,0}}{T}{x}}
\hva \and
\inferrule[(ARename)]
{x,y\notin \ms{dom}(\Omega)\cup\ms{dom}(\Gamma)\\\\
\ajudge{\Omega}{\Gamma}{f_1}{e}{f_2}{T}{x}\\
}
{\ajudge{\Omega}{\Gamma}{f_1}{e}{f_2[x\mapsto y]}{T}{y}\\}
\hva \and
\inferrule[(AOp)]
{y\notin \ms{dom}(\Omega)\cup\ms{dom}(\Gamma)\\\\
\forall j\in[1,m_i],\ajudge{\Omega}{\Gamma}{f_{1_j}}{e_j}{f_{2_j}}{\Int}{x_j}\\
}
{\ajudge{\Omega}{\Gamma}{\sum\limits_{j=1}^{m_i} f_{1_j}}{\op{i}{e}}{\sum\limits_{j=1}^{m_i} f_{2_j}}{\Int}{y}}
\hva \and
\inferrule[(AFix)]
{z\notin \ms{dom}(\Omega)\cup\ms{dom}(\Gamma)\\\\
\ajudge{\Omega}{\mctx{\Gamma}{x}{T}}{0}{e}{f}{T}{y}\\
}
{\ajudge{\Omega}{\Gamma}{0}{\fix{x}{T}{e}}{0}{T}{z}}
\hva \and
\inferrule[(APabs)]
{z\notin \ms{dom}(\Omega)\cup\ms{dom}(\Gamma)\\\\
\ajudge{\Omega}{\mctx{\Gamma}{x}{T_1}}{0}{e}{f}{T_2}{y}\\
}
{\ajudge{\Omega}{\Gamma}{0}{\Lam{x}{e}{T_1}}{0}{\poly{x}{T_1}{T_2}}{z}}
\hva \and
\inferrule[(APapp)]
{\ajudge{\Omega}{\Gamma}{f_1}{e}{f_2}{\poly{x}{T_1}{T_2}}{y}\\\\
\ajudge{\Omega}{\Gamma}{0}{pv}{0}{T_1}{z}\\w\notin \ms{dom}(\Omega)\cup\ms{dom}(\Gamma)\\
}
{\ajudge{\Omega}{\Gamma}{f_1}{\app{e}{pv}}{f_2}{T_2[x\mapsto pv]}{w}}
\hva \and
\inferrule[(APair)]
{\ajudge{\Omega}{\Gamma}{f_1}{e_1}{f_2}{T_1}{x_1}\\
\ajudge{\Omega}{\Gamma}{f_3}{e_2}{f_4}{T_2}{x_2}\\x\notin \ms{dom}(\Omega)\cup\ms{dom}(\Gamma)\\
}
{
\ajudge{\Omega}{\Gamma}{f_1+f_3}{\pair{e_1}{e_2}}{f_2[x_2\mapsto\snd{x}]+f_4[x_1\mapsto\fst{x}]}{\prodT{T_1}{T_2}}{x}
}
\hva \and
\inferrule[(AAbs)]
{\ajudge{\Omega}{\mctx{\Gamma}{x}{T_1}}{f_3}{e}{f_4}{T_2}{y}\\z\notin \ms{dom}(\Omega)\cup\ms{dom}(\Gamma)\\
\pjudge{\Omega}{\mctx{\Gamma}{x}{T_1}}{f_3}{\le}{f_1}\\
\pjudge{\mctx{\Omega}{y}{T_2}}{\mctx{\Gamma}{x}{T_1}}{f_2+f_3}{\le}{f_4+f_1}
}
{\ajudge{\Omega}{\Gamma}{0}{\llam{x}{e}{T_1}}{0}{(\arrow{x}{f_1}{T_1}{f_2}{T_2}{y}}{z})}

\end{mathpar}

\end{center}
\caption{Algorithmic Typing Rules}
\label{fig:typing3}
\end{figure}

\begin{figure}[thb]

\begin{center}
\begin{mathpar}
\inferrule[(AAppp)]
{\ajudge{\Omega}{\Gamma}{f_1}{e_1}{f_2}{(\arrow{x}{f_3}{T_1}{f_4}{T_2}{y})}{z}\\
\ajudge{\Omega}{\Gamma}{f_5}{e_2}{f_6}{T_1}{w}\\
\pjudge{\mctx{\Omega}{x}{T_1}}{\Gamma}{f_3[w\mapsto x]}{\leq}{(f_2 + f_6)[w\mapsto x]}\\
}
{\ajudge{\mctx{\Omega}{x}{T_1}}{\Gamma}{f_1+f_5}{\app{e_1}{e_2}}{(f_2 + f_6-f_3)[w\mapsto x]+f_4}{T_2}{y}}
\hva \and
\inferrule[(AAppn)]
{\ajudge{\Omega}{\Gamma}{f_1}{e_1}{f_2}{(\arrow{x}{f_3}{T_1}{f_4}{T_2}{y})}{z}\\
\ajudge{\Omega}{\Gamma}{f_5}{e_2}{f_6}{T_1}{w}\\
\pjudge{\mctx{\Omega}{x}{T_1}}{\Gamma}{f_3[w\mapsto x]}{>}{(f_2 + f_6)[w\mapsto x]}\\
\mathbf{if~}w\notin \ms{dom}(\Omega)\cup\ms{dom}(\Gamma),\\
\mathbf{~then~}f=\mini{x:T_1}{(f_3-f_2-f_6)[w\mapsto x]}\\
\mathbf{else~}f=f_3-f_2-f_6
}
{\ajudge{\mctx{\Omega}{x}{T_1}}{\Gamma}{f+f_1+f_5}{\app{e_1}{e_2}}{(f+f_2 + f_6-f_3)[w\mapsto x]+f_4}{T_2}{y}}
\hva \and
\inferrule[(ALet)]
{\ajudge{\Omega}{\Gamma}{f_1}{e_1}{f_2}{T_1}{z}\\
\ajudge{\Omega}{\mctx{\Gamma}{x}{T_1}}{f_3}{e_2}{f_4}{T_2}{y}\\
\pjudge{\Omega}{\mctx{\Gamma}{x}{T_1}}{f_3}{\leq}{f_2}\\
}
{\ajudge{\mctx{\Omega}{x}{T_1}}{\Gamma}{f_1+f_5}{\Let{x}{e_1}{e_2}}{f_2 -f_3+f_4}{T_2}{y}}
\hva \and
\inferrule[(ACons)]
{\ajudge{\Omega}{\Gamma}{f_1}{e_0}{f_2}{T_i}{x_0}\\
\forall j\in[1,m_i],\ajudge{\Omega}{\Gamma}{f_{3_j}}{e_j}{f_{4_j}}{\ind}{x_j}\\
f_5=matd(y,\con{i}{x_0}{x_1}{x_{m_i}}. (f_2+\sum\limits_{j=1}^{m_i} f_{4_j}),\overrightarrow{\con{j(i\neq j)}{x_0}{x_1}{x_{m_j}}. +\infty})\\
}
{\ajudge{\Omega}{\Gamma}{f_1 + \sum\limits_{j=1}^{m_i} f_{3_j}}{\con{i}{e_0}{e_1}{e_{m_i}}}{f_5}{\ind}{y}}
\hva \and
\inferrule[(ADes)]
{\ajudge{\Omega}{\Gamma}{f_1}{e_0}{f_2}{\ind}{x}\\
\forall i,\ajudge{\Omega}{\Gamma,\,x_0:T_i,...,\,x_{m_i}:\ind}{f_{3i}}{e_i}{f_{4i}}{T_1}{y}\\
f=\maxi{}{\maxi{i}{\maxi{x_0:T_i,...,\,x_{m_i}:\ind}{f_{3i}}},f_2}\\
\mathbf{if~}x\notin \Omega\cup\Gamma,\mathbf{~then~}f_5=\mini{x:\ind}{f-f_2}\\
\mathbf{else~}f_5=f-f_2\\
}
{\ajudge{\Omega}{\Gamma}{f_1+f_5}{\des{C}{m}{e}}{\mini{i}{\mini{x_0:T_i,...,\,x_{m_i}:\ind}{f_5+f_2-f_{3i}+f_{4i}}}}{T_1}{y}}
\end{mathpar}
\end{center}
\caption{Algorithmic Typing Rules}
\label{fig:typing4}
\end{figure}

\begin{figure}[thb]

\begin{center}
\begin{mathpar}\small
\inferrule* [lab=(TTickp),narrower=1]
{\tjudge{\Omega}{\Gamma}{f_1}{e}{f_2}{T}{x}\\
p\ge 0
}
{\tjudge{\Omega}{\Gamma}{f_1+p}{\tick{p}{e}}{f_2}{T}{x}}
\hva \and
\inferrule* [lab=(TTickn)]
{\tjudge{\Omega}{\Gamma}{f_1-p}{e}{f_2}{T}{x}\\
p < 0
}
{\tjudge{\Omega}{\Gamma}{f_1}{\tick{p}{e}}{f_2}{T}{x}}
\hva \and
\inferrule[(TInt)]
{ x\notin \ms{dom}(\Omega)\cup\ms{dom}(\Gamma)}
{\tjudge{\Omega}{\Gamma}{\fctx{0}}{i}{0}{\Int}{x}}
\hva \and
\inferrule* [lab=(TDrop)]
{\tjudge{\Omega}{\Gamma}{f_1}{e}{f_2}{T}{x}\\\\
\pjudge{\mctx{\Omega}{x}{T}}{\Gamma}{f_3}{\leq}{f_2}\\
}
{
\tjudge{\Omega}{\Gamma}{f_1}{e}{f_3}{T}{x}
}
\hva \and
\inferrule* [lab=(TFix)]
{z\notin \ms{dom}(\Omega)\cup\ms{dom}(\Gamma)\\\\
x\notin \mathbf{typefv}(e)\cup\mathbf{fv}(T)\\\\
\tjudge{\Omega}{\mctx{\Gamma}{x}{T}}{0}{e}{0}{T}{y}\\
}
{\tjudge{\Omega}{\Gamma}{0}{\fix{x}{T}{e}}{0}{T}{z}}
\hva \and
\inferrule[(TRename)]
{x,y\notin \ms{dom}(\Omega)\cup\ms{dom}(\Gamma)\\\\
\tjudge{\Omega}{\Gamma}{f_1}{e}{f_2}{T}{x}\\
}
{\tjudge{\Omega}{\Gamma}{f_1}{e}{f_2[x\mapsto y]}{T[x\mapsto y]}{y}\\}
\hva \and
\inferrule[(TProj1)]
{y\notin \ms{dom}(\Omega)\cup\ms{dom}(\Gamma)\\\\
\tjudge{\Omega}{\Gamma}{f_1}{e}{f_2}{\prodT{T_1}{T_2}}{x}\\\\
\pjudge{\mctx{\Omega}{x}{\prodT{T_1}{T_2}}}{\Gamma}{f_3[y\mapsto \fst x]}{\leq}{f_2}\\
}
{\tjudge{\Omega}{\Gamma}{f_1}{\fst{e}}{f_3}{T_1}{y}}
\hva \and
\inferrule[(TProj2)]
{y\notin \ms{dom}(\Omega)\cup\ms{dom}(\Gamma)\\\\
\tjudge{\Omega}{\Gamma}{f_1}{e}{f_2}{\prodT{T_1}{T_2}}{x}\\\\
\pjudge{\mctx{\Omega}{x}{\prodT{T_1}{T_2}}}{\Gamma}{f_3[y\mapsto \snd x]}{\leq}{f_2}\\
}
{\tjudge{\Omega}{\Gamma}{f_1}{\snd{e}}{f_3}{T_2}{y}}
\hva \and
\inferrule* [lab=(TVar)]
{x:T \in \Gamma
}
{\tjudge{\Omega}{\Gamma}{0}{x}{0}{T}{x}}
\hva \and
\inferrule[(TRelax)]
{\tjudge{\Omega}{\Gamma}{f_1}{e}{f_2}{T}{x}\\\\
\pjudge{\Omega}{\Gamma}{f_3}{\ge}{0}\\
}
{
\tjudge{\Omega}{\Gamma}{f_1+f_3}{e}{f_2+f_3}{T}{x}
}
\end{mathpar}
\end{center}
\caption{Typing Rules}
\label{fig:typing5}
\end{figure}
\begin{figure}[thb]

\begin{center}
\begin{mathpar}\small
\hva \and
\inferrule[(TPair)]
{\tjudge{\Omega}{\Gamma}{f_1}{e_1}{f_2}{T_1}{x_1}\\\\
\tjudge{\Omega}{\Gamma}{f_3}{e_2}{f_4}{T_2}{x_2}\\
x\notin \ms{dom}(\Omega)\cup\ms{dom}(\Gamma)
}
{
\tjudge{\Omega}{\Gamma}{f_1+f_3}{\pair{e_1}{e_2}}{f_2[x_1\mapsto\fst{x}]+f_4[x_2\mapsto\snd{x}]}{\prodT{T_1}{T_2}}{x}
}
\hva \and
\inferrule[(TErase)]
{\tjudge{\Omega}{\Gamma}{f_1}{e}{f_2}{T}{y}\\
x:T'\in\Omega\\\\
x\notin \mathbf{fv}(\Omega\backslash x)\cup\mathbf{fv}(\Gamma)\cup\mathbf{fv}(f_1)\cup(\mathbf{fv}(f_2)\backslash y)\cup\mathbf{fv}(T\backslash y)\\
}
{\tjudge{\Omega\backslash x}{\Gamma}{f_1}{e}{f_2}{T}{y}}
\hva \and
\inferrule[(TPabs)]
{z\notin \ms{dom}(\Omega)\cup\ms{dom}(\Gamma)\\\\
\tjudge{\Omega}{\mctx{\Gamma}{x}{T_1}}{0}{e}{0}{T_2}{y}\\
}
{\tjudge{\Omega}{\Gamma}{0}{\Lam{x}{e}{T_1}}{0}{\poly{x}{T_1}{T_2}}{z}}
\hva \and
\inferrule[(TPapp)]
{\tjudge{\Omega}{\Gamma}{f_1}{e}{f_2}{\poly{x}{T_1}{T_2}}{y}\\\\
\tjudge{\Omega}{\Gamma}{0}{pv}{0}{T_1}{z}\\
w\notin \ms{dom}(\Omega)\cup\ms{dom}(\Gamma)\\
}
{\tjudge{\Omega}{\Gamma}{f_1}{\app{e}{pv}}{f_2}{T_2[x\mapsto pv]}{w}}
\hva \and
\inferrule[(TOp)]
{y\notin \ms{dom}(\Omega)\cup\ms{dom}(\Gamma)\\\\
\forall j\in[1,m_i],\tjudge{\Omega}{\Gamma}{f_{1_j}}{e_j}{f_{2_j}}{\Int}{x_j}\\
}
{\tjudge{\Omega}{\Gamma}{\sum\limits_{j=1}^{m_i} f_{1_j}}{\op{i}{e}}{\sum\limits_{j=1}^{m_i} f_{2_j}}{\Int}{y}}
\end{mathpar}
\end{center}
\caption{Typing Rules}
\label{fig:typing6}
\end{figure}

\begin{figure}[thb]

\begin{center}
\begin{mathpar}\small
\inferrule[(TAbs)]
{ \tjudge{\Omega}{\mctx{\Gamma}{x}{T_1}}{f_1}{e}{f_2}{T_2}{y}\\
z\notin \ms{dom}(\Omega)\cup\ms{dom}(\Gamma)
}
{ \tjudge{\Omega}{\Gamma}{0}{\llam{x}{e}{T_1}}{0}{(\arrow{x}{f_1}{T_1}{f_2}{T_2}{y}}{z})}
\hva \and
\inferrule[(TApp)]
{\tjudge{\Omega}{\Gamma}{f_1}{e_1}{f_2}{(\arrow{x}{f_3}{T_1}{f_4}{T_2}{y})}{z}\\
\tjudge{\Omega}{\Gamma}{f_5}{e_2}{f_6}{T_1}{w}\\
\pjudge{\mctx{\Omega}{x}{T_1}}{\Gamma}{f_3[w\mapsto x]}{\leq}{(f_2 + f_6)[w\mapsto x]}\\
}
{\tjudge{\mctx{\Omega}{x}{T_1}}{\Gamma}{f_1+f_5}{\app{e_1}{e_2}}{(f_2 + f_6-f_3)[w\mapsto x]+f_4}{T_2}{y}}
\hva \and
\inferrule[(TLet)]
{\tjudge{\Omega}{\Gamma}{f_1}{e_1}{f_2}{T_1}{z}\\
\tjudge{\Omega}{\mctx{\Gamma}{x}{T_1}}{f_3}{e_2}{f_4}{T_2}{y}\\
\pjudge{\Omega}{\mctx{\Gamma}{x}{T_1}}{f_3}{\leq}{f_2[z\mapsto x]}\\
}
{\tjudge{\mctx{\Omega}{x}{T_1}}{\Gamma}{f_1+f_5}{\Let{x}{e_1}{e_2}}{f_2[z\mapsto x] -f_3+f_4}{T_2}{y}}
\hva \and
\inferrule[(TCons)]
{y\notin \ms{dom}(\Omega)\cup\ms{dom}(\Gamma)\\
\tjudge{\Omega}{\Gamma}{f_1}{e_0}{f_2}{T_i}{x_0}\\
\forall j\in[1,m_i],\tjudge{\Omega}{\Gamma}{f_{3_j}}{e_j}{f_{4_j}}{\ind}{x_j}\\
\pjudge{\Omega,\,x_0:T_i,...,\,x_{m_i}:\ind}{\Gamma}{f_5[y\mapsto \con{i}{x_0}{x_1}{x_{m_i}}]}{\leq}{f_2+\sum\limits_{j=1}^{m_i} f_{4_j}}{}\\
}
{\tjudge{\Omega}{\Gamma}{f_1 + \sum\limits_{j=1}^{m_i} f_{3_j}}{\con{i}{e_0}{e_1}{e_{m_i}}}{f_5}{\ind}{y}}
\hva \and
\inferrule[(TDes)]
{\tjudge{\Omega}{\Gamma}{f_1}{e_0}{f_2}{\ind}{x}\\
\forall i,\forall j\in[1,m_i],x_j\notin \mathbf{fv}(f_3)\cup\mathbf{fv}(T_1)\\
\forall i,\tjudge{\Omega}{\Gamma,\,x_0:T_i,...,\,x_{m_i}:\ind}{f_2[x\mapsto \con{i}{x_0}{x_1}{x_{m_i}}]}{e_i}{f_3}{T_1}{y}
}
{\tjudge{\Omega}{\Gamma}{f_1}{\des{C}{m}{e}}{f_3}{T_1}{y}}
\end{mathpar}

\end{center}
\caption{Typing Rules}
\label{fig:typing7}
\end{figure}

\begin{figure}[thb]

\begin{center}
\begin{mathpar}\small

\inferrule[(EAppl)]
{\ectx{e_1}{p}\hookrightarrow \ectx{e_1'}{q}}
{\ectx{\app{e_1}{e_2}}{p}\hookrightarrow \ectx{\app{e_1'}{e_2}}{q}}
\hva \and
\inferrule[(EAppr)]
{\ectx{e_2}{p}\hookrightarrow \ectx{e_2'}{q}}
{\ectx{\app{v_1}{e_2}}{p}\hookrightarrow \ectx{\app{v_1}{e_2'}}{q}}
\hva \and
\inferrule[(EProj1)]
{ }
{\ectx{\fst\pair{v_1}{v_2}}{p}\hookrightarrow \ectx{v_1}{p}}
\hva \and
\inferrule[(EProj2)]
{ }
{\ectx{\snd\pair{v_1}{v_2}}{p}\hookrightarrow \ectx{v_2}{p}}
\hva \and
\inferrule[(EPair1)]
{\ectx{e_1}{p}\hookrightarrow \ectx{e_1'}{q}}
{\ectx{\pair{e_1}{e_2}}{p}\hookrightarrow \ectx{\pair{e_1'}{e_2}}{q}}
\hva \and
\inferrule[(EPair2)]
{\ectx{e_2}{p}\hookrightarrow \ectx{e_2'}{q}}
{\ectx{\pair{v_1}{e_2}}{p}\hookrightarrow \ectx{\pair{v_1}{e_2'}}{q}}
\hva \and
\inferrule[(ETick)]
{q-p\ge 0 }
{\ectx{\tick{p}{e}}{q}\hookrightarrow \ectx{e}{q-p}}
\hva \and
\inferrule[(EOp)]
{\mathbf{Eval}(\mathbf{op_i},\overrightarrow{v})= v'\\
}
{\ectx{\op{i}{v}}{p}\hookrightarrow
\ectx{v'}{p}
}
\hva \and
\inferrule[(EFix)]
{ }
{\ectx{\fix{x}{T}{e}}{p}\hookrightarrow \ectx{e[x\mapsto\fix{x}{T}{e}] }{p}}
\hva \and
\inferrule[(EPapp)]
{ }
{\ectx{\app{\Lam{x}{e}{T}}{v}}{p}\hookrightarrow \ectx{e[x\mapsto v]}{p}}
\hva \and
\inferrule[(EApp)]
{ }
{\ectx{\app{\llam{x}{e_1}{T_1}}{v_2}}{p}\hookrightarrow \ectx{e_1[x\mapsto v_2]}{p}}
\hva \and
\inferrule[(ECon0)]
{\ectx{e_0}{p}\hookrightarrow \ectx{e_0'}{q}}
{\ectx{\con{j}{e_0}{e_1}{e_{m_j}}}{p}\hookrightarrow \ectx{\con{j}{e_0'}{e_1}{e_{m_j}}}{q}}
\hva \and
\inferrule[(ELetl)]
{\ectx{e_1}{p}\hookrightarrow \ectx{e_1'}{q}}
{\ectx{\Let{x}{e_1}{e_2}}{p}\hookrightarrow \ectx{\Let{x}{e_1'}{e_2}}{q}}
\hva \and
\inferrule[(ELet)]
{ }
{\ectx{\Let{x}{v_1}{e_2}}{p}\hookrightarrow \ectx{{e_2[x\mapsto v_1]}}{q}}
\hva \and
\inferrule[(EConl)]
{\ectx{e_i}{p}\hookrightarrow \ectx{e_i'}{q}\\
i\ge 1
}
{\ectx{C_j({v_0},({v_1},...,v_{i-1},e_i,e_{i+1},...,e_{m_j}))}{p}\hookrightarrow
\ectx{C_j({v_0},({v_1},...,v_{i-1},e_i',e_{i+1},...,e_{m_j}))}{q}
}
\hva \and
\inferrule[(EOpi)]
{\ectx{e_i}{p}\hookrightarrow \ectx{e_i'}{q}\\
}
{\ectx{\mathbf{op_j}({v_0},({v_1},...,v_{i-1},e_i,e_{i+1},...,e_{m_j}))}{p}\hookrightarrow
\ectx{\mathbf{op_j}({v_0},({v_1},...,v_{i-1},e_i',e_{i+1},...,e_{m_j}))}{q}
}
\hva \and
\inferrule[(ECasel)]
{\ectx{e_0}{p}\hookrightarrow \ectx{e_0'}{q}}
{\ectx{\destructor{e_0}{C}{m}{e}}{p}\hookrightarrow \ectx{\destructor{e_0'}{C}{m}{e}}{q}}
\hva \and
\inferrule[(ECase)]
{ }
{\ectx{\destructor{\con{j}{v_0}{v_1}{v_{m_j}}}{C_j}{m_j}{e_j} }{p}\hookrightarrow \ectx{e_j[x_0,...,x_{m_j}\mapsto v_0,...,v_{m_j}]}{p} }

\end{mathpar}

\end{center}
\caption{Evalutaion Rules}
\label{fig:Eval}
\end{figure}

\subsection{Embedding}

We choose the AARA variant presented in \cite{JMSCS:HJ22}, which serves as a comprehensive summary of two decades of development in automatic amortized resource analysis.

The syntax is given as follows:
\[
\begin{array}{lll}
e ::= 
& x \\[0.2em]
& \mid \langle \rangle \\[0.2em]
& \mid \mathbf{let}\; x = e_1 \;\mathbf{in}\; e_2 \\[0.2em]
& \mid \langle e_1, e_2 \rangle \\[0.2em]
& \mid \mathbf{letp}\; \langle x_1, x_2 \rangle = e_1 \;\mathbf{in}\; e_2 \\[0.2em]
& \mid \mathbf{left}(e) \mid \mathbf{right}(e) \\[0.2em]
& \mid \mathbf{case}\; e \{\, \mathbf{left}(x_1) \mapsto e_1 \mid \mathbf{right}(x_2) \mapsto e_2 \,\} \\[0.2em]
& \mid \mathbf{nil} \mid \mathbf{cons}(x_1, x_2) \\[0.2em]
& \mid \mathbf{case}\; x \{\, \mathbf{nil} \mapsto e_0 \mid \mathbf{cons}(x_1, x_2) \mapsto e_1 \,\} \\[0.2em]
& \mid x_1(x_2) \\[0.2em]
& \mid \mathbf{fun}\; f\,x = e \\[0.2em]
& \mid \mathbf{tick}\; q \\[0.2em]
& \mid \mathbf{share}\; x \;\mathbf{as}\; x_1, x_2 \;\mathbf{in}\; e
\end{array}
\]

In AARA, both construction and deconstruction operations incur explicit resource costs, such as
$c_{\mathsf{Let}_1}$ or $c_{\mathsf{app}}$. All such costs can be uniformly encoded using the cost
construct $\tick{c}{e}$.

Notably, the report~\cite{JMSCS:HJ22} does not include typing rules for the pair constructor
$\langle e_1, e_2 \rangle$ or the destructor
$\mathbf{letp}\; \langle x_1, x_2 \rangle = e_1 \;\mathbf{in}\; e_2$.
However, pairs can be straightforwardly encoded in AARA using the $\mathbf{cons}$ constructor.
Concretely,
\[
\langle e_1, e_2 \rangle
\;\equiv\;
\Let{x_1}{e_1}{
\Let{x_2}{e_2}{
\Let{x_3}{\mathbf{cons}(x_2,\mathbf{nil})}{
\mathbf{cons}(x_1,x_3)}}}.
\]
Similarly,
\[
\mathbf{letp}\; \langle x_1, x_2 \rangle = e_1 \;\mathbf{in}\; e_2
\;\equiv\;
\Let{x}{e_1}{
\mathbf{case}\; x \{
\mathbf{nil} \mapsto \cdots \;\text{(unused branch)}
\mid \mathbf{cons}(x_1, x_2) \mapsto e_2
\}}.
\]
Therefore, we omit further discussion of pair translation in what follows.

We use the following encodings of standard types:
\[
\begin{aligned}
\mathsf{Unit} &= \mathsf{ind}(\{\}), \\
\mathsf{List}(T) &= \mathsf{ind}(\{\mathbf{nil}(\mathsf{Unit},0), \mathbf{cons}(T,1)\}), \\
T_1 + T_2 &= \mathsf{ind}(\{\mathbf{left}(T_1,0), \mathbf{right}(T_2,0)\}).
\end{aligned}
\]

Our translation function $h$ is defined as follows:
\[
\begin{array}{rcl}
h(x) &=& \tick{c_{\mathsf{var}}}{x} \\[0.2em]

h(\langle \rangle) &=& \tick{c_{\mathsf{Unit}}}{\mathsf{Unit}} \\[0.2em]

h(\Let{x}{e_1}{e_2})
&=&
\tick{c_{\mathsf{Let}_3}}{
\Let{x}{\tick{c_{\mathsf{Let}_1}}{h(e_1)}}{
\tick{c_{\mathsf{Let}_2}}{h(e_2)}}} \\[0.4em]

h(\mathbf{left}(e)) &=& \tick{c_{\mathsf{left}}}{\mathbf{left}(h(e))} \\[0.2em]

h(\mathbf{right}(e)) &=& \tick{c_{\mathsf{right}}}{\mathbf{right}(h(e))} \\[0.2em]

h(\mathbf{case}\; x \{\mathbf{left}(x_1)\mapsto e_1 \mid \mathbf{right}(x_2)\mapsto e_2\})
&=&
\mathsf{matd}\bigl(x,\{
\mathbf{left}(x_1).\tick{c_{\mathsf{CaseLeft}}}{h(e_1)},\\&&
\mathbf{right}(x_2).\tick{c_{\mathsf{CaseRight}}}{h(e_2)}
\}\bigr) \\[0.4em]

h(\mathbf{nil}) &=& \tick{c_{\mathsf{nil}}}{\mathsf{Nil}()} \\[0.2em]

h(\mathbf{cons}(x_1,x_2))
&=& \tick{c_{\mathsf{cons}}}{\mathsf{Cons}(x_1,x_2)} \\[0.2em]

h(\mathbf{case}\; x \{\mathbf{nil}\mapsto e_0 \mid \mathbf{cons}(x_1,x_2)\mapsto e_1\})
&=&
\mathsf{matd}\bigl(x,\{
\mathbf{nil}.\tick{c_{\mathsf{CaseNil}}}{h(e_0)},\\&&
\mathbf{cons}(x_1,x_2).\tick{c_{\mathsf{CaseCons}}}{h(e_1)}
\}\bigr) \\[0.4em]

h(x_1\,x_2) &=& \tick{c_{\mathsf{app}}}{\app{x_1}{x_2}} \\[0.2em]

h(\mathbf{fun}\; f\,x = e)
&=&
\tick{c_{\mathsf{fun}}}{(\mathbf{fix}\; f.\lambda x.\,h(e))} \\[0.2em]

h(\tick{q}{})
&=&
\tick{q}{\mathsf{Unit}} \\[0.2em]

h(\mathbf{share}\; x\; \mathbf{as}\; x_1,x_2\; \mathbf{in}\; e)
&=&
h(e[x_1\mapsto x,\; x_2\mapsto x])
\end{array}
\]

Most of the translation are straightforwards to show as the direction translation, except for the $\textbf{share } x \textbf{ as } x_1, x_2 \textbf{ in } e$. Since in our type system, the type no longer carries potentials, thus we don't need to split the type.

Now we do some definitions that will be used in our proof.

\begin{definition}[Potentials and Types]
We first define the potentials $\Phi_x(A)$ of a type $A$ so that it returns the potential function of a program variable $x$ of the type $A$.
\[
\begin{array}{llll}
\Phi_x(\mathsf{1})& =  Unit \\[0.2em]
\Phi_x(L^p(A))& =  fix~f~matd(x,\{nil(x).0,cons(x_1,x_2).(\Phi_{x_1}(x_1)+p+f~x_2)\} \\[0.2em]
\Phi_x(A^p+B^r)& =  fix~f~matd(x,\{left(x).(\Phi_{x}(A)+p),right(x).(\Phi_{x}(B)+r)\} \\[0.2em]
\Phi_x(A\times B)& =  \Phi_{x}(A)[x\mapsto \fst x]+ \Phi_x(B)[x\mapsto \snd x] \\[0.2em]
\Phi_x(A\xrightarrow{p/q}B)& =  0 \\[0.2em]
\end{array}
\]

We also define types $Type(\cdot)$ of a AARA types as follows:
\[
\begin{array}{llll}
Type(\mathsf{1})& =  Unit \\[0.2em]
Type(L^p(A))& =  List(Type(A)) \\[0.2em]
Type(A^p+B^r)& =  Type(A)+Type(B) \\[0.2em]
Type(A\times B)& =  Type(A)\times Type(B) \\[0.2em]
Type(A\xrightarrow{p/q}B)& =  \arrow{x}{p+\Phi_x(A)}{Type(A)}{q+\Phi_y(B)}{Type(B)}{y} \\[0.2em]
\end{array}
\]
Thus the potential function of a context $\Gamma$ in AARA system is defined as follows:
\[
\begin{array}{llll}
\Phi(\mctx{x}{A}{\Gamma})& =  \Phi_x(A)+\Phi(\Gamma) \\[0.2em]
\Phi(\cdot)&=0 \\[0.2em]
\end{array}
\]
We can also define the type context translation as follows:
\[
\begin{array}{llll}
Pure(\mctx{x}{A}{\Gamma})& =  \mctx{x}{Type(A)}{Pure(\Gamma)}\\[0.2em]
\Phi(\cdot)&=\cdot \\[0.2em]
\end{array}
\]
\end{definition}

Next, we prove the embedding. Before presenting the main proof, we first establish two useful lemmas.

\begin{definition}[Context Substitution]
For an AARA typing context $\Gamma$, we define context substitution $[x \mapsto y](\Gamma)$ as follows:
\begin{itemize}
  \item If $\Gamma = \Gamma_1,~x\!:\!A_1,~\Gamma_2,~y\!:\!A_2,~\Gamma_3$ and
  $A_3 \curlyveedownarrow (A_1,A_2)$, then
  \[
    [x \mapsto y](\Gamma) = \Gamma_1,~\Gamma_2,~y\!:\!A_3,~\Gamma_3.
  \]

  \item If $\Gamma = \Gamma_1,~y\!:\!A_1,~\Gamma_2,~x\!:\!A_2,~\Gamma_3$ and
  $A_3 \curlyveedownarrow (A_1,A_2)$, then
  \[
    [x \mapsto y](\Gamma) = \Gamma_1,~y\!:\!A_3,~\Gamma_2,~\Gamma_3.
  \]

  \item If $\Gamma = \Gamma_1,~x\!:\!A_1,~\Gamma_2$ and $y \notin Var(\Gamma)$, then
  \[
    [x \mapsto y](\Gamma) = \Gamma_1,~y\!:\!A_1,~\Gamma_2.
  \]

  \item If $x \notin Var(\Gamma)$, then
  \[
    [x \mapsto y](\Gamma) = \Gamma.
  \]

  \item Otherwise, $[x \mapsto y](\Gamma)$ is undefined.
\end{itemize}

Note that for any $A_1$ and $A_2$ such that $Type(A_1) = Type(A_2)$, a type
$A_3$ satisfying $A_3 \curlyveedownarrow (A_1,A_2)$ always exists.
\end{definition}

We now state a lemma that holds in the AARA system and a lemma that holds in $\lang{}$ that will be used in the embedding proof.

\begin{theorem}[AARA Substitution]\label{thm:aara-substitution}
If $\Gamma \vdashpq{p}{q} e : A$, then for all variables $x$ and $y$,
if either
(i) both $x$ and $y$ appear in $\Gamma$ with $x\!:\!A_1$ and $y\!:\!A_2$ and $Type(A_1) = Type(A_2)$,
or
(ii) at most one of $x$ and $y$ appears in $\Gamma$,
then
\[
([x \mapsto y]\Gamma) \vdashpq{p}{q} ([x \mapsto y]e) : A.
\]
\end{theorem}

\begin{proof}
The proof proceeds by direct induction on the AARA typing derivation and is straightforward.
\end{proof}

\begin{theorem}[Weakening]\label{thm:na-var-weakening}
If $\tjudge{\Gamma}{\cdot}{f_1}{e}{f_2}{A}{x}$ and $y$ is a fresh variable with respect to $\Gamma$,
then
\[
\tjudge{\mctx{\Gamma}{y}{B}}{\cdot}{f_1}{e}{f_2}{A}{x}.
\]
\end{theorem}

\begin{proof}
The proof proceeds by direct induction on the typing derivation of $\lang{}$ and is straightforward.
\end{proof}

We now proceed to the main embedding proof.

\begin{theorem}[AARA embedding]
For $\Gamma \vdashpq{p}{q} e:A$, there exists a potential function
$\Phi_x'(A)$ such that
$\pjudge{Pure(\Gamma)}{\cdot}{q+\Phi_x(A)}{\le}{\Phi_x'(A)}$ and
\[
\tjudge{Pure(\Gamma)}{\cdot}{\Phi(\Gamma)+p}{h(e)}{\Phi_x'(A)}{Type(A)}{x}.
\]
\end{theorem}

\begin{proof}
We prove the theorem by induction on the AARA typing derivation
$\Gamma \vdashpq{p}{q} e:A$.

\begin{itemize}
  \item \textbf{Case L:Var.}
  The AARA typing rule is:
  \[
  \inferrule[(L:Var)]
  {q\ge q'+c_{Var}}{x:A \vdashpq{q}{q'} x:A}
  \]
  By translation, we must show that there exists $\Phi_x'(A)$ such that
  $\pjudge{x:Type(A)}{\cdot}{q+\Phi_x(A)}{\le}{\Phi_x'(A)}$ and
  \[
  \tjudge{x:Type(A)}{\cdot}{\Phi_x(A)+q}{\tick{c_{Var}}{x}}{\Phi_x'(A)}{Type(A)}{x}.
  \]
  Using rules \textsc{TTick}, \textsc{TVar}, and \textsc{TRelax}, we derive:
  \[
  \tjudge{x:Type(A)}{\cdot}{\Phi_x(A)+q}{\tick{c_{Var}}{x}}{\Phi_x(A)+q-c_{Var}}{Type(A)}{x}.
  \]
  Taking $\Phi_x'(A)=\Phi_x(A)+q-c_{Var}$ completes this case.

  \item \textbf{Case L:Unit, L:Let, L:Left, L:Right, L:MatchSum, L:Nil, L:Cons, L:MatchList.}
  These cases follow analogously to Case L:Var.

  \item \textbf{Case L:App.}
  \[
  \inferrule[(L:App)]
  {q\ge p+c_{App}\qquad q-q'\ge p-p'+c_{App}
  }{x_1:A\xrightarrow{p/p'} B, x_2:A \vdashpq{q}{q'} x_1~x_2:B}
  \]
  By translation, we must show
  \[
\begin{array}{ll}
{\mctx{x_1:\arrow{x}{p+\Phi_x(A)}{Type(A)}{q+\Phi_y(B)}{Type(B)}{y}}{x_2}{Type(A)}}~|~
  {\cdot}~|~{q+\Phi_{x_2}(A)}\vdash\\{\tick{c_{app}}{(x_1~x_2)}}:[{q'+\Phi_{y}(B)}]_{y}{Type(B)}.
\end{array}
  \]

  By rule \textsc{TApp}, we obtain
  \[
\begin{array}{ll}
  \tjudge{\mctx{x_1:\arrow{x}{p+\Phi_x(A)}{Type(A)}{q+\Phi_y(B)}{Type(B)}{y}}{x_2}{Type(A)}}
  {\cdot}{q+\Phi_{x_2}(A)}{\\\tick{c_{app}}{(x_1~x_2)}}{q-p-c_{App}+\Phi_y(B)}{Type(B)}{y}.
\end{array}
  \]
  Since $q\ge p+c_{App}$, this judgment is well-typed.
  Moreover, from $q-q'\ge p-p'+c_{App}$, taking
  $\Phi_y'(B)=q-p-c_{App}+\Phi_y(B)$ completes the proof of this case.

  \item \textbf{Case L:fun.}
  \[
  \inferrule[(L:fun)]
  {\Gamma \curlyveedownarrow(\Gamma, \Gamma)\qquad
   \Gamma,f:A\xrightarrow{p/p'} B,x:A \vdashpq{p}{p'}e:B\qquad
   q\ge q'+c_{fun}
  }{\Gamma \vdashpq{q}{q'}\textbf{fun } f\,x = e:B}
  \]
  Since $\Gamma \curlyveedownarrow(\Gamma,\Gamma)$, we have $\Phi(\Gamma)=0$.
  By translation, it suffices to show
  \[
\begin{array}{ll}
  {\mctx{Pure(\Gamma)}{f}{\arrow{x}{p+\Phi_x(A)}{Type(A)}{q+\Phi_y(B)}{Type(B)}{y}},x:Type(A)}~|~
  {\cdot}~|~{q+\Phi_{x_2}(A)}\vdash\\{\tick{c_{fun}}{(fix~f.\lambda x.h(e))}}:[{q'+\Phi_{y}(B)}]_{y}{Type(B)}.
\end{array}
  \]

  By the induction hypothesis, there exists $\Phi_y'(B)$ such that
  \[
  \pjudge{Pure(\Gamma,f:A\xrightarrow{p/p'} B,x:A)}{\cdot}{p'+\Phi_y(B)}{\le}{\Phi_y'(B)}
  \]
  and
  \[
  \tjudge{Pure(\Gamma,f:A\xrightarrow{p/p'} B,x:A)}{\cdot}{p+\Phi_x(A)}{e}{\Phi_y'(B)}{Type(B)}{y}.
  \]
  Consequently,
  \[
\begin{array}{ll}
  \tjudge{\mctx{Pure(\Gamma)}{f}{\arrow{x}{p+\Phi_x(A)}{Type(A)}{q+\Phi_y(B)}{Type(B)}{y}},x:Type(A)}
  {\cdot}{c_{fun}+\Phi_{x_2}(A)}{\\\tick{c_{fun}}{(fix~f.\lambda x.h(e))}}{\Phi_y(B)}{Type(B)}{y}.
\end{array}
  \]
  Applying rule \textsc{TRelax} yields the desired judgment.

  \item \textbf{Case L:Relax.}
  \[
  \inferrule[(L:Relax)]
  {\Gamma \vdashpq{p}{p'} e:A\qquad q\ge p\qquad q-q'\le p-p'}
  {\Gamma \vdashpq{q}{q'} e:A}
  \]
  By induction, there exists $\Phi_x'(A)$ such that
  \[
  \pjudge{Pure(\Gamma)}{\cdot}{p'+\Phi_x(A)}{\le}{\Phi_x'(A)}
  \]
  and
  \[
  \tjudge{Pure(\Gamma)}{\cdot}{\Phi(\Gamma)+p}{h(e)}{\Phi_x'(A)}{Type(A)}{x}.
  \]
  Let $\Phi_x''(A)=\Phi_x'(A)+q-p$. Applying rule \textsc{TRelax} yields
  \[
  \pjudge{Pure(\Gamma)}{\cdot}{q-p+p'+\Phi_x(A)}{\le}{\Phi_x''(A)}
  \]
  and
  \[
  \tjudge{Pure(\Gamma)}{\cdot}{\Phi(\Gamma)+q}{h(e)}{\Phi_x''(A)}{Type(A)}{x}.
  \]
  This completes the case.

  \item \textbf{Case L:Weak.}
  This case follows directly from Theorem~\ref{thm:na-var-weakening}, and then reduces to Case L:Relax.

  \item \textbf{Case L:share.}
  \[
  \inferrule[(L:share)]
  {A \curlyveedownarrow(A_1, A_2)\qquad \Gamma,~x_1:A_1,~x_2:A_2\vdashpq{q}{q'} e:B }
  {\Gamma,x:A \vdashpq{q}{q'} e:B}
  \]
  By Theorem~\ref{thm:aara-substitution}, we have
  \[
  \Gamma,x:A \vdashpq{q}{q'} [x_1\mapsto x,x_2\mapsto x]e:B.
  \]
  By induction (note that $\Gamma,x:A$ contains strictly fewer variables than
  $\Gamma,x_1:A_1,x_2:A_2$), there exists $\Phi_x'(A)$ such that
  \[
  \pjudge{Pure(\Gamma)}{\cdot}{q'+\Phi_x(A)}{\le}{\Phi_x'(A)}
  \]
  and
  \[
  \tjudge{Pure(\Gamma)}{\cdot}{\Phi(\Gamma)+q}{h(e)}{\Phi_x'(A)}{Type(A)}{x}.
  \]
  This corresponds exactly to the translation
  \[
  h(\textbf{share } x \textbf{ as } x_1, x_2 \textbf{ in } e)
  =
  h(e[x_1\mapsto x,x_2\mapsto x]),
  \]
  and therefore the case is proved.
\end{itemize}
\end{proof}

\subsection{Untypability of AARA}

In this section, we formally show that AARA is unable to type the \emph{map-append} example as well as the \emph{append} example introduced earlier.

To this end, it suffices to show that the term $append~e_1~e_2$ is not typable in AARA, and that the intermediate closure $append~e_1$ already leads to a problematic typing. In AARA, the closure $append~e_1$ must have a type of the form
\[
List^{p_2}(\mathsf{int}) \xrightarrow{p_3/q_3} List^{q_2}(\mathsf{int}),
\]
and therefore the only admissible typing judgment for $append~e_1$ is
\[
\cdot \vdashpq{p_1}{q_1} append~e_1 :
List^{p_2}(\mathsf{int}) \xrightarrow{p_3/q_3} List^{q_2}(\mathsf{int}).
\]

Consequently, for the application $append~e_1~e_2$, the function consumes
$p_2 \cdot \mathsf{length}(e_2)$ units of resource and produces
$q_2 \cdot \mathsf{length}(e_1 + e_2)$ units of resource.
Since $p_1$, $q_1$, $p_2$, $q_2$, $p_3$, and $q_3$ are all constants that do not depend on $e_1$,
we may choose $e_1$ such that
\[
\mathsf{length}(e_1) > \frac{p_3 + p_1}{q_2}.
\]
Next, let $e_2$ be the empty list $\mathsf{nil}$. In this case, the output potential
$q_2 \cdot \mathsf{length}(e_1 + \mathsf{nil}) = q_2 \cdot \mathsf{length}(e_1)$
is strictly greater than the total input potential available.

This contradicts the soundness of AARA, which guarantees that all typable terms can be evaluated safely without exceeding available resources. Hence, the term $append~e_1~e_2$ cannot be typable in AARA.

Therefore, $append~e_1~e_2$—and consequently higher-order examples such as \emph{map-append}—are not typable in AARA.